\documentclass[reqno]{amsart}

\usepackage{a4wide}
\usepackage{amsmath,amssymb,amsfonts,amsthm}
\usepackage{moreverb,rotating,graphics}
\usepackage[english]{babel} 
\usepackage{framed,fancybox}
\usepackage{caption, subcaption}
\usepackage{mathrsfs, esint, dsfont}
\usepackage[colorlinks, linkcolor={blue}, citecolor={red}, urlcolor = {blue}]{hyperref}
\usepackage{xcolor}
\definecolor{light-gray}{gray}{0.70}
\definecolor{dark-gray}{gray}{0.40}
\definecolor{very-light-gray}{gray}{0.90}
\usepackage{tikz}
\usepackage{pgfplots}
\pgfplotsset{width=7cm,compat=1.18}\usepgfplotslibrary{fillbetween}

\usepackage{bbm}
\usepackage{float}
\usepackage{enumerate}
\usepackage{mathrsfs}  
\usepackage{textcomp}
\usepackage{color}
\usepackage{enumitem}
\usepackage{listings}
\usepackage{csquotes}

\newtheorem{theorem}{Theorem}[section]
\newtheorem{proposition}[theorem]{Proposition}
\newtheorem{remark}[theorem]{Remark}
\newtheorem{lemma}[theorem]{Lemma}

\newtheorem{definition}[theorem]{Definition}

\theoremstyle{definition}
\newtheorem{example}{Example}[section]

\newcommand{\bbI}{\mathbb{I}}

\newcommand{\ri}{\mathrm{i}}
\newcommand{\cS}{\mathscr{S}}
\newcommand{\cD}{\mathcal{D}}

\newcommand{\cC}{\mathcal{C}}
\newcommand{\cR}{\mathcal{R}}
\newcommand{\C}{\mathbb{C}}
\newcommand{\R}{\mathbb{R}}
\newcommand{\N}{\mathbb{N}}
\newcommand{\Z}{\mathbb{Z}}

\newcommand\ba {{ \mathbf a}}

\newcommand\bk{{\mathbf k}}
\newcommand\bK{{\mathbf K}}

\newcommand\bR{{\mathbf R}}

\newcommand\bx{{\mathbf x}}
\newcommand\by{{\mathbf y}}

\newcommand\bnull{{\mathbf 0}}

\def\cA{{\mathcal A}}

\def\cC{{\mathcal C}}
\def\cD{{\mathcal D}}

\def\cF{{\mathcal F}}

\def\cH{{\mathcal H}}

\def\cM{{\mathcal M}}
\def\cN{{\mathcal N}}

\def\cR{{\mathcal R}}
\def\cS{{\mathcal S}}

\def\sC{{\mathscr{C}}}

\def\rd{{\mathrm{d}}}
\def\re{{\mathrm{e}}}
\def\ri{{\mathrm{i}}}

\newcommand\1{{\ensuremath {\mathds 1} }} 

\def\ssH{{\mathsf{H}}}
\def\ssh{{\mathsf{h}}}

\def\sst{{\mathsf{t}}}

\def\bra{\langle}
\def\ket{\rangle}
\def\bulk{{\rm{bulk}}}
\def\edge{{\rm{edge}}}
\def\ess{{\rm{ess}}}

\def\cut{{\rm{cut}}}
\def\Sf{{\mathrm{Sf}}}
\def\Ker{{\mathrm{Ker}}}
\def\op{{\rm{op}}}
\def\Ran{\mathrm{Ran}}
\def\rank{\mathrm{rank}}
\def\Lat{\mathbb{L}}
\def\RLat{\mathbb{L}^*}
\def\dist{\operatorname{dist}}

\newcommand{\norm}[1]{\left\| #1 \right\|}

\newcommand{\myfootnote}[1]{
	\renewcommand{\thefootnote}{}
	\footnotetext{\scriptsize#1}
	\renewcommand{\thefootnote}{\arabic{footnote}}
}

\title[Edge states for periodic tight--binding operators with soft walls]{Edge states for tight--binding operators with soft walls}

\author{Camilo Gómez Araya \qquad David Gontier \qquad Hanne Van Den Bosch}

\date{\today}

\begin{document}

\myfootnote{Camilo Gómez Araya: CEREMADE, Université Paris-Dauphine, PSL University,75016 Paris, France;\\
email: \href{camilo.gomez-araya@dauphine.eu}{camilo.gomez-araya@dauphine.eu}}
\myfootnote{David Gontier: CEREMADE, Université Paris-Dauphine, PSL University,75016 Paris, France \& ENS/PSL University, DMA, F-75005, Paris, France;\\
email: \href{gontier@ceremade.dauphine.fr}{gontier@ceremade.dauphine.fr}}
\myfootnote{Hanne Van Den Bosch: Departamento de Ingenier\'{\i}a Matem\'atica \& Center for Mathematical Modeling, Facultad de Ciencias F\'{\i}sicas y Matem\'aticas, Universidad de Chile and CNRS IRL 2807, Beauchef 851, Piso 5, Santiago, Chile;\\
email: \href{hvdbosch@dim.uchile.cl}{hvdbosch@dim.uchile.cl}}

\begin{abstract}
		We study one-- and two--dimensional periodic tight-binding models in the presence of a potential that grows to infinity in one direction, hence preventing the particles to escape in this direction (the soft wall). We prove that a spectral flow appears in these edge models, as the wall is shifted with respect to the lattice. We identity this flow with the number of Bloch bands. This provides a lower bound for the number of edge states appearing in such models. For the two-dimensional case, we compute the spectral flow for edges that have any rational orientation with respect to the lattice. The results are illustrated by applying them to the one-dimensional SSH chain and the Wallace model for graphene. 
		
		\bigskip
		\noindent \sl \copyright~2025 by the authors. This paper may be reproduced, in its entirety, for non-commercial purposes.
\end{abstract}


\maketitle    

\tableofcontents


\section{Introduction}

The goal of this paper is to study \emph{edge states} in terminated periodic structures described by tight--binding (TB) Hamiltonians. Such TB operators are extensively studied in condensed matter, as they provide simple models which correctly reproduce the physics of more complex ones (represented {\em e.g.} by Schrödinger operators acting on the continuum). This research is motivated by models from solid state physics such as the one-dimensional Su-Schrieffer-Heeger (SSH) chain \cite{SuSchHee-79} and the two-dimensional Wallace model for graphene \cite{Wal-47}. Indeed, it is known for these models that edge states may appear when these periodic systems are restricted to a halfspace. For the Wallace model for instance, the presence of these edge states depends strongly on the direction of the cut. Notable examples are the {\em zigzag cut}, where a flat band of edge modes appear, and the {\em armchair cut}, where no edge states appear, see~\cite{AkhBee-08, NetGuiPer-09, DelUllMon-11, FefLeeWei-16, FefFliWei-22, FefFliWei-24}. 
Even in the simple one--dimensional SSH model,
an edge state appears or disappears depending on which of the nonequivalent \emph{bonds} is cut.
In these one-dimensional systems, the topological versus \emph{trivial} nature of experimentally observed edge states is still a topic of hot debate (see e.g. \cite{PanDas-20} for theoretical considerations and \cite{ConJarLeg-22} for recent experimental results) ever since the first reports of Majorana fermions in superconducting wires \cite{DasRonMos-12, MouZuoFro-12}.

\medskip

In the present article, we give a framework to study the appearance of edge modes in general TB models arising in a wide variety of physical situations, such as graphene sheets, superconducting wires or optical lattices. However, instead of working with {\em hard truncations} (Dirichlet boundary conditions), as it is usually done in numerical simulations, we are interested in edge modes that arise from the lattice termination by a \emph{soft wall}, that is, a continuous potential barrier that impedes the propagation in one halfspace. The main reason is that one cannot translate properly a {\em hard cut} in TB models. To illustrate the problem, imagine a TB problem on the line $\Z$, with the presence of a hard wall whose position is parametrized by $t \in \R$, with the convention that one \emph{deletes} all sites on the left of $t$. Then, when $t$ varies in an interval of the form $(n-1, n)$, nothing changes, but when $t$ crosses the site $n$, the corresponding Hilbert space suddenly switches from $\ell^2(\{n, n+1, \cdots\})$ to $\ell^2(\{n+1, n+2,  \cdots\})$. The corresponding operators are not continuous in $t$, which makes this case difficult to study. Note that, for some experimental realizations of periodic systems, confinement by a smooth potential instead of a sharp lattice termination is easily achieved. Numerical results on the effect of soft confinement in different configurations can be found for instance in \cite{BucCocHof-12, GalLeeBar-17, GebIrsHof-20}.

We prove the appearance of a spectral flow in the gaps of the essential spectrum as the soft wall is shifted, and relate this flow with a number of Bloch bands (see Theorem~\ref{th:main_general_1d}). In particular, our result implies the presence of edge states for \emph{many} values of $t$. Surprisingly, this spectral flow is independent of the shape of the wall. However, in the case where the soft wall varies slowly with respect to the lattice constant, we can prove that edge states always appear, for {\em all} values of $t$, regardless of the topological properties of the material (see Theorem~\ref{th:main_fix_t0} below). 

\medskip

When studying the analogue of this problem for models set on the continuum, that is for Schrödinger operators with periodic potentials in $\R$ or $\R^2$, one can use the spectral flow to study the effect of hard truncations, modelled by a Dirichlet boundary condition along a line parametrized by a translation parameter $t$ (the cut). These models have been studied in~\cite{HemKoh-11, HemKoh-11a, HemKoh-12, HemKohSta-15, Dro-21, Gon-20, Gon-21, Gon-23}, where it is proved that a spectral flow appears as the cut is translated. In addition, this spectral flow can be computed explicitly, and is related to a number of Bloch bands of the initial bulk operator. The results of the present paper are the discrete analogue of these, although the techniques of proof are rather different. 

\medskip

In the two--dimensional case, we also prove that numerous edge modes must appear when one cuts such materials with a commensurate angle having large numerators/denominators. Actually, following the lines of~\cite{Gon-21}, one could prove that {\em all bulk gaps are filled with edge spectrum} in the incommensurate case (although we do not provide a full proof here, as it is similar to the one in~\cite{Gon-21}).

\medskip

Let us give a short description of our main result in the one--dimensional setting. We consider a general TB bulk operator $H$ acting on $\ell^2(\Z, \C^N)$, of the convolution form
$$
   \forall \Psi \in \ell^2(\Z, \C^N), \quad \forall n \in \Z, \qquad (H \Psi)_n = \sum_{m \in \Z} h(m) \Psi_{n-m} =: (h*\Psi)_n
$$
where for each $n \in \Z$, $h(n)$ is an  $N \times N$ matrix such that $h^*(n) = h(-n)$. We assume throughout that $\sum_n \norm{h_n}_{\rm op} < \infty$, which guarantees the boundedness of the operator by Young's inequality (see Lemma~\ref{lem:bulk_self_adjoint} below). Strictly speaking, the model is set on $a_0 \Z$ with $a_0 > 0$ the length of the unit cell. In what follows, we take $a_0 = 1$.  By taking a Fourier transform, the spectrum of $H$ is purely essential, of the form
\[
    \sigma_\bulk := \sigma(H) = \bigcup_{k \in \R} \sigma(H_k), \quad \text{with} \quad H_k := \sum_{m \in \Z} h(m) \re^{- \ri k m} .
\]
The map $k \mapsto H_k$ is $2 \pi$--periodic and analytic, and takes values in $\cS_N$, the set of Hermitian $N \times N$ matrices. Let $\lambda_{1,k} \le \cdots \le \lambda_{N,k}$ be the eigenvalues of $H_k$ ranked in increasing order. By standard perturbation theory~\cite{Kat-95}, for all $1 \le j \le N$, the map $k \mapsto \lambda_{j,k}$ is continuous and $2 \pi$--periodic. The $j$-th {\em Bloch band} is the image of this map, namely the interval $\bigcup_k \{ \lambda_{j,k} \}$. The connected components of $\R \setminus \sigma_\bulk$ are the  \emph{gaps} in the essential spectrum. For an energy $E \in \R \setminus \sigma_\bulk$, we denote by $\cN(E)$ the number of Bloch bands below $E$. It is the unique integer such that 
\[
\forall k \in [-\pi, \pi], \qquad \lambda_{\cN(E),k} < E  < \lambda_{\cN(E)+1, k},
\]
with the convention that $\lambda_{0, k} = - \infty$ and $\lambda_{N+1, k}= + \infty$. Occasionally, when comparing several bulk operators, we will also write $\cN_H(E)$.

\medskip

We now perturb the bulk Hamiltonian by adding the wall potential. We fix a Lipschitz continuous function $w : \R \to \cS_N$ satisfying
\begin{equation}\label{eq:wall_def_intro}
     \forall v \in \C^N\setminus\{0\}, \qquad \lim_{x \to -\infty} \langle v,w(x) v \rangle = +\infty  \quad \text{and} \quad     \lim_{x \to +\infty} \langle v, w(x) v \rangle = 0.
\end{equation}
The first limit states that the lowest eigenvalue of $w(x)$ diverges to $+\infty$ as $x \to - \infty$, and the second that all entries of the matrix $w$ converge to $0$ as $x \to + \infty$. 
From $w$, we define the wall operator $W(t)$ acting on $\ell^2(\Z, \C^N)$ as the multiplication operator
\begin{equation} \label{eq:def:W}
    \forall \Psi \in \ell^2(\Z, \C^N), \qquad (W(t) \Psi)_n := w(n-t) \Psi_n.
\end{equation}
The parameter $t \in \R$ indicates a shift in the position of the wall.
Our goal is to study the family of {\em edge} operators $t \mapsto H^\sharp(t)$ defined by
\[
    H^\sharp(t) := H + W(t), \quad \text{specifically} \quad 
    \left( H^\sharp(t) \Psi \right)_n = (h*\Psi)_n + w(n-t) \Psi_n.
\]
It is not difficult to see that the family $t \mapsto H^\sharp (t)$ is translation equivariant, in the sense that $H^\sharp(t + 1) = \tau_1^* H^\sharp(t) \tau_1$, where $\tau_1$ is the usual translation operator $(\tau_1 \Psi)_n := \Psi_{n-1}$. We prove in Proposition~\ref{prop:basics_intro} below that for all $t \in \R$, $H^\sharp(t)$ is self-adjoint with a domain independent of $t$, and that the essential spectrum of $H^\sharp(t)$ is also independent of $t$. However, some extra eigenvalues may appear in the essential gaps as $t$ moves. For $E \in \R \setminus \sigma_\bulk$, we denote by 
\[
    \Sf(H^\sharp(t), E, [0, 1])
\]
the {\bf spectral flow} of $t \mapsto H^\sharp(t)$ at energy $E$ as $t$ increases from $0$ to $1$. It counts the net number of eigenvalues of $t \mapsto H^\sharp(t)$ crossing the energy $E$ downwards (see Appendix~\ref{sec:appendix:SF} for a precise definition, following~\cite{AtiPatSin-76, Phi-96, DolSchWat-23}). Let us state our main result.

\begin{theorem} \label{th:main_general_1d}
    Let $w : \R \mapsto \cS_N$ be a Lipschitz function satisfying~\eqref{eq:wall_def_intro}.
   For all $E \in \R \setminus \sigma_\bulk$, the energy $E$ is in the essential gaps of all operators $H^\sharp(t)$, and
        \[
        \Sf( H^\sharp(t) , E, [0, 1]) = -\cN(E),
        \]
        where we recall that $\cN(E)$ is the number of Bloch bands below $E$ for the bulk--operator.
\end{theorem}

In particular, the spectral flow is independent of the precise expression for the soft wall. A nonzero spectral flow implies that, at least for some values of $t \in \R$, the set $\sigma_\edge(t) := \sigma(H^\sharp(t)) \setminus \sigma_\bulk$, sometimes called the {\em edge spectrum}, is non empty. It only consists of eigenvalues. The corresponding eigenvectors are called {\em edge} states, and describe localized modes which are blocked by the wall on the left, and which cannot propagate in the medium on the right.

\medskip

Let us briefly explain the idea of the proof with the help of the graphical representation in Figure~\ref{fig:ideas}. First, we show that the spectral flow for any soft wall is equal to the spectral flow for a simplified \emph{steep wall} model that essentially behaves as a step potential with value $\Sigma \gg E$ (Figure \ref{fig:ideas}.1). Next, we show that the spectral flow decouples in a left and right contribution by summing a compact perturbation that erases some bonds (Figure \ref{fig:ideas}.2). Here, the right-hand side of the picture does not contribute to the spectral flow at energy $E$. 
A similar procedure can be performed on the \textit{dislocated model} shown in Figure \ref{fig:ideas}.3. In this model, as $t$ increases by one, one site decouples from the otherwise periodic chain. A suitable compact perturbation yields a cut dislocated model shown in Figure \ref{fig:ideas}.4. 
As before, the left-hand-side does not contribute to the spectral flow while the right-hand-sides coincides with that of the steep wall model. 
This allows us to show that the spectral flow for any soft wall model equals the spectral flow of the related dislocated model. 
Then, the spectral flow for the dislocated model can be computed exactly, by interpreting it as the limit of finite chains with periodic boundary conditions.

\begin{figure}
\includegraphics[]{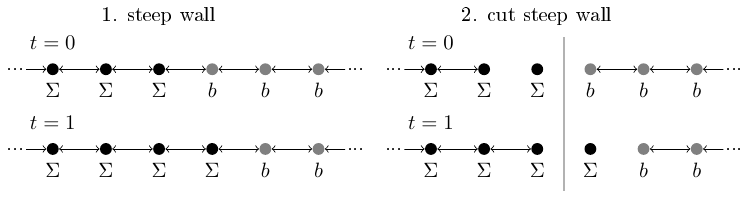}

\includegraphics[]{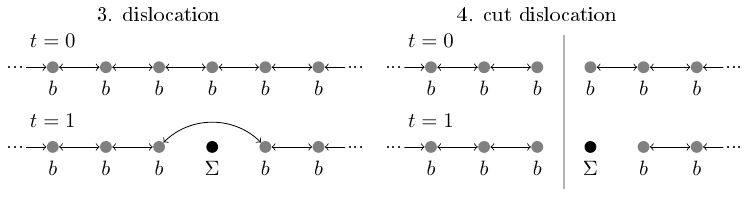}
\caption{Schematic illustration of the models in the proof. In the first figure, as $t$ varies from $0$ to $1$, the wall is pushed one cell to the right. We prove that all spectral flows coincide for these four scenarios, and explicitly compute the third one.} \label{fig:ideas}
\end{figure}

\begin{remark} \label{rem:HilbertHotel}
    The result is the discrete analogue to the one in~\cite{Gon-23}, up to the minus sign. This is due to the fact that in the present article, the wall is moving to the right when $t$ increases, which is the direction opposite to the one in~\cite{Gon-23}. In~\cite{Gon-21, Gon-23} (see also~\cite{Gon-24}), one of us gave an interpretation of this result, that we called the {\bf Grand Hilbert Hotel}, that might be useful. It is similar to the {\em charge pumping} denomination of Thouless~\cite{Tho-83}. Basically, we may think of each Bloch band as a hotel floor with an infinite number of rooms, each room occupied by one guest (or fermion).
 As $t$ goes from $0$ to $1$, the wall is pushed on the right, deleting the leftmost rooms. A particle from the lowest band moves up to the second band, giving a spectral flow of $-1$ in the gap above the first band. Two particles cross the gap above the second band: one from the deleted leftmost room, and one to accommodate the new arrival from the first band, and so on.
\end{remark}

In practice, one is often interested in the spectrum of $H^\sharp(t = t_0)$, that is for a single value of $t$. We have the following (see Section~\ref{sec:proof_main_fix_t0} for the proof).

\begin{theorem} \label{th:main_fix_t0}
    Let $w : \R \mapsto \cS_N$ be a soft wall which is $\nu$--Lipschitz, for some $\nu > 0$. If $E$ belongs to an essential gap of $H$, then for all $t_0 \in \R$ the operator $H^\sharp(t_0)$ has at least $\cN(E)$ eigenvalues in each interval of the form $(\lambda, \lambda+\nu]$ in this gap.
\end{theorem}

\begin{remark} \label{rem:Lipschitz}
 Recall that we took the length scale of the chain to be $a_0 = 1$, so $\lambda$ and $\nu$ have the same units. To restore a length scale and work instead on $a_0 \Z$, it is convenient to define the operator $W$ in~\eqref{eq:def:W} by $(W(t) \Psi)_R := w(R - ta_0) \Psi_R$ for $R \in a_0 \Z$, so that $t \mapsto H(t)$ is translation equivariant with period $1$. If $w(\cdot)$ is $\nu$-Lipschitz, then $W(\cdot)$ and $H(\cdot)$ are $\nu a_0$--Lipschitz, and there are $\cN(E)$ eigenvalues in each interval of the form $(\lambda, \lambda + \nu a_0]$.
\end{remark}

This theorem states that the density of eigenvalues in the gap containing $E$ is at least $\cN(E)/\nu$. If $\nu$ is small, the potential $W$ increases slowly to $+\infty$ on the left. It somehow acts as a local energy shift. It is therefore not so surprising that many eigenvalues appear in all gaps in this situation. On the other hand, if $\nu$ is very large, (imagine a situation where one approximates a hard cut by a sequence of soft wall models with $\nu \to \infty$) then the spectral flow does not give information on the spectrum of $H^\sharp(t_0)$: it may or may not contain extra eigenvalues. Figure~\ref{fig:SSH} below illustrates this situation for the SSH model.

\medskip

Before we go on, let us recall that there are several ways to derive TB models from the ones set on the continuum. One approach, called the {\em strong binding limit}, is given in~\cite{FefLeeWei-17, FefWei-20, ShaWei-22}, and consists in studying the limit of Schrödinger operators acting on $L^2(\R^d)$ with a periodic potential, when the \emph{wells} of the potential approach a delta-function centered on each atom. In this limit (similar to the semi-classical one), the wave-function $\Psi$ is replaced by a discrete function giving the amplitude of $\Psi$ associated to each atom. Another way to derive TB models is through the use of Wannier functions. In this case, one finds a countable set of functions $(w_i)$ which are well-localized around the atoms $(\bx_i)$, and which span the lowest part of the spectrum of the Hamiltonian. A wave--function $\Psi(x)$ is then approximated by the sum $\sum c_i w_i(x)$. Thanks to the localization of $w_i$, the value $c_i w_i(\bx_i)$ can be interpreted as the value $\Psi(\bx_i)$. The hopping transition between atoms $i$ and $j$ are then $t_{ij} = \bra w_i, H w_j \ket$. We refer to~\cite{SlaKos-54, GorBowHer-97} for more details.

\medskip

The paper is structured as follows. In Section~\ref{sec:1d_model}, we prove some basic properties of one-dimensional Hamiltonians with and without soft walls. We explain how to reduce the problem to the somewhat easier case of periodic {\em Jacobi} operators and prove Theorem~\ref{th:main_general_1d} assuming it holds for Jacobi operators. As an illustration of the general results we discuss the SSH chain in Section~\ref{sec:SSH} and present numerical simulations of its spectrum. In Section~\ref{sec:proof_main_fix_t0}, we give the proof of Theorem~\ref{th:main_fix_t0}, again assuming the version of Theorem~\ref{th:main_general_1d} for Jacobi operators.

The computation of the spectral flow for periodic Jacobi operators takes up the entire Section~\ref{sec:spectral_flows_Jacobi}. This is the main technical part of the paper.

In Section~\ref{sec:2d}, we explain how to extend our result in the two--dimensional setting. The main observation is that a 2d model, even cut by a (commensurate) soft wall, is still periodic in the direction along the wall. After a Fourier transform in this direction, we are left with a family of one-dimensional models, indexed by the momentum in the direction along the wall. We prove in Theorem~\ref{th:lot_of_edge_states_if_incommensurate} that the greater the incommensurability of the cut, the more abundant the edge states. In Section~\ref{sec:exemple_wallace}, we specify to the Wallace model for graphene to illustrate the concepts and give numerical results on the edge states.

Finally, we provide in Appendix~\ref{sec:appendix:SF} a precise definition of the Spectral Flow adapted to our setting, and establish some of its properties that are used in the body of the paper.

\section{Spectral flows in one-dimensional models}
\label{sec:1d_model}

In this section, we prove our results in the one--dimensional setting.

\subsection{Generalities for bulk and edge Hamiltonians}

\subsubsection{Bulk periodic Hamiltonians}
\label{ssec:periodic_bands}

We work in this section in the Hilbert space $\ell^2(\Z, \C^N)$, with operators that are periodic, i.e., commute with the shift operator $\tau_1$ defined by $(\tau_1 \Psi)_n := \Psi_{n-1}$. Such operators appear in TB models for a periodic chain of atoms, and $N$ represents the number of particles in one unit cell. The corresponding Hamiltonian can be written as a discrete convolution in the form
\begin{equation} \label{eq:convolutionForm}
    (H \Psi)_n = \sum_{m \in \Z} h(m) \Psi_{n-m} = (h*\Psi)_n.
\end{equation}

\begin{lemma} \label{lem:bulk_self_adjoint}
      If $h \in \ell^1(\Z, \C^{N\times N})$ and satisfies the symmetry condition $h(n)^* = h(-n)$, then $H$ is a bounded self-adjoint operator on $\ell^2(\Z, \C^N)$.
\end{lemma}
\begin{proof}
The symmetry condition $h(n)^* = h(-n)$ ensures that $H$ is a symmetric operator. In addition, the fact that $h$ is summable implies that $H$ is a bounded operator. Indeed, by the discrete Young inequality for convolutions, we have
\[
    \forall \Psi \in \ell^2(\Z, \C^N), \quad
        \|H \Psi \|_{\ell^2} = \| h*\psi \|_{\ell^2} \le \| h \|_{\ell^1} \| \psi \|_{\ell^2},
        \quad \text{with} \quad
    \| h \|_{\ell^1} := \sum_{n=-\infty}^\infty  \| h(n) \|_{\op}.
\]
So $H$ is a bounded symmetric operator, hence self-adjoint.
\end{proof}

\begin{remark} \label{rem:larger_blocks}
   Throughout this section, we will use the term \emph{periodic} to refer to 1-periodic models. If a Hamiltonian is periodic with period $K$ for some integer $K>1$, we can identify $\ell^2(\Z, \C^N)$ with $\ell^2(K \Z, \C^{K\,N})$ hence with $\ell^2(\Z, \C^{K\,N})$, and obtain a 1-periodic model with larger \emph{blocks}. This effectively changes the length--scale $a_0$ into $\widetilde{a}_0 := K a_0$. See Subsection~\ref{sec:SSH} for an example of this.
\end{remark}

We define the unitary Fourier transform by $\cF : \ell^2(\Z, \C^N) \to L^2([-\pi, \pi], \C^N)$ by
\begin{equation} \label{eq:def_fourier_1d}
        \left( \cF[\Psi] \right)(k) := \frac{1}{\sqrt{2 \pi}} \sum_{m \in \Z} \Psi_m \re^{ -\ri k m}, \quad \text{so that} \quad
    \Psi_n =\frac{1}{\sqrt{2 \pi}} \int_{-\pi}^\pi \cF[\Psi](k) \re^{ \ri k n} \rd k,
\end{equation}
and note that for any $u \in L^2([-\pi, \pi],\C^N)$,
\begin{equation} \label{eq:def:Hk}
\left(\cF H \cF^* u \right) (k) = H_k u(k ), \quad \text{with} \qquad
    H_k := \sum_{m \in \Z} h(m) \re^{- \ri k m}.
\end{equation}

Since $h(n) = h^*(-n)$, the matrix $H_k$ is Hermitian for all $k \in \R$. Let $\lambda_{1, k} \le \cdots \le \lambda_{N, k}$ be the eigenvalues of $H_k$ ranked in increasing order. Since the map $k \mapsto H_k$ is $2 \pi$--periodic and analytic, the maps $k \mapsto \lambda_{j,k}$ are $2 \pi$--periodic and continuous (we loose analyticity when eigenvalues intersect). In what follows, we set $\Omega^* := (- \pi, \pi]$. The spectral theorem then shows that
\begin{equation} \label{eq:spectrum_Hbulk}
    \sigma_\bulk := \sigma(H) = \bigcup_{k \in \Omega^*} \sigma \left( H_k \right) = 
    \bigcup_{j=1}^N \bigcup_{k \in \Omega^*}  \{ \lambda_{j,k} \}.
\end{equation}
The spectrum of $H$ consists of $N$ \emph{Bloch bands}, which are the images of $k \mapsto \lambda_{j,k}$. Each Bloch band is a closed interval by continuity and periodicity of this map. The intervals in the complement of this spectrum are called \emph{gaps}, or {\em bulk gaps}.


\subsubsection{Soft wall operators}
\label{sec:softWall}
Our goal is to study such periodic Hamiltonians, in the presence of a soft wall. Let us specify the class of walls that we use in this article.
\begin{definition}
    An operator $W$ acting on $\ell^2(\Z, \C^N)$ is a soft wall operator if it is of the form
    \[
        \forall n \in \Z, \qquad (W \Psi)_n = w(n) \Psi_n,
    \]
    where $w : \R \mapsto \cS_N$ is $\nu$--Lipschitz for some $\nu > 0$, and satisfies~\eqref{eq:wall_def_intro}, {\em i.e.} has the limits
   \[
     \forall v \in \C^N\setminus\{0\}, \qquad \lim_{x \to -\infty} \langle v,w(x) v \rangle = +\infty  \quad \text{and} \quad     \lim_{x \to +\infty} \langle v, w(x) v \rangle = 0.
    \]
\end{definition}
The first limit implies that for all $\Sigma \in \R$, there is $x_\Sigma \in \R$ so that 
\[
    \forall x \le x_\Sigma, \qquad w(x) \ge \Sigma.
\]
Since $w$ is Lipschitz, it is continuous. Also, since it goes from $+\infty$  to $0$, it is bounded from below. The shifted, or translated, soft wall operator is 
\[
    (W(t) \Psi)_n := w(n - t) \Psi_n.
\]
It is an unbounded self-adjoint operator, with domain
\[
    \cD (W(t)) := \left\{ \Psi \in \ell^2(\Z, \C^N), \quad W(t) \Psi \in \ell^2(\Z, \C^N) \right\}.
\]
Note that $\norm{W(t) -W(0)}_{\op} = \sup_{n \in \Z} \norm{w(n-t) - w(n)}_{\op} \le \nu |t|$, so this domain is independent of $t$ and we will denote it simply by $\cD(W)$. Let $\mu_{1, x} \le \cdots \le \mu_{N, x}$ be the eigenvalues of $w(x)$ for all $x \in \R$. Since $w$ is $\nu$--Lipschitz, all the curves $x \mapsto \mu_{1, x}$ are $\nu$--Lipschitz as well. In addition, all these curves go from $+\infty$ to $0$ on $\R$. As $W(t)$ is block diagonal with block elements $w(n - t)$, we directly deduce that
\begin{equation} \label{eq:spectrum_Wt}
    \sigma(W(t)) =  \bigcup_{n \in \Z} \sigma( w(n-t) ) = \bigcup_{j=1}^N  \bigcup_{n \in \Z} \{ \mu_{j, n - t} \}, 
    \quad \text{and that} \quad \sigma_{\rm ess}(W_t) = \{ 0 \}.
\end{equation}
The spectrum is purely discrete, composed of eigenvalues. They are all of finite multiplicities (except maybe $0$), and the only accumulation point is $0$. In addition, the map $t \mapsto W(t)$ is translation equivariant, in the sense that
\[
    W(t+1) = \tau_1^* W(t) \tau_1, \qquad \text{with} \quad (\tau_1 \Psi)_n := \Psi_{n-1}.
\]
In particular, the spectrum is $1$--periodic (which can be read directly from~\eqref{eq:spectrum_Wt}).  Figure~\ref{fig:spectrum_Wt} illustrates the situation.

\begin{figure}[ht]
    \centering
    \includegraphics[height=150pt]{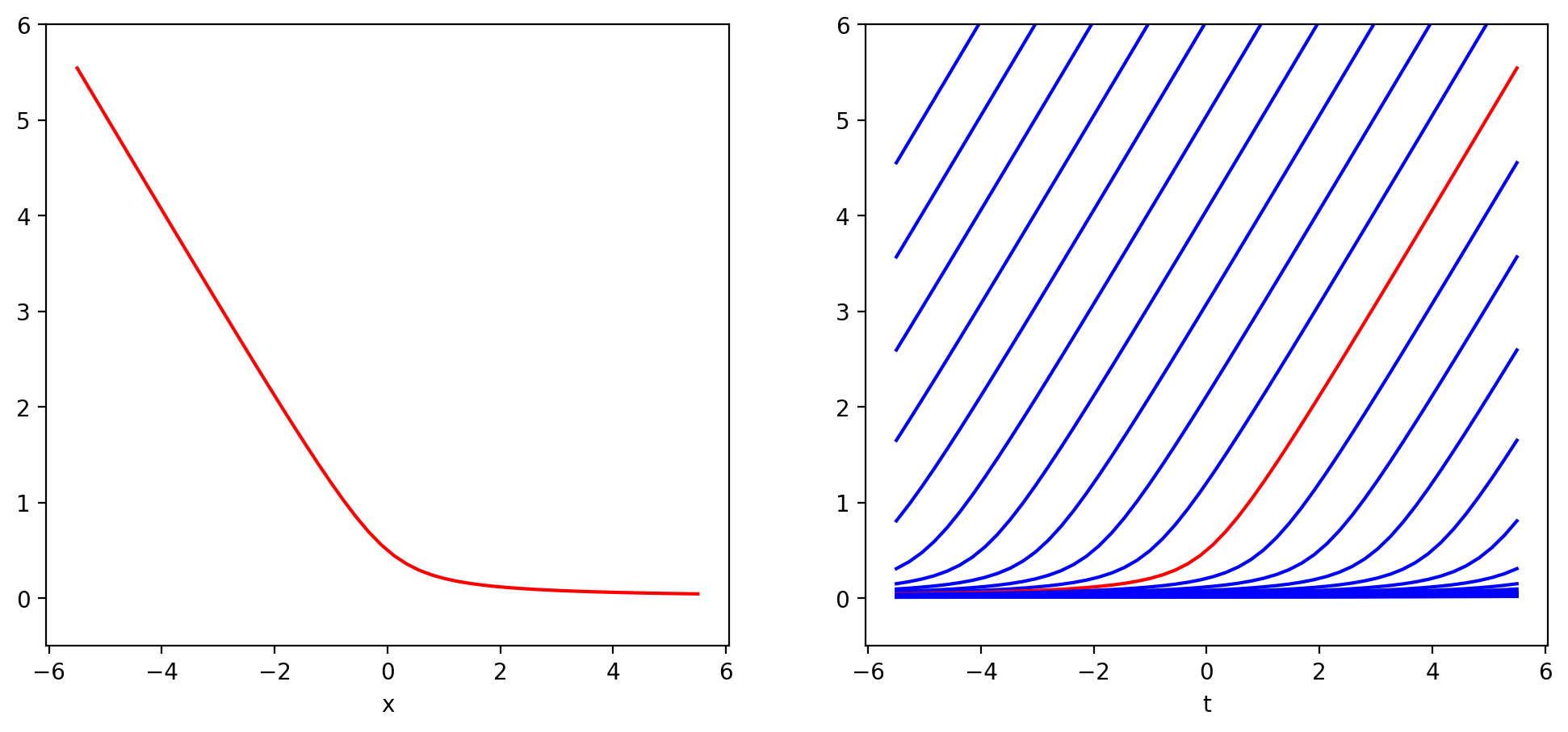}
    \caption{Left: the function $w(x) := \frac12 \left(\sqrt{x^2 + 1} - x \right)$, with $N=1$. Right: the spectrum of the corresponding operator $W(t)$, as a function of $t$ (the $n=0$ case is in red). As $t$ increases, the wall moves to the right. Here, $t \mapsto W(t)$ is operator increasing.}
    \label{fig:spectrum_Wt}
\end{figure}
 
We define the {\em edge} operator $H^\sharp(t) := H + W(t)$. Let us record some properties of this operator.

\begin{proposition} \label{prop:basics_intro}
    Let $w : \R \mapsto \cS_N$ be a Lipschitz function satisfying \eqref{eq:wall_def_intro}. 
    \begin{enumerate}
        \item For all $t \in \R$, the operator $H^\sharp(t)$ is self-adjoint on the constant domain $\cD(W)$.
        \item The map $t \mapsto H^\sharp(t)$ is norm--resolvent continuous and translation equivariant, in the sense that $H^\sharp(t + 1) = \tau_1^* H^\sharp(t) \tau_1$, where $\tau_1$ is the usual translation operator $(\tau_1 \Psi)_n := \Psi_{n-1}$;
        \item The essential spectrum $\sigma_\ess( H^\sharp(t) )$ is independent of $t$, and equals $\sigma_\bulk$.
       \end{enumerate}
\end{proposition}

\begin{proof}[Proof of Proposition~\ref{prop:basics_intro}]
    For the first point, the bulk operator $H$ is bounded and self-adjoint, hence is a bounded perturbation of the self-adjoint operator $W(t)$. Thus, $H^\sharp(t)$ is self-adjoint on $\cD(W)$. 
    
    For the second point, note that $H$ is independent of $t$ while
    the map $t \mapsto W(t)$ is norm--resolvent continuous. Indeed, the operator $(\ri - W(t))^{-1}$ is block diagonal with block elements $(\ri - w(n-t))^{-1}$, and we have the bound
\begin{align*}
    \left\| \dfrac{1}{\ri - w(n-t)} - \dfrac{1}{\ri - w(n - t')} \right\|_\op 
    &=  \left\| \dfrac{1}{\ri - w(n-t)} \left[ w(n - t) - w(n-t') \right] \dfrac{1}{\ri - w(n - t')} \right\|_\op  \\
    &\le \nu | t' - t |.
\end{align*}
    The translation equivariance follows from the translation equivariance of $W(t)$ and the periodicity of $H$. Since $\tau_1$ is unitary, it implies that the spectrum of $H^\sharp(t)$ is 1-periodic in $t$. 
    
    \medskip
    
    It remains to prove the third point, and compute the essential spectrum of $H^\sharp(t)$. We omit the argument $t$ for shortness. We identify 
    $\ell^2(\Z, \C^N)$ with  $\ell^2(\Z^-, \C^N) \oplus \ell^2(\Z^+, \C^N)$
    and define the restrictions $\Pi_L$ from $\ell^2(\Z, \C^N)$ onto $\ell^2(\Z^-, \C^N)$ and $\Pi_R $ from $\ell^2(\Z, \C^N)$ onto $\ell^2(\Z^+, \C^N)$.
    We write
    \begin{align*}
        H^{\sharp} =  H^{\sharp}_L  \oplus  H^{\sharp}_R + K, 
        \quad \text{with} \quad 
        H^{\sharp}_L =  \Pi_L  H^{\sharp} \Pi_L^*,  \quad H^{\sharp}_R =  \Pi_R  H^{\sharp} \Pi_R^* 
    \end{align*}
    and
      \begin{align*}
        K =  \Pi_R  H^{\sharp} \Pi_L^* + \Pi_L  H^{\sharp} \Pi_R^* = \Pi_R  H \Pi_L^* + \Pi_L  H \Pi_R^*,
    \end{align*}
    where the last equality comes from the fact that $W$ is diagonal.
    First, we claim that $K$ is a compact operator. Indeed, for $\ell \in \N$, we introduce the finite range one--periodic Hamiltonian
    \begin{equation} \label{eq:def:Hell}
        H_\ell \Psi := h_\ell * \Psi, \quad \text{with kernel} \quad h_\ell(n) = h(n) \1 (| n | \le \ell).
    \end{equation}
    Since $h_\ell$ converges to $h$ in $\ell^1(\Z, \C^{N \times N})$, we have $\| H - H_\ell \|_{\rm op} \to 0$ as $\ell \to \infty$ by Young's inequality (see the proof of Lemma~\ref{lem:bulk_self_adjoint}). So we also have $\| K - K_\ell \|_{\rm op} \to 0$, with $K_\ell := \Pi_R  H_\ell \Pi_L^* + \Pi_L  H_\ell \Pi_R^*$. In addition, $K_\ell$ is finite rank for all $\ell > 0$, so $K$ is compact, as stated.
    
    \medskip
    
    Since $K$ is compact, it does not affect the essential spectrum, so $\sigma_\ess(H^\sharp) = \sigma_\ess(H^\sharp_L) \bigcup \sigma_\ess(H^\sharp_R)$. Let us identify the essential spectra of $H^{\sharp}_L$ and $H^{\sharp}_R$.
    To the right of the wall, $W$ is bounded and decays to zero. Thus, $\Pi_R  W \Pi_R^*$ is compact and
    $$
    \sigma_{\ess}(H^\sharp_R) = \sigma_{\ess}(H_R).
    $$
We have $\sigma_{\ess}(H_R) = \sigma_\bulk$: the inclusion $ \sigma_{\ess}(H_R) \subset \sigma_\bulk$ can be shown directly by extending a singular Weyl sequence for $H_R$ by zero to the left. For the opposite inclusion, consider a singular Weyl sequence $(\Psi^k)_{k \in \N}$ for $H$ and fix a translation $R_k$ such that $\norm{\Pi_L \Psi^{k}_{R_k} } \le 1/k$, where $\Psi_R(n) := \Psi(n - R)$ is the translation of $\Psi$ by $R$. Then $\Pi_R \Psi^{k}_{R_k}$ is a singular Weyl sequence for $H_R$.     
 
 To the left of the wall, $W_L$ has compact resolvent, since it has eigenvalues of finite multiplicity and without accumulation points. Since $H$ is bounded, 
 $H_L (W_L - i)^{-1}$ is compact and $\sigma_\ess (H^{\sharp}_L) = \sigma_\ess (W_L) = \emptyset$.
 \end{proof}

\subsection{Computation of the spectral flow: reduction to the Jacobi case}

We now compute the spectral flow of $t \mapsto H^\sharp(t)$. Recall that, for $E \in \R \setminus \sigma_\bulk$, the spectral flow $\Sf(H^{\sharp}(t), E, [0, 1])$ counts the net number of eigenvalue branches that cross $E$ downwards, as $t$ varies between $0$ and $1$. A precise definition together with some properties can be found in Appendix~\ref{sec:appendix:SF}. The following example illustrates the concept.
\begin{example}
    For the soft wall operator $t \mapsto W(t)$, we have
    \[
        \forall E > 0, \qquad \Sf(W(t), E, [0, 1]) = -N, \qquad \text{while} \qquad
        \forall E < 0, \qquad \Sf(W(t), E, [0, 1]) = 0.
    \]
    This can be seen directly from Figure~\ref{fig:spectrum_Wt}. This case fits in the framework of Theorem~\ref{th:main_general_1d} with $H = 0$, since in this case the $N$ \emph{bands} consist of the singleton $0$.
\end{example}

We will take advantage of the robustness of the spectral flow and reduce the problem to simpler models. For instance, we will first prove the results for $1$--periodic Jacobi operators (for which $h(n) = 0$ for $|n | \ge 2$), and then extend it to all $1$--periodic Hamiltonians by an approximation argument. The proof that the spectral flow equals $-\cN(E)$ for Jacobi operator is postponed to the next Section. In this section, we explain how to generalize the result from Jacobi operators to any periodic Hamiltonian.

A {\em Jacobi operator} is an operator $H$ acting on $\ell^2(\Z, \C^N)$ of the form
\begin{equation} \label{eq:def:PeriodicJacobi}
 \forall \Psi \in \ell^2(\Z, \C^N), \qquad (H \Psi)_n := a_{n-1}^* \Psi_{n-1} + b_n \Psi_n + a_n \Psi_{n+1},
\end{equation}
where $(a_n)_{n \in \Z}$ and $(b_n)_{n \in \Z}$ are two families of complex--valued $N \times N$ matrices, with $b_n = b_n^*$. We will often represent such operators with a block matrix of the form
\[
    H = \left( \begin{array}{c c c |c c c}
    \ddots & \ddots & & & & \\
    \ddots  & b_{-2} & a_{-2} &  & & \\
    & a_{-2}^* & b_{-1} & a_{-1} & & \\
    \hline
    & & a_{-1}^* & b_0 & a_0 & \\
    & &  & a_0^* & b_1 & \ddots  \\
    & & & & \ddots & \ddots \\
\end{array} \right).
\]
We used the canonical basis of $\ell^2(\Z, \C^N)$, and the two straight lines indicate  the decomposition $\Z = \Z^- \cup \Z^+$ with $\Z^- := \{ n \in \Z, \ n < 0 \}$ and $\Z^+ = \{ n \in \Z, \ n \ge 0 \}$. A Jacobi operator is $K$--{\em periodic} if $a_{n+K} = a_n$ and $b_{n+K} = b_n$. It is $1$--periodic, or simply periodic, if $b_n = b$ and $a_n = a$ are independent of $n \in \Z$. Periodic Jacobi operators are of the convolution form \eqref{eq:convolutionForm} with
\[
h(-1)= a, \quad h(0) = b, \quad h(1) = a^*.
\]
In this simple case, the bulk spectrum~\eqref{eq:spectrum_Hbulk} is given by
\begin{equation} \label{eq:spectrum_Hbulk_Jacobi}
    \sigma_\bulk := \sigma(H) = \bigcup_{k \in \Omega^*} \sigma \left( H_k \right) \quad \text{with} \quad H_k := a^* \re^{ - \ri k} + b  + a \re^{\ri k}.
\end{equation}

We will establish a version of Theorem~\ref{th:main_general_1d} in the special case of Jacobi operators.
\begin{proposition} \label{prop:main_Jacobi}
    Let $H$ be a Jacobi operator and $w : \R \mapsto \cS_N$ be a Lipschitz function satisfying \eqref{eq:wall_def_intro}. Define the operator family $H^\sharp(t)$ as in Section~\ref{sec:softWall}.
   For all $E \in \R \setminus \sigma_\bulk$, the energy $E$ is in the essential gaps of all operators $H^\sharp(t)$, and
        \[
        \Sf( H^\sharp(t) , E, [0, 1]) = -\cN(E).
        \]
\end{proposition}
The proof of this Proposition is technically the most involved part of the paper. Its full proof is the topic of Section~\ref{sec:spectral_flows_Jacobi}. For now, let us explain how to deduce Theorem~\ref{th:main_general_1d} from it, that is how to go from the periodic Jacobi case to the general $1$--periodic Hamiltonians.


\begin{proof}[Proof of Theorem~\ref{th:main_general_1d} as a corollary of Proposition~\ref{prop:main_Jacobi}]

First, we explain how the case of finite-range interactions can be reduced to the Jacobi case, using {\em supercells}, and next we show that convolution by $\ell^1$-kernels can be approximated by finite rank kernels.

\subsubsection*{Step 1 : From $1$--periodic Jacobi operators to finite range interactions}
Assume indeed that the periodic Hamiltonian $H$ has a kernel $h$ with finite range, say $h(n) = 0$ for all $n > \ell$. Then we can write $H$ as a Jacobi operator acting on $\ell^2(\ell \Z, \C^{\ell N})$, by defining the supercell kernel $\widetilde{h}$ with
\[
    \widetilde{h}(-1)= \widetilde{a}, \quad \widetilde{h}(0) = \widetilde{b}, \quad \widetilde{h}(1) = \widetilde{a}^*.
\]
where $\widetilde{b}$ and $\widetilde{a}$ are matrices of size $(\ell N) \times (\ell N)$, given by
\[
    \widetilde{b} := \begin{pmatrix}
    h(0) & h(-1) & \cdots & h(-\ell+1) \\
    h(1) & h(0) & \cdots & h(-\ell+2) \\
    \vdots & \ddots & \ddots & \vdots \\
    h(\ell-1) & \cdots  & h(1) & h(0)
    \end{pmatrix}
    \quad  \text{and} \quad
   \widetilde{a} := \begin{pmatrix}
       h(-\ell) & 0 & \cdots & 0 \\
       h(-\ell+1) & h(-\ell) & \ddots & \vdots \\
    \vdots & \vdots & \ddots & 0\\
    h(-1) & h(-2) & \cdots  & h(-\ell)
    \end{pmatrix}.
\]
We denote by $\widetilde{H}$ the corresponding operator acting on $\ell^2(\Z, \C^{\ell N})$. Then $\sigma(H) = \sigma(\widetilde{H})$, and each Bloch band of $H$ becomes a set of $\ell$ bands for $\widetilde{H}$. In particular, $\cN_{\widetilde{H}}(E) = \ell \cN_H(E)$. 

\medskip

For the supercell soft wall, we set
\[
\widetilde{w}(x) := \begin{pmatrix}
w(\ell x) & 0 & \cdots & 0 \\
0& w(\ell x +1) & \ddots & \vdots \\
\vdots & \ddots & \ddots & 0 \\
0 & \cdots & 0 & w(\ell x + \ell - 1)
\end{pmatrix},
\]
and define $\widetilde{W}(t)$ by the construction in Subsection~\ref{sec:softWall}.
The operator $\widetilde{H}^\sharp(t):= \widetilde{H} + \widetilde{W}(t)$ is a periodic Jacobi operator with the same spectrum as the original $H^\sharp(\ell t)$. As $t$ increases from $0$ to $1$, the parameter $\ell t$ varies from $0$ to $\ell$. In particular, the spectral flow satisfies
\[
    \Sf \left( \widetilde{H}^\sharp(t), E, [0, 1]\right) 
    = \Sf \left( H^\sharp(t), E, [0, \ell]\right) 
    = \ell \, \Sf \left( H^\sharp(t), E, [0, 1]\right). 
\]
On the other hand, Proposition~\ref{prop:main_Jacobi} gives
\[
    \Sf \left( \widetilde{H}^\sharp(t), E, [0, 1]\right) = - \cN_{\widetilde{H}}(E) = -\ell \cN_H(E), 
\]
which proves Theorem~\ref{th:main_general_1d} in the case of finite range kernels.

\subsubsection*{Step 2 : From finite range to integrable kernels}

Now consider a general kernel $h$ in $\ell^1$. For $\ell \in \N$, we consider the periodic Hamiltonian $H_\ell$ with finite range potential
\[
    h_\ell(n) := h(n) \1(| n | \le \ell).
\]
This kernel has already been introduced in~\eqref{eq:def:Hell}. Recall that $\| H - H_\ell \|_{\rm op} \to 0$ as $\ell \to \infty$ by Young's inequality. Note also that $\| H^\sharp(t) - H_\ell^\sharp(t) \|_{\rm op} = \| H - H_\ell \|_{\rm op}$ is independent of $t$, hence goes to $0$ uniformly in $t \in \R$. In particular, the spectrum of $H_\ell^\sharp(t)$ approaches uniformly the spectrum of $H^\sharp(t)$. 

\medskip

Let us fix $E \in \R \setminus \sigma_\bulk$ and set
\[
    \varepsilon := \dist (E, \sigma_\bulk) > 0.
\]
The previous point shows that there exists $\ell_0$ such that, for all $\ell \ge \ell_0$, we have $\| H - H_\ell \|_{\rm op} \le \varepsilon/2$. For $\ell \ge \ell_0$, we obtain that $E$ is not the essential spectrum of $H_\ell^\sharp(t)$ for all $t$, hence the spectral flow $\Sf \left( {H}^\sharp_\ell(t), [0, 1] , E\right)$ is well-defined. In addition, we have, with obvious notations, $\cN_{H_\ell}(E) = \cN_{H}(E)$ for $\ell \ge \ell_0$. 

\medskip

We can now apply the robustness of the spectral flow with respect to small bounded perturbation, see Lemma~\ref{lem:stability_SF_periodic} in the Appendix, and get that
\[
    \Sf \left( H^\sharp(t), E, [0, 1] \right) = \lim_{ \ell \to \infty} \Sf \left( H^\sharp_\ell(t), E, [0, 1] \right)  = - \cN(E),
\]
where we used that Theorem~\ref{th:main_general_1d} holds for finite range interactions.
\end{proof}

\subsection{The spectrum at fixed $t_0$, proof of Theorem~\ref{th:main_fix_t0}} 
\label{sec:proof_main_fix_t0}

In this section, we prove Theorem~\ref{th:main_fix_t0}, which states that for all $t_0$ and all $E \in \R \setminus \sigma_\bulk$, there are at least $\cN(E)$ eigenvalues in each interval of size $\nu$ in the gap where $E$ lies. 

\begin{proof}
To start, recall all eigenvalues branches are $\nu$-Lipschitz. Unfortunately, we could not find a reference for this well-known result, so let us give a short proof of this fact, and prove the following: Let $t \mapsto A(t)$ be a continuous family of self-adjoint operators which is $\nu$-Lipschitz, in the sense $\| A(t) - A(s) \|_\op \le \nu | t - x |$ (this implies that they all have the same domain), then all isolated branches of eigenvalues are also $\nu$-Lipschitz. 

\medskip

In the case where $A(t)$ is matrix valued (finite dimensional case), we first note that for all normalized $\Psi$, the map $t \mapsto \bra \Psi, A(t) \Psi \ket$ is $\nu$-Lipschitz. Indeed, we have
\[
    \left| \langle \Psi, A(t) \Psi \rangle - \langle \Psi, A(s) \Psi \rangle \right| = \left| \bra \Psi, A(t) - A(s) , \Psi \ket \right| \le \| A(t) - A(s) \|_\op \| \Psi \|^2 \le \nu | t -s |.
\]
Since the $j$-th lowest eigenvalue $\lambda_j(t)$ is defined as a min--max combination of such $\nu$--Lipschitz functions, we conclude that all maps $t \mapsto \lambda_j(t)$ are $\nu$--Lipschitz. 

\medskip

We now turn to the general case. Let $t_0 \in \R$, let $\lambda_0 := \lambda(t_0)$ be an isolated eigenvalue of $A(t_0)$, and let $m = \dim \Ker(A(t_0) - \lambda_0)$ be its multiplicity. By continuity of the spectrum~\cite{Kat-95}, there is $\varepsilon > 0$ and $\eta > 0$ so that for all $t \in (t_0 - \varepsilon, t_0 + \varepsilon)$ the interval $(\lambda_0 - \eta, \lambda_0 + \eta)$ contains exactly $m$ eigenvalues of $A(t)$. We label the smallest eigenvalue as $\lambda_1(t)$ and $\lambda_m(t)$ for the largest one. In addition, by taking $\varepsilon$ and $\eta$ small enough, we may further assume that
\[
\begin{cases}
        {\rm dist} \left( \lambda_0 - \eta, \sigma (A(t)) \right) = {\rm dist} \left( \lambda_0 - \eta, \lambda_1(t) \right) = \lambda_1(t) - (\lambda_0 - \eta), \\
        {\rm dist} \left( \lambda_0 + \eta, \sigma (A(t)) \right) = {\rm dist} \left( \lambda_0 + \eta, \lambda_m(t) \right)
         = (\lambda_0 + \eta) - \lambda_m(t).
\end{cases}
\]
In particular, since $A(t)$ is self-adjoint, we have $\| \left( A(t) - (\lambda_0 - \eta) \right)^{-1} \|_\op = \left(  \lambda_1(t) - (\lambda_0 + \eta) \right)^{-1}$. Let us prove that $t \mapsto \lambda_1(t)$ is $\nu$--Lipschitz on the interval $(t_0 - \varepsilon, t_0 + \varepsilon)$. We have
\begin{align*}
    \left| \lambda_1(t) - \lambda_1(s) \right| & = \left| \lambda_1(t) - (\lambda_0 - \eta) \right| \cdot  \left| \lambda_1(s) - (\lambda_0 - \eta) \right|  \left| \frac{1}{\lambda_1(s) - (\lambda_0 - \eta)} - \frac{1}{\lambda_1(t) - (\lambda_0 - \eta)} \right|,
\end{align*}
and the last term is also
\begin{align*}
    & \left| \left\| \frac{1}{ A(s) - (\lambda_0 - \eta)} \right\|_\op -  \left\| \frac{1}{ A(t) - (\lambda_0 - \eta)} \right\|_\op \right|  \le \left\| \frac{1}{ A(s) - (\lambda_0 - \eta)} - \frac{1}{ A(t) - (\lambda_0 - \eta)} \right\|_\op \\
    & \quad  \le \left\| \frac{1}{ A(s) - (\lambda_0 - \eta)} \right\|_\op \left\| \frac{1}{ A(t) - (\lambda_0 - \eta)} \right\|_\op  \left\| A(t) - A(s) \right\|_\op
\end{align*}
We conclude that $\left| \lambda_1(t) - \lambda_1(s) \right|  \le \nu | t - s |$. The proof that $\lambda_m(s)$ is $\nu$--Lipschitz on the interval $(t_0 - \varepsilon, t_0 + \varepsilon)$ is similar. Finally, for any local continuous branch of eigenvalue $\lambda(\cdot)$ with $\lambda(t_0) = \lambda_0$, we have
\[
    \lambda_1(t) - \lambda_0 \le \lambda(t) - \lambda_0 \le \lambda_m(t) - \lambda_0, \quad \text{so} \quad
    \left| \lambda(t) - \lambda(t_0) \right| \le \nu | t - t_0 |.
\]
We therefore proved that for all $t_0$, there is $\varepsilon > 0$ so that any continuous branch $\lambda(\cdot)$ of isolated eigenvalue satisfies $| \lambda(t) - \lambda(t_0) | \le \nu | t - t_0 |$ for all $t \in (t_0 - \varepsilon, t_0 +  \varepsilon)$. By compactness of $[t, s]$ for all $t, s \in \R$, we conclude that $\lambda(\cdot)$ is globally $\nu$--Lipschitz.

\medskip

Now, let $E \in \R \setminus \sigma_\bulk$ with $\cN(E) > 0$. Since there is a nonzero spectral flow, $H^\sharp(t)$ has eigenvalues for at least some values of $t$. By shifting $t$ if necessary, we may assume $t=0$. We label this eigenvalue $\lambda_0(0)$, and denote by $\lambda_0(t)$ a corresponding continuous branch.  We use this as a starting point to label the eigenvalue curves $\lambda_k(t)$ for $t \in [0,1]$ such that, for all $t \in [0,1]$
\[
\lambda_k(t) \le \lambda_{k+1}(t), \text{ where } k \in \Z,
\]
and with the convention that $\lambda_k(t)$ equals the edge of the essential spectrum if there exists no $k$-th eigenvalue at this value of $t$. Since the spectrum is periodic, we have $\lambda_k(1)= \lambda_{j}(0)$ for some $j \in \Z$. The spectral flow equals $-\cN(E)$, so there are exactly $\cN(E)$ eigenvalues that cross $\lambda_k(1)$ upwards. We deduce that we actually have
$\lambda_k(1) = \lambda_{k + \cN(E)}(0)$ for all values of $k \in \Z$ such that $\lambda_k(0)$ is an eigenvalue of $H^\sharp (0)$. 
Then, the Lipschitz continuity of $t \mapsto \lambda_k(t)$ gives
\[
\lambda_{k + \cN(E)}(0)- \lambda_{k}(0) = \lambda_k(1) - \lambda_k(0) \le \nu. 
\]
Hence, there are at least $\cN(E)$ eigenvalues for $H^\sharp(0)$ in each interval of length $\nu$ in the gap. 
If $\nu$ is smaller than the width of the gap, the same argument shows that there is edge spectrum for any $t_0 \in [0,1]$, and we can repeat the counting argument for this $t_0$. 

\end{proof}

\subsection{Example : the SSH model}
\label{sec:SSH}

Before going to the technicalities of the proof, let us illustrate our Theorem in the case of the Su--Schrieffer--Heeger (SHH) model~\cite{SuSchHee-79}.

\subsubsection{Presentation of the model, and basic facts}

\begin{figure}[ht]
    \centering
    \includegraphics[width=0.5\textwidth]{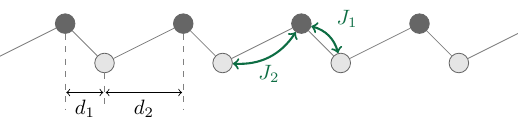}
    \includegraphics[width=0.4\textwidth]{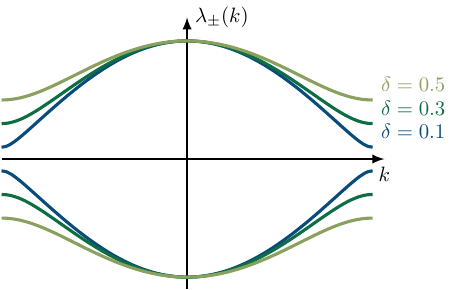}
    \caption{Left: illustration of the SSH model. Right: its band structure for different values of $\delta := | |J_1| - |J_2| |$.}
    \label{fig:SSH_chain}
\end{figure}
The SSH model is a tight-binding Hamiltonian describing a polyacetylene chain, a simple polymer where the distance between neighboring carbon-hydrogen (CH) pairs takes alternating values. We denote by $d_1 > 0$ (resp. $J_1 \in \C$) the distance (resp. hopping amplitude) between the $(2n)$-th (CH) pair and the $(2n + 1)$-th one, and by $d_2 = 1 - d_1 > 0$ (resp. $J_2 \in \C$) the ones between $(2n+1)$-th and the $(2n+2)$-th. If $d_1 \neq d_2$, then the distances alternate, hence so are the hopping parameters, see Figure~\ref{fig:SSH_chain}. In the tight-binding approximation (and for an idealized infinite chain), the Hamiltonian takes the form of a $2$--periodic Jacobi operator acting on $\ell^2(\Z, \C)$ with alternating hopping strengths $J_1$ and $J_2$, that is
 \[
\begin{pmatrix}
    \ddots & \ddots & & & & \\
    \ddots & 0 & J_1 & & & \\
    &  J_1^* & 0 & J_2 & & \\
    & & J_2^* & 0 & J_1 & \\
    & & & J_1^* & 0 & \ddots \\
    & & & & \ddots & \ddots
\end{pmatrix}.
\]
It most cases, we take $J_1, J_2$ real numbers. 
Complex values of $J_1, J_2$ happen to appear in the case of the Wallace model for graphene with zigzag cuts.
In this case, the phases can be removed by defining a \emph{gauge transformation} $U= e^{i A}$ where $A$ is the multiplication operator by 
$$
A_n =\begin{cases}
\sum_{j=0}^{n-1} \phi_n & n \ge 1 \\
0 & n=0 \\
-\sum_{j=n}^{-1} \phi_n & n \le -1 
 \end{cases} \quad  \text{ with } 
 \phi_n =\begin{cases}
\arg{J_1} & n \text{ odd} \\
\arg{J_2} & n \text { even }
\end{cases}.
$$

\medskip
The SSH Hamiltonian becomes a $1$-periodic Jacobi operator $H$ if we take $N = 2$ and the $2 \times 2$ blocks
    \begin{equation} \label{eq:matrices_SSH}
    b = \begin{pmatrix}
        0 & J_1 \\ J^*_1 & 0
    \end{pmatrix}
    \quad \text{and} \quad
    a =  \begin{pmatrix}
        0 & 0 \\ J_2 & 0
    \end{pmatrix}.
    \end{equation}
 \medskip
According to \eqref{eq:spectrum_Hbulk_Jacobi}, the bulk spectrum can be computed from 
$$H_k = b + a \re^{\ri k} + a^* \re^{- \ri k} = \begin{pmatrix}
    0 & J_1 + J_2 \re^{-\ri k} \\
    J_1 + J_2 \re^{\ri k} & 0
\end{pmatrix},$$ 
whose spectrum is $\pm | J_1 + J_2 \re^{\ri k} |$. Note that, as $k$ varies on $\R$, the complex number $J_1 + J_2 \re^{\ri k}$ describes a circle with center $J_1$, and radius $| J_2 |$. With this in mind, we conclude that
\[
    \sigma_\bulk := \sigma(H) = \sigma_\ess (H) = [ -W, - \delta ] \cup [\delta, W], \quad W := | J_1 | + | J_2 |, \quad \delta := \left| | J_1 | - | J_2 | \right|.
\]
In particular, is $| J_1 | \neq | J_2 |$, there is a gap of size $2 \delta > 0$ around the origin, separating the two Bloch bands. In this gap, we have $\cN(E) = 1$.

\begin{remark}[Hard truncation for SSH]
\label{rem:hardcut_SSH}
For the hard cut Hamiltonian
\[
H^\sharp := \begin{pmatrix}
    0 & J_1 & &   \\
    J_1^* & 0  & J_2 &  \\
     & J_2^* & 0 & \ddots \\
     & & \ddots & \ddots 
\end{pmatrix},
\]
seen as an operator on $\ell^2(\N_0, \C)$ (without supercell), then any $\psi \in \ell^2(\N_0, \C)$ solution to $H^\sharp \psi = 0$ solves $J_1 \psi_1 = 0$ and $\psi_{n+2} = \frac{-J_1^*}{J_2} \psi_n$ for all $n \ge 0$. We easily deduce that $\psi_{2n + 1} = 0$, and that $\psi$ is non null and square integrable at $+ \infty$ iff $| J_1 | < | J_2 |$. So $0 \in \sigma \left( H^\sharp \right)$ iff $| J_1 | < | J_2 |$. When the hard cut is translated, and crosses one atom, the role of $J_1$ and $J_2$ are exchanged. So, depending on the location of the hard cut, the eigenvalue $0$ is a (protected) eigenvalue, or is not an eigenvalue.
\end{remark}

\subsubsection{Numerical simulations for the SSH chain}

We now compute numerically the spectral flow in the SSH model, in the presence of a shifted soft wall. For the soft wall, we take
\[  
    w(x) = \begin{pmatrix}
        w_1(x) & 0 \\ 0 & w_1(x+d_1)
    \end{pmatrix}
    \quad \text{with} \quad 
    w_1(x) := \begin{cases}
        0 & \quad \text{for} \quad x \ge 0 \\
        - \nu x  & \quad \text{for} \quad x \le 0
    \end{cases},
\]
which corresponds to the evaluation of $w_1$ at the location of each of the two atoms in the first unit cell. Note that $w$ satisfy the assumption of a soft wall, and is $\nu$--Lipschitz.

\medskip

To compute the edge spectrum of the edge operator, we choose the following simple procedure: we restrict $H^\sharp(t)$ to the the box $[-L, L]$ (with $2L+1$ unit cells), so we consider the finite matrices
\[
    H^\sharp_L(t) := \begin{pmatrix}
        b & a &  & &  \\
        a^* & b & a &  & \\
        0 & \ddots & \ddots & \ddots &  \\
        & &  a^* & b & a \\
        & &  & a^* & b
    \end{pmatrix} 
    + 
    \begin{pmatrix}
    W_t(-L) & & &  \\
    & W_t(-L + 1)  & & \\
    & & \ddots & \\
    & &  & W_t(L)
    \end{pmatrix}.
\]
Then, for each eigenpair $(\lambda, u)$ with $\| u \|_{\ell^2_L} = 1$, we check if the mass of $u$ is localized near the origin, or, more specifically, not localized at the boundary of our simulation box. This procedure discards the spurious edge modes due to the truncation.  We detect that $u$ is a {\em true} edge mode whenever $\| u \1(-L/2 \le n \le L/2) \|_{\ell^2_L} \ge 3/4$.

\medskip

For our numerical simulations, we took $d_1 = 1/4$ (so $d_2 = 3/4$) for the distances between the atoms, $J_1 = 3/2$ and $J_2 = 1/2$ for the hopping parameters, $L = 100$, and we took different values of $\nu \in [0.5, 1, 5, 10]$. The results are displayed in Figure~\ref{fig:SSH}. In each of these figures, we observe a spectral flow of $-1$ in the middle gap, and of $-2$ above the second Bloch band, in accordance with Theorem~\ref{th:main_general_1d}. The flow of eigenvalues becomes steeper and steeper as $\nu$ increases, as predicted in Theorem~\ref{th:main_fix_t0}. Note that our soft wall $w$ is continuous, but not continuously differentiable. This explains the {\em kinks} in the Right figures ($\nu = 5$ and $\nu = 10$), which happens when the kink of the wall touches the second atom located at $x_2 = 1/4$. Note also that since we took $x \mapsto w_1$ non-increasing, the maps $t \mapsto W(t)$ and $t \mapsto H^\sharp(t)$ are operator non-decreasing. So all eigenvalue branches are non-decreasing, as can be checked on these pictures. 

Finally, the first two simulations correspond to values of $\nu$ smaller than the size of the gap ($2|J_1 - J_2|=2$) and we can observe that the lower bound on the number of eigenvalues from Theorem~\ref{th:2d_fix_t0} is almost achieved here.

\begin{figure}
    \centering
    \includegraphics[width=0.24 \textwidth]{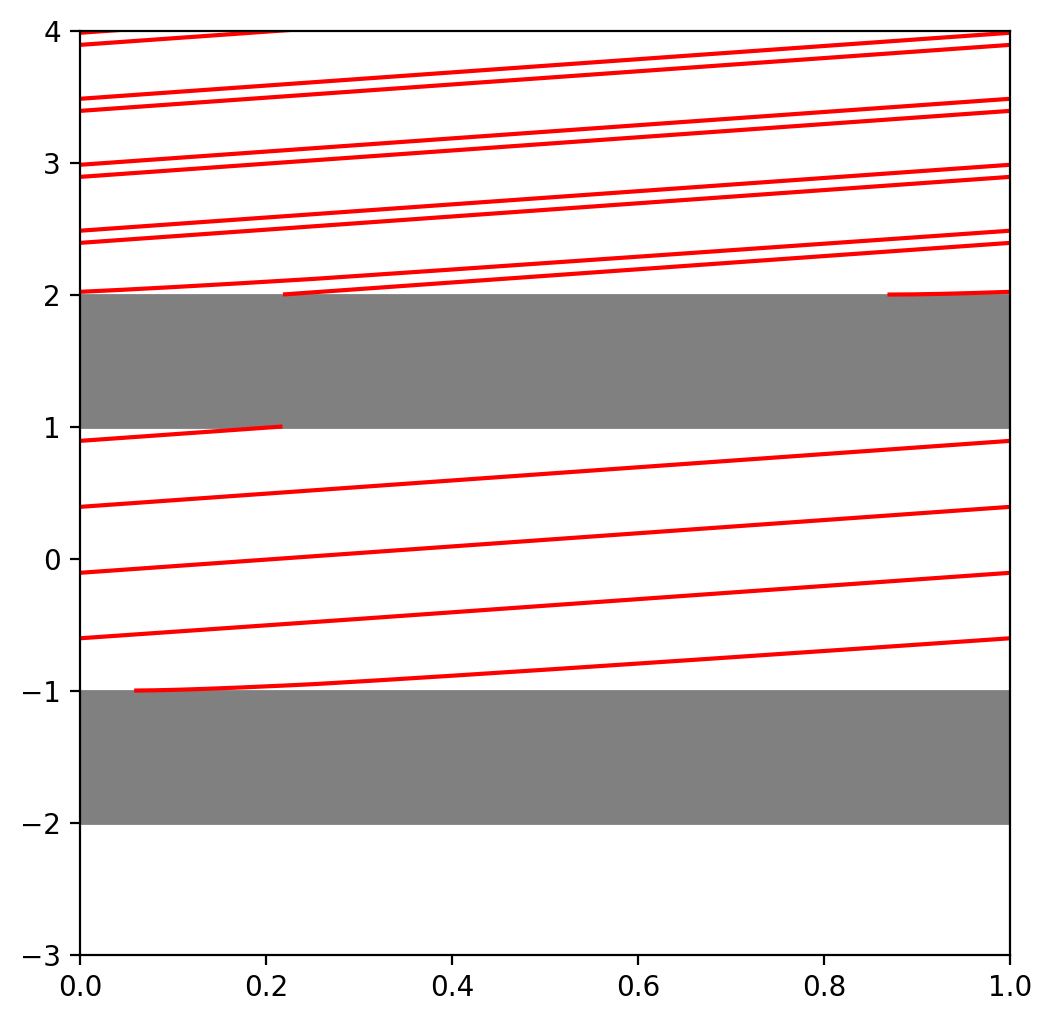}
    \includegraphics[width=0.24 \textwidth]{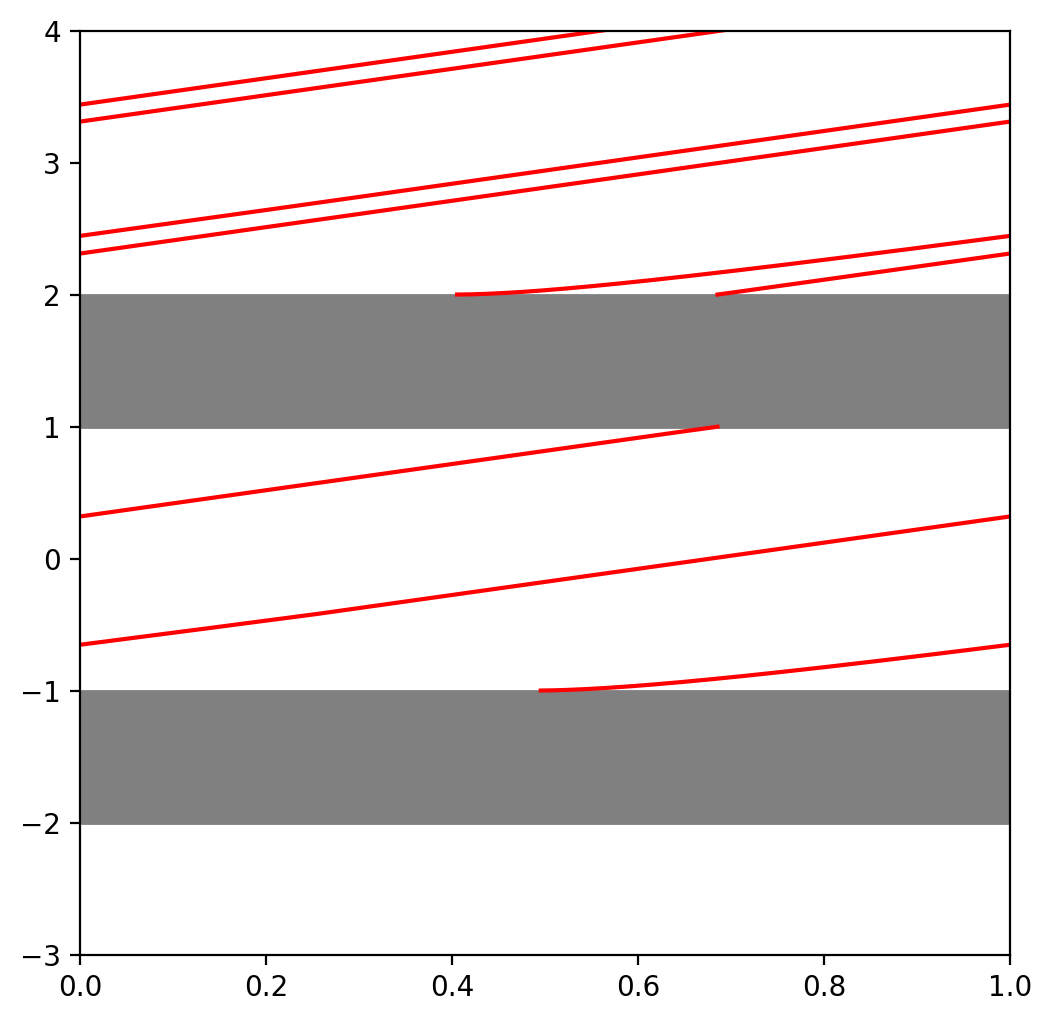}    \includegraphics[width=0.24 \textwidth]{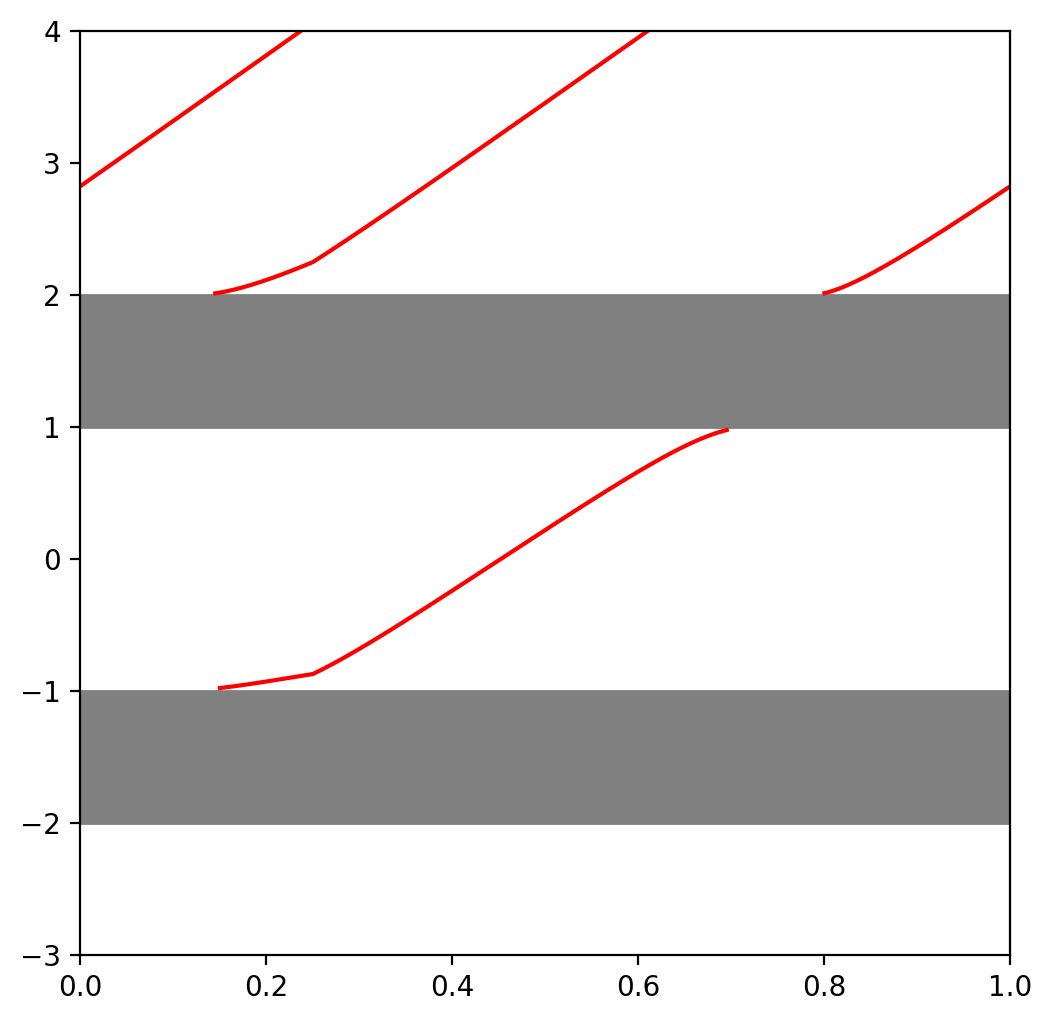}
    \includegraphics[width=0.24 \textwidth]{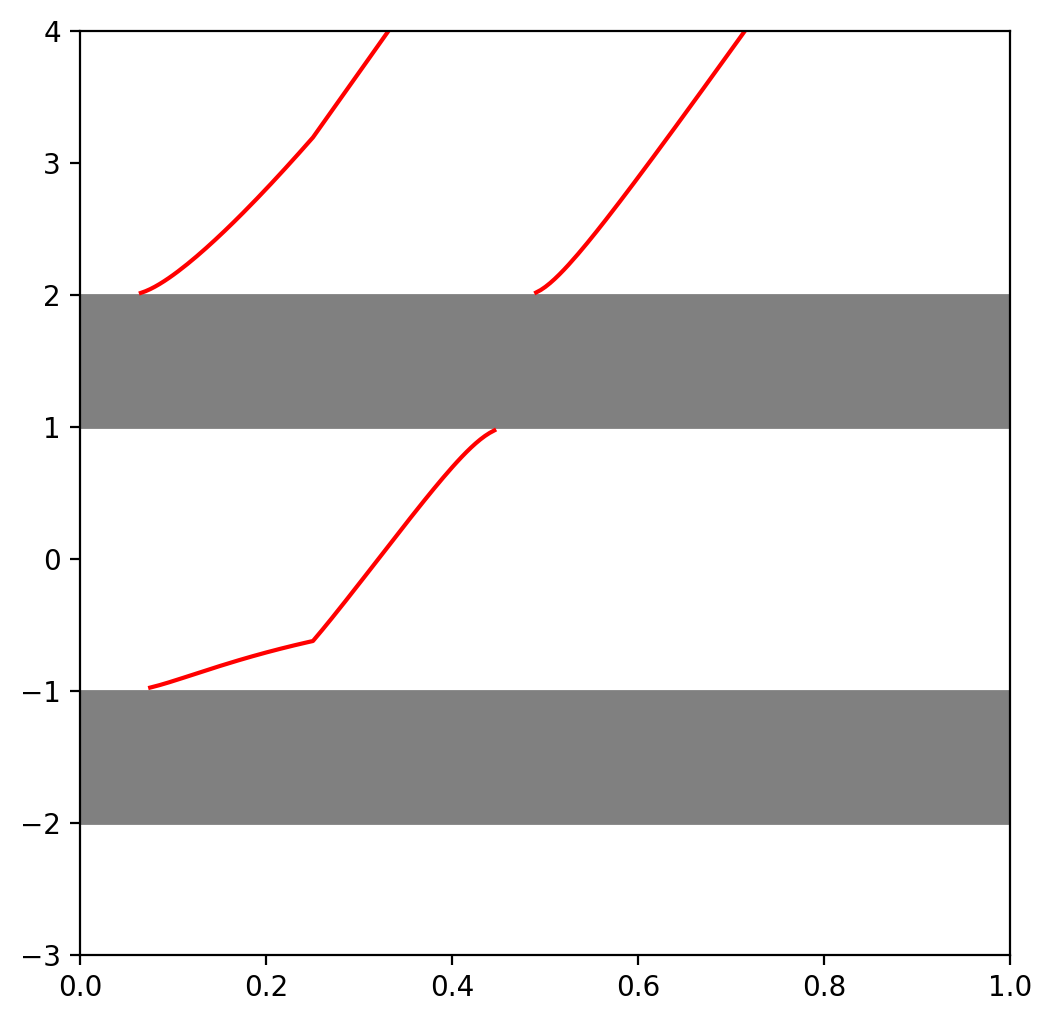}
    \caption{Numerics for the spectrum of $t \mapsto H^\sharp(t)$ in the SSH model with the potential $w_1(x)$ , for different values of $\nu$. From left to right, $\nu = 0.5$, $\nu = 1$, $\nu = 5$ and $\nu = 10$. The essential spectrum is displayed in grey, and the edge spectrum in red. For $E$ in the gap between the $2$ bands, $\cN(E) = 1$ and the spectral flow is $-1$, above the bands, $\cN(E) = 2$ and the spectral flow equals $-2$. When the Lipschitz constant increases, the curves become steeper.}
    \label{fig:SSH}
\end{figure}


\section{Computation of spectral flows in the Jacobi case}
\label{sec:spectral_flows_Jacobi}

It remains to prove Proposition~\ref{prop:main_Jacobi}, that is to prove that the spectral flow is $-\cN(E)$ in the case of periodic Jacobi operators. In order to do so, we will use the robustness of the spectral flow to make successive modifications to the operators.
In Subsections~\ref{ssec:cut_operator}, the Jacobi operator $H^\sharp(t)$ will be shown to have the same spectral flow as a {\em cut} model $H^\Sigma_\cut (t)$, which, itself, has a same spectral flow equal to a dislocated model $\widetilde H(t)$, see Subsection~\ref{sec:dislocated}. This last spectral flow can be computed explicitly, and turns out to be equal to $-\cN(E)$.

\medskip

Before we start, we emphasize that several techniques are available to study Jacobi operators, such as the ones involving {\em transfer matrices}. If $H$ is a Jacobi operator, then the equation $H \Psi = E \Psi$ implies that
\[
    \begin{pmatrix}
        \Psi_{n+1} \\ \Psi_n 
    \end{pmatrix} = T(E) \begin{pmatrix}
    \Psi_n \\ \Psi_{n-1}
    \end{pmatrix}, \quad \text{with} \quad
    T(E) := \begin{pmatrix}
        - a^{-1} (b - E) & -a^{-1} a^* \\
        1 & 0
    \end{pmatrix},
\]
so that we can study the properties of $\Psi$ by studying the transfer matrix $T(E)$. However, this matrix is not always well-defined (in the previous formulation, we need $a$ to be invertible for instance). Actually, we were not able to recover our results using the formalism of transfer matrices, even in the case where $a$ is invertible. Our proof rather relies on a direct study of the full operator, and a variational argument.

\subsection{Reduction to the cut case} 
\label{ssec:cut_operator}

We first deform the potential $w$ to a simpler form. In what follows, we fix a Jacobi operator $H \in \ell^2(\Z, \C^N)$ with kernel $h$ having diagonal entries $b$ and off-diagonal entries $a$. Let us fix $E \in \R \setminus \sigma_\bulk$, and consider some $\Sigma \in \R$ large enough so that
\[
    \Sigma \ge E + 4 C_{a,b}  + 1, \qquad \text{with} \quad
    C_{a,b} := \max \{ \| a \|_\op, \| b \|_\op \}.
\]
Since $w(x) \to +\infty$ as $x \to - \infty$, there is $x_\Sigma \in \R$ so that $w(x) \ge \Sigma$ for all $x < x_\Sigma$. Without loss of generality, we may assume $x_\Sigma \le -1$. We consider the modified potential (here and in what follows, $\Sigma$ is shorthand for $\Sigma \bbI_N$)
\begin{equation} \label{eq:def:wSigma}
    w_\Sigma(x) := \begin{cases}
        0 & \quad \text{if} \quad x \ge 0, \\
        -x (\Sigma  - b) & \quad \text{if} \quad x \in [-1, 0], \\
        (\Sigma   - b)  & \quad \text{if} \quad x \in [x_\Sigma, -1], \\
        (x - x_\Sigma + 1) (\Sigma  - b) + (x_\Sigma - x) w(x) & \quad \text{if} \quad x \in [x_\Sigma - 1, x_\Sigma], \\
        w(x) & \quad \text{if} \quad x \le x_\Sigma - 1.
    \end{cases}
\end{equation}
We also denote by $W_\Sigma(t)$ the corresponding modified wall operator. The key properties of $w_\Sigma$ and $W_\Sigma$ are summarized in the following straightforward Lemma.
\begin{lemma} \label{lem:wSigma}
    The potential $w_\Sigma$ is Lipschitz (hence continuous), is linear on $[-1, 0]$ and satisfies 
    \[
        \forall x > 0, \quad w_\Sigma(x) = 0, \quad \text{while} \quad
        \forall x < -1, \quad w_\Sigma(x) \ge E + 3 C_{a,b} + 1.
    \]
\end{lemma}
\begin{proof}
    The first points are obvious. If $x < -1$, we note that $(\Sigma - b) \ge \Sigma - \| b \|_\op \ge E + 3 C_{a,b} + 1$ and that, if $x < x_\Sigma \le -1$, we also have $w(x) \ge \Sigma \ge E + 3 C_{a,b} + 1$. The last inequality follows since $w_\Sigma$ is a combination of these two matrices for $x < -1$.
\end{proof}

The difference of potentials $(w - w_\Sigma)$ is continuous, vanishes for $x < x_\Sigma$, and goes to $0$ as $x \to + \infty$. In particular, for all $t \in \R$, $(W - W_\Sigma)(t)$ is compact, and the map $t \mapsto (W - W_\Sigma)(t)$ is continuous and translation equivariant.

\medskip
 
Let $H^{ \Sigma}(t) := H + W_\Sigma(t)$ be the corresponding edge Hamiltonian. We have
\[
    H^{ \Sigma}(t) - H^{\sharp}(t) = W_\Sigma(t) - W(t),
\]
which is compact for all $t \in \R$. In addition, both operators $H^\Sigma$ and $H^\sharp$ are translation equivariant. By robustness of the spectral flow with respect to compact operators, see Lemma~\ref{lem:stability_SF_compact_periodic}, we deduce that
\begin{equation} \label{eq:H_sharp_to_H_Sigma}
\Sf \left( H^{\sharp}(t), E, [0, 1] \right) = \Sf \left( H^{\Sigma}(t), E, [0, 1] \right).
\end{equation}
In what follows, we study the operator $H^{ \Sigma}(t)$. The main advantage of this new potential $w_\Sigma$ is that it is piece-wise linear for $x \in [-1, 0]$.


\medskip

Our goal is now to compute the spectral flow of $t \mapsto H^\Sigma(t)$. For $t \in [0,1]$ we define the cut operator
\[
    K(t) := \left( \begin{array}{c c c |c c c}
        \ddots & \ddots & & & & \\
        \ddots  & 0 & (1 - t) a &  & & \\
         & (1 - t) a^* & 0 & a & & \\
        \hline
        & & a^* & 0 & t a & \\
        & &  & t a^* & 0 & \ddots  \\
        & & & & \ddots & \ddots \\
    \end{array} \right),
\]
where all undisplayed matrix elements are $0$. This operator is piece-wise linear, of the form $K(t) = tK(1) + (1 - t) K(0)$, with
\[
    K(0) = \left( \begin{array}{c c c |c c c}
        \ddots & \ddots & & & & \\
        \ddots  & 0 &  a &  & & \\
        &  a^* & 0 & a & & \\
        \hline
        & & a^* & 0 & 0 & \\
        & &  & 0 & 0 & \ddots  \\
        & & & & \ddots & \ddots \\
    \end{array} \right), 
\qquad 
    K(1) = \left( \begin{array}{c c c |c c c}
        \ddots & \ddots & & & & \\
        \ddots  & 0 & 0 &  & & \\
        & 0 & 0 & a & & \\
        \hline
        & & a^* & 0 & a & \\
        & &  & a^* & 0 & \ddots  \\
        & & & & \ddots & \ddots \\
    \end{array} \right).
\]
Note that $K(1) = \tau_1 K(0) \tau_1^*$, so the map $t \mapsto K(t)$ can be extended by translation equivariance, and this extension is continuous. For simplicity, in what follows, we will restrict the study to $t \in [0, 1]$. 

We now define the cut Jacobi operator
\[
    H_\cut^\Sigma(t) := H^\Sigma(t) - K(t), \quad \text{so that} \quad
    H^\Sigma(t) = H^\Sigma_\cut (t) + K(t).
\]
Since $t \mapsto H^\Sigma(t)$ and $t \mapsto K(t)$ are continuous and translation equivariant, so is $t \mapsto H_\cut^\Sigma(t)$. So, using again the robustness of the spectral flow with respect to compact, translation equivariant operators, see Lemma~\ref{lem:stability_SF_compact_periodic}, we find that
\[
    \Sf(H^\Sigma(t) , E, [0, 1]) =  \Sf(H^\Sigma_\cut(t) , E, [0, 1]).
\]
Also, by construction of the cut operator, for all $t \in [0, 1]$, $H_\cut^\Sigma(t)$ is block diagonal with respect to the decomposition $\Z = \Z^- \cup \Z^+$ (note that this is no longer the case for $t \in \R \setminus [0, 1]$). We write $H_\cut^\Sigma(t) = H^\Sigma_L(t) \oplus  H^\Sigma_R(t)$. 
In an analogous way, we decompose the bulk operator as $H_{\rm cut}(t) = H_L(t) \oplus H_R(t)$ and the wall operator as $W(t) = W_L(t) \oplus W_R(t)$, so that $H^\Sigma_L(t) = H_L(t) + W_L(t)$ and $H^\Sigma_R(t) = H_R(t) + W_R(t)$. 

\medskip

These left and right operators are not periodic nor translation equivariant in $t$. Still, one can define the spectral flow as the net number of eigenvalues crossing the energy $E$ downwards. Since these operators are no longer translation equivariant, a priori this spectral flow could depend on the value of $E$ in the essential gap (see Appendix~\ref{sec:appendix:SF}). In any case, since $H_\cut^\Sigma(t)$ is block diagonal, we have
\[
    \Sf(H^\Sigma_\cut(t) , E, [0, 1]) = \Sf(H_L^\Sigma(t) , E, [0, 1])  + \Sf(H_R^\Sigma(t) , E, [0, 1]).
\]

\begin{lemma} \label{lem:H_Sigma_to_H_R_Sigma}
    We have 
    $\Sf \left( H_L^{\Sigma}(t), E, [0, 1] \right) = 0$ and thus
    \[
    \Sf(H^\Sigma (t) , E, [0, 1]) = \Sf(H_R^\Sigma(t) , E, [0, 1]).
    \]
\end{lemma}

\begin{proof}
    We write $H^\Sigma_L(t) = H_L(t) + W_L(t)$. Explicitly, the matrix representing $H_L(t)$ is given by (we keep the straight lines to emphasize that this matrix in only infinite in one direction)
    \[
    H_L(t) = \left( \begin{array}{c c c c | }
        \ddots & \ddots & &  \\
        \ddots  & b &  a &    \\
        &  a^* & b & ta   \\
        & & ta^* & b   \\
        \hline
    \end{array} \right).
    \]
    Note that $H_L(t)$ is a bounded operator on $\ell^2(\Z^-, \C^N)$, with uniform bound
    \[
        \forall t \in [0, 1], \qquad \| H_L(t) \|_{\op} \le 3 C_{a,b}.
    \]
    Indeed, we have, using that $0 \le t \le 1$, and with $\Psi_0 := 0$, that
    \[
        \| H_L(t) \Psi  \|_{\ell^2}^2 \le \sum_{n=-\infty}^{-1} ( \| a^* \|_{\rm op} \| \Psi_{n-1} \| + \| b \|_{\rm op} \| \Psi_n \| + \| a \|_{\rm op} \| \Psi_{n+1} \| )^2 \le  C_{a,b}^2 9 \sum_{n=-\infty}^{-1} \| \Psi_n \|^2 = 9 C_{a,b}^2 \| \Psi \|_{\ell^2}^{2}.
    \]
    
    On the other hand, $W_L(t)$ is block diagonal with block elements involving $w_\Sigma(n - t)$. Since $n\le -1$ and $t \in [0, 1]$, we have $n - t \le -1$, so $W_L(t) \ge \Sigma - \| b \|_\op$. Hence, 
$$ H^\Sigma_L(t) \ge W_L(t) - \| H_L(t) \|_\op \ge \Sigma - \| b \|_\op -  3 C_{a,b} \ge E+1,$$ 
     by our choice of $\Sigma$ in~\eqref{eq:def:wSigma}.
As a consequence, $E$ is never in the spectrum of $ H^\Sigma_L(t)$. This proves the lemma.
\end{proof}

We now focus on the right part. Due to our specific piece-wise linear choices for $w_\Sigma$ and for the cut operator $K(t)$, we obtain that $H^\Sigma_R(t)$ is also piece-wise linear, of the form $H^\Sigma_R(t) = (1 - t) H^\Sigma_R(0) + t H^\Sigma_R(1)$, with
\[
H^\Sigma_R(0) =  \left( \begin{array}{| c c c c }
     \hline
    b & a &  & \\
     a^* & b & a & \\
    & a^* & b & \ddots \\
    & & \ddots & \ddots
    \end{array} \right), 
\qquad \text{and} \qquad
  H^\Sigma_R(1) =  \left( \begin{array}{| c c c c }
        \hline
        \Sigma & 0 &  & \\
        0 & b & a & \\
        & a^* & b & \ddots \\
        & & \ddots & \ddots
    \end{array} \right).
\]
As noted before, the half-line operators are no longer translation equivariant: the presence of the boundary at $0$ makes the model at $t = 1$ not unitarily equivalent to the one at $t = 0$. The key observation is that the relevant part of the spectrum is still periodic, since $ H^\Sigma_R(1)$ is a direct sum of $\Sigma$ with (a shifted version of) $  H^\Sigma_R(0) $. We will study such families of operators in the next section.


\subsection{Dislocated Jacobi operators}

\label{sec:dislocated}

In order to study the right--half Jacobi operator $H^\Sigma_R(t)$, we study yet another {\em dislocated} model. The idea to study this model comes from Hempel and Kohlmann in~\cite{HemKoh-11, HemKoh-11a, HemKoh-12}, and was used in~\cite{Gon-20, Gon-23} to compute spectral flows for Schrödinger operators acting on $L^2(\R)$. In this setting, the dislocated model is modelled by a Schrödinger operator of the form $H_t := -\partial_{xx}^2 + V_t(x)$ with the dislocated potential
\[
    V_t(x) = V(x) \1(x < 0) + V(x - t) \1(x \ge 0),
\]
where $V$ is $1$--periodic. The map $t \mapsto H_t$ is norm--resolvent continuous and $1$-periodic, and we recover periodic models at the endpoints. As mentioned in the introduction, a difficulty with the discrete model is that it is not possible to create such a continuous shift. Our workaround is to introduce a translation equivariant dislocated model $\widetilde{H}^\Sigma(t)$.

\medskip

The operator $\widetilde{H}^\Sigma(t)$ acts on the full space $\ell^2(\Z, \C^N)$ and is defined for $t \in [0, 1]$, by
\[
     \widetilde{H}^\Sigma(t) := 
     \left( \begin{array}{c c c |c c c c c}
        \ddots & \ddots & & & &&\\
        \ddots  & b &  a &  & & & &\\
        &  a^* & b & (1-t) a & ta & & &\\
        \hline
        & & (1 - t) a^* & (1-t)b + t \Sigma & (1-t)a &  & &\\
        & &  ta^* & (1-t)a^* & b & a &   \\
        & & & & a^* & b & \ddots &\\
        & & & & & \ddots & \ddots
    \end{array} \right).
\]
This operator is constructed to interpolate linearly between $\widetilde{H}^\Sigma(t=0)$ and $\widetilde{H}^\Sigma(t=1)$, with
\[
    \widetilde{H}^\Sigma(0) := 
    \left( \begin{array}{c c c |c c c }
        \ddots & \ddots & & & \\
        \ddots  & b &  a &  & &  \\
        &  a^* & b &  a &  & \\
        \hline
        & & a^* & b & a &   \\
        & &  & a^* & b & \ddots   \\
        & & & &  \ddots & \ddots \\
    \end{array} \right), \qquad
    \widetilde{H}^\Sigma(1) := 
    \left( \begin{array}{c c c |c c c c c}
        \ddots & \ddots & & & &&\\
        \ddots  & b &  a &  & & & &\\
        &  a^* & b &  0 & a & & &\\
        \hline
        & & 0 & \Sigma & 0 &  & &\\
        & &  a^* & 0 & b & a &   \\
        & & & & a^* & b & \ddots &\\
        & & & & & \ddots & \ddots
    \end{array} \right).
\]
See also Figure~\ref{fig:ideas}.3. The operator $\widetilde{H}^\Sigma(0)$ is exactly the periodic Jacobi operator $H$, and its spectrum is $\sigma_\bulk$. For the operator $\widetilde{H}^\Sigma(1)$, we note the following. Let us write $\Z = \{0\} \cup (\Z \setminus \{ 0 \})$, and $\ell^2(\Z, \C^N) = \C^N \oplus \ell^2(\Z \setminus \{ 0\}, \C^N)$. With respect to this decomposition, we write
\[
    \forall \Psi \in \ell^2(\Z, \C^N), \qquad \Psi = \{ \Psi_0 \} \oplus \Phi, \quad \text{with} \quad \Phi = ( \cdots, \Psi_{-2}, \Psi_{-1}, \Psi_1, \Psi_2, \cdots).
\]
Then the operator $\widetilde{H}^\Sigma(1)$ acts as
\[
     \forall \Psi = \{ \Psi_0 \} \oplus \Phi \in \ell^2(\Z, \C^N), \qquad \widetilde{H}^\Sigma(1) \left[  \{ \Psi_0 \} \oplus \Phi \right] = \{ \Sigma \Psi_0 \} \oplus H \Phi.
\]
In what follows, we will write 
\begin{equation} \label{eq:widetilde_Sigma}
    \widetilde{H}^\Sigma(1) = \Sigma \, \widetilde{\oplus} \, \widetilde{H}^\Sigma(0),
\end{equation}
for such a decomposition. In particular, we have
\[
    \sigma \left( \widetilde{H}^\Sigma(1) \right) = \sigma \left( H \right) \bigcup \{ \Sigma \} = \sigma_\bulk \bigcup \{ \Sigma \}.
\]
In addition, the extra eigenvalue $\Sigma$ is of multiplicity $N$. Since the energy $E$ is not in the spectrum of $H$, and since $\Sigma > E$, it is not in the spectrum of $\widetilde{H}^\Sigma(0)$ nor $\widetilde{H}^\Sigma(1)$.

\medskip

Finally, since $\widetilde{H}^\Sigma(t)$ is a compact perturbation of $H$ (only a finite number of matrix elements are modified), we have
\[
    \sigma_\ess \left( \widetilde{H}^\Sigma(t)\right) = \sigma_\ess(H) = \sigma_\bulk,
\]
so the essential spectrum is independent of $t \in [0,1]$. Again, some eigenvalues may appear in these essential gaps. Since $E$ is in such an essential gap, the spectral flow $ \Sf \left( \widetilde{H}^\Sigma(t), E, [0, 1] \right)$ is well-defined. Actually, we prove that it equals the one of the right--cut model introduced in the previous section.

\begin{lemma} \label{lem:disloc_to_halfspace}
We have
    \[
        \Sf \left( {H}_R^\Sigma(t), E, [0, 1] \right) = \Sf \left( \widetilde{H}^\Sigma(t), E, [0, 1] \right).
    \]
\end{lemma}
\begin{proof}
We consider yet another cut operator $\widetilde{K}(t)$, defined by $\widetilde{K}(t) = (1 - t) \widetilde{K}(0) + t \widetilde{K}(1)$, with
\[
\widetilde{K}(0) := 
\left( \begin{array}{c c  |c c c }
    \ddots & \ddots &  & \\
    \ddots & 0 & a & 0 & \\
    \hline
    &  a^* & 0 &  \ddots &\\
    &  0 & \ddots & \ddots    \\
\end{array} \right),
\qquad
\widetilde{K}(1) := 
\left( \begin{array}{c c  |c c c }
    \ddots & \ddots &  & \\
    \ddots & 0 & 0 & a & \\
    \hline
    & 0 & 0 &  \ddots &\\
    &  a^* & \ddots & \ddots    \\
\end{array} \right),
\]
where all undisplayed elements are null. Note that $\widetilde{K}(1) = 0 \, \widetilde{\oplus} \, \widetilde{K}(0)$, hence that $(\widetilde{H}^\Sigma - s \widetilde{K})(1) = \Sigma \, \widetilde{\oplus} \, (\widetilde{H}^\Sigma - s \widetilde{K})(0)$ for all $s \in [0, 1]$. Therefore, by robustness of the spectral flow under compact perturbations, we find that (see Lemma~\ref{lem:stability_SF_compact_periodic} below)
\[
\Sf(\widetilde{H}^\Sigma(t), E, [0, 1]) = \Sf(\widetilde{H}^{\Sigma}(t) - \widetilde{K}(t), E, [0, 1]) .
\]
The operator $\widetilde{H}(t) - \widetilde{K}(t)$ is block diagonal, of the form $\widetilde{H}_L^\Sigma(t) \oplus \widetilde{H}_R^\Sigma(t)$. This time, the left part is constant, independent of $t$, hence has a null spectral flow. For the right part, we recover the half Jacobi matrix $H_R^\Sigma(t)$.
\end{proof}

\subsection{Computation of the spectral flow for the dislocated model}
The main result of this section is an exact computation of the spectral flow in the dislocated case.

\begin{lemma} \label{lem:Sf_dislocation}
    We have
    \[
        \Sf \left( \widetilde{H}^\Sigma(t), E, [0, 1] \right) = - \cN(E).
    \]
\end{lemma}
This result was first proved in~\cite{HemKoh-11, HemKoh-11a, HemKoh-12} in the case of Schrödinger operators acting on the continuum. The proof presented here is similar, but more complete, and uses solely tools from spectral theory.

\subsubsection{The finite periodic dislocated model}

First, we introduce a finite dimensional version of our dislocated operator. This is the point in the proof where working with periodic Jacobi operators instead of general periodic Hamiltonians becomes easier to write and study. For $\ell \in \N$, we denote by $\Omega_\ell$ the set of $\ell$ integers centered around $0$. Explicitly,
\[
    \Omega_\ell := \left\{ - \left\lfloor \frac{\ell-1}{2} \right\rfloor, - \left\lfloor \frac{\ell-1}{2} \right\rfloor + 1, \cdots , -1, 0, 1,  \cdots , \left\lfloor \frac{\ell}{2} \right\rfloor - 1, \left\lfloor \frac{\ell}{2} \right\rfloor \right\},
\]
so $\Omega_1 = \{ 0 \}$, $\Omega_2 = \{ 0, 1\}$, $\Omega_3 = \{ -1, 0, 1\}$, and so on. In what follows, we identify $\Omega_\ell$ with the torus $\Z / \ell \Z$, but we chose the labelling of $\Omega_\ell$ to display our matrices. For such $\ell$ with $\ell \ge 5$, we denote by $\widetilde{H}^\Sigma_\ell(t)$ a periodic version of the dislocated Hamiltonian $\widetilde{H}^\Sigma(t)$. This operator acts on $\ell^2(\Omega_\ell, \C^N)$. In the canonical basis, this is the operator (the bars in the matrix still denote the separation between strictly negative and positive indices)
\[
    \widetilde{H}_\ell^\Sigma(t) := 
    \left( \begin{array}{c c c |c c c c c}
        \ddots & \ddots & & & && a^* \\
        \ddots  & b &  a &  & & & &\\
        &  a^* & b & (1-t) a & ta & & &\\
        \hline
        & & (1 - t) a^* & (1-t)b + t \Sigma & (1-t)a &  & &\\
        & &  ta^* & (1-t)a^* & b & a &   \\
        & & & & a^* & b & \ddots &\\
        a & & & & & \ddots & \ddots
    \end{array} \right).
\]
Note the presence of $a^*$ (resp. $a$) in the upper--right (resp. lower--left) corner. Again, this family of matrices is linear in $t$, and interpolates between $\widetilde{H}_\ell^\Sigma(t=0)$ and $\widetilde{H}_\ell^\Sigma(t=1)$, with
\[
\widetilde{H}_\ell^\Sigma(0) := 
\left( \begin{array}{c c c |c c c }
    \ddots & \ddots & & & &a^* \\
    \ddots  & b &  a &  & &  \\
    &  a^* & b &  a &  & \\
    \hline
    & & a^* & b & a &   \\
    & &  & a^* & b & \ddots   \\
    a & & & &  \ddots & \ddots \\
\end{array} \right), \qquad
\widetilde{H}_\ell^\Sigma(1) := 
\left( \begin{array}{c c c |c c c c c}
    \ddots & \ddots & & & && a^*\\
    \ddots  & b &  a &  & & & &\\
    &  a^* & b &  0 & a & & &\\
    \hline
    & & 0 & \Sigma & 0 &  & &\\
    & &  a^* & 0 & b & a &   \\
    & & & & a^* & b & \ddots &\\
    a & & & & & \ddots & \ddots
\end{array} \right).
\]
At $t = 0$, the matrix $\widetilde{H}^\Sigma_\ell(0)$ is a finite version of the periodic Jacobi matrix $H$. We will write $H^\Sigma_\ell(0) =: H_\ell$ below to emphasize this point. Its spectrum can be directly computed using Fourier series. Using plane waves with periodicity $\ell$, we find (compare with~\eqref{eq:spectrum_Hbulk_Jacobi})
\[
    \sigma(H_\ell) = \bigcup_{k \in \Omega_\ell^*} \sigma \left( a^* \re^{- \ri k} + b + a \re^{ \ri k} \right),
\]
with $\Omega_\ell^* := \frac{2\pi}{\ell} \Z \cap \Omega^*$ where we recall that $\Omega^* = (-\pi, \pi]$. In particular, we have $\sigma(H_\ell) \subset \sigma(H)$, and the energy $E$ is not in the spectrum of $H_\ell$. Actually, for each of the $N$ Bloch bands $\{ \lambda_{j,k} \}_{k \in \Omega^*}$ of $H$, we can associate the $\ell$ eigenvalues $\{ \lambda_{j,k} \}_{k \in \Omega^*_\ell}$ of $H_\ell$.

\medskip

At $t = 1$, apart from the presence $\Sigma$ in the middle, we also recognize a periodic model, and we have a relation of the form (compare with~\eqref{eq:widetilde_Sigma})
\[
    \widetilde{H}_\ell^\Sigma(1) = \Sigma \, \widetilde{\oplus} \, H_{\ell-1}, 
    \qquad \text{hence} \qquad 
    \sigma \left( \widetilde{H}_\ell^\Sigma(1) \right) = \sigma \left( H_{\ell-1} \right) \bigcup \{ \Sigma \},
\]
where $\Sigma$ is an eigenvalue of multiplicity $N$.

\begin{lemma} \label{lem:Sf_finite_dislocated}
    For all $\ell \ge 5$, we have $\Sf \left( \widetilde{H}^\Sigma_\ell(t), E, [0, 1]\right) = - \cN(E)$.
\end{lemma}

\begin{proof}
The map $t \mapsto \widetilde{H}^\Sigma_L(t)$ is continuous (even linear) in $t$. So the branches of eigenvalues of these matrices are also continuous in $t$. At $t = 0$, there are $\ell \cN(E)$ eigenvalues of $\widetilde{H}_\ell^\Sigma(0)$ below the energy $E$, and at $t = 1$, there are $(\ell-1) \cN(E)$ eigenvalues of $\widetilde{H}_\ell^\Sigma(1)$ below $E$. This difference implies that the net number of eigenvalues crossing $E$ when $t$ increases to $1$ equals $-\cN (E)$
\end{proof}

Let us recall our analogy of the Grand Hilbert Hotel given in Remark~\ref{rem:HilbertHotel}. In this finite periodic case, there are only a finite number of rooms per floor. The difference between the initial and final number of rooms counts, without ambiguity, the number of eigenvalues that cross from one band to the next. 


\subsubsection{From the finite dislocated model to the infinite dislocated model}
\label{sec:finite_to_infinite}

In the previous section, we proved that the spectral flow of the supercell dislocated model $\widetilde{H}^\Sigma_\ell$ is $-\cN(E)$ for all $\ell \ge 5$.  In this section, we justify that one may take the limit $\ell \to \infty$ and deduce that spectral flow is also $-\cN(E)$ for the full dislocated model $\widetilde{H}^\Sigma$.
Throughout this subsection, we drop the superscript $\Sigma$, and the {\em tilde} notation for clarity: we will simply write $H(t)$ and $H_\ell(t)$ for the dislocated models acting on $\ell^2(\Z, \C^N)$ and $\ell^2(\Omega_\ell, \C^N)$ respectively. Recall that at $t = 0$, the operators $H(0)$ and $H_\ell(0)$ are periodic, and that
\[
    \sigma (H_\ell(0)) \subset \sigma_\bulk, \quad \sigma (H(0)) = \sigma_\bulk.
\]
\medskip

In order to relate the operators $H_\ell$ and $H$ that act on different Hilbert spaces, we introduce $i_\ell : \cH_\ell \to \cH$ for the extension by zero. Explicitly,
\[
    \forall \Psi \in \cH_\ell, \ \forall n \in \Z, \quad (i_\ell \Psi)_n = \begin{cases}
        \Psi_n & \quad \text{if} \quad n \in \Omega_\ell, \\
        0 & \quad \text{else}
    \end{cases},
\]
whose adjoint is the restriction operator $i_\ell^* : \cH \to \cH_\ell$ given by
\[
    \forall \Psi \in \cH, \ \forall n \in \Omega_\ell, \quad (i_\ell^* \Psi)_n = \Psi_n.
\]

We also introduce a smooth cutoff function $\chi: \R \mapsto [0,1]$ such that 
\[
\chi(x) = 1 \quad \forall x \in [-1/4,1/4], \qquad \chi(x) = 0 \quad \forall x \in \R \setminus (-1/3,1/3),
\]
and define $\chi_\ell$ as the multiplication operator by $\chi(\cdot/\ell)$ in $\cH_\ell$. Note that this function is supported in $[-\ell/3, \ell/3]\subset \Omega_\ell$. We will use the following identities throughout this section. 

\begin{lemma} \label{lem:estimates}
    For all $\ell \ge 12$ and $t \in [0,1]$ we have,
    \begin{enumerate}
        \item $ H (t) i_\ell \chi_\ell = i_\ell  H_\ell (t)\chi_\ell$ from $\cH_\ell$ to $\cH$, and $\chi_\ell i_\ell^* H(t) = \chi_\ell  H_\ell(t) i_\ell^*$ from $\cH$ to $\cH_\ell$,
        \item For any multiplication operator $\phi$ on $\cH_\ell$ with support in $\Omega_\ell \setminus [-\ell/4, \ell/4]$, we have $H_\ell(t) \phi = H_\ell(0) \phi$ as operators from $\cH_\ell$ to $\cH_\ell$. \\
        Similarly, For any multiplication operator $\widetilde \phi$ on $\cH$ with support in $\Z \setminus [-\ell/4, \ell/4]$, we have $H(t)\widetilde \phi = H(0)\widetilde \phi$ as operators from $\cH$ to $\cH$.
        \item The operator $R_\ell := \left[ H_\ell, \chi_\ell \right] = H_\ell(t) \chi_\ell  - \chi_\ell  H_\ell(t)$ acting on $\cH_\ell$ is independent of $t$, and satisfies $\| R_\ell \|_{\op, \ell} \le \frac{C}{\ell}$. 
    \end{enumerate}
\end{lemma}

The first point states that we can replace ${H}(t)$ with ${H}_\ell(t)$ in the support of $\chi_\ell$ (the cut-off function erases the boundary conditions). The second point states that outside the support of the cut-off function, the operators are independent of $t \in [0, 1]$. The last point says that one can commute the Hamiltonians with the cut-off operator, up to an error of order $O(\ell^{-1})$.

\begin{proof}
For the first point, we note that for all $\Psi \in \cH_\ell$, the function $\chi_\ell \Psi$ has support in $[-\ell/3, \ell/3]$. Whenever $\ell \ge 9$, this implies that  $H_\ell(t) \chi_\ell \Psi$ has support in $[-\ell/3 - 1, \ell/3 + 1]$, hence  $H_\ell(t) \chi_\ell \Psi = H(t) i_\ell \chi_\ell \Psi$.
The second relation is the adjoint of the first.
\medskip

For the second point, we note that $H_\ell(t) = H_\ell(0) + M_\ell(t)$, where $M_\ell(t) := i_\ell M(t) i_\ell^*$, and where $M(t)$ is a finite rank operator with nonzero matrix elements only in the entries $\{ -1, 0, 1\}$. So $M_\ell(t) \phi = 0$ whenever $\ell > 4$, since we assumed that $\phi= 0$ in $\Omega_\ell \cap [-\ell/4, \ell/4] \supset \{ -1, 0, 1\}$. The argument for $\widetilde \phi$ is identical.

\medskip

For the third point, since $\chi_\ell$ is equal to the identity in the support of $M$ for all $\ell > 8$, we have
 $M_\ell(t) \chi_\ell = \chi_\ell M_\ell(t)$. In particular, $[ H_\ell(t), \chi_\ell] = [H_\ell(0), \chi_\ell]$ is independent of $t$. Next, we find that
\[
    \left[ H_\ell(0) \chi_\ell - \chi_\ell H_\ell(0) \right]_{ij} = \begin{cases}
        a \ (\chi_\ell(n+1) - \chi_\ell(n)) & \quad \text{if} \quad (i,j) = (n, n+1) \\
        a^* (\chi_\ell(n-1) - \chi_\ell(n)) & \quad \text{if} \quad (i,j) = (n, n-1) \\
        0 & \quad \text{else}
    \end{cases}.
\]
Note that $ | \chi_\ell(n+1) - \chi_\ell(n) | \le \frac{C}{\ell}$. So
\[
    R_\ell := H_\ell(t) \chi_\ell - \chi_\ell H_\ell(t) = H_\ell(0) \chi_\ell - \chi_\ell H_\ell(0) \quad \text{satisfies} \quad \| R_\ell \|_{\op, \ell} \le 2 \| a \|_{\rm op} \frac{C}{\ell}. \qedhere
\]
\end{proof}

First, we prove that any complex number in the resolvent set of $H(t)$ is in the resolvent set of $H_\ell(t)$ for sufficiently large values of $\ell$.
\begin{lemma} \label{lem:resolvent-inclusion}
    There is a constant $C$, such that, for all $t \in [0,1]$ and all $\lambda \in \C$ satisfying
    \[
    \dist(\lambda, \sigma({H}(t))) := \varepsilon  > 0,
    \]
    we have     
    \[
     \dist(\lambda, \sigma({H}_\ell(t))) \ge \frac{\varepsilon}{4} \text{ for all } \ell \ge \max( 12, C/\varepsilon).
    \]
\end{lemma}

Recall that, for a self-adjoint operator $A$, we have $\dist(\lambda, \sigma(A)) = \| ( \lambda - A)^{-1} \|_{\rm op}^{-1}$.
\begin{proof}
Since the essential spectrum of $H(t)$ coincides with $\sigma_\bulk$, we find that $\dist(\lambda, \sigma_\bulk) \ge \varepsilon$. 
    Recall that $\sigma(H_\ell(0)) \subset \sigma_\bulk$, so in particular $H_\ell(0) - \lambda$ is invertible with $\|(H_\ell(0) - \lambda)^{-1}\|_{\rm op} \le \varepsilon^{-1}$. The same bound holds for $(H(t) - \lambda)^{-1}$. 
    We combine both resolvents to construct an approximate inverse for $H_\ell(t) -\lambda$.
    Indeed, for any $\ell \ge 9$, the identities from Lemma~\ref{lem:estimates} give
    \begin{equation} \label{eq:resolvent-comutator-base}
        \left( H_\ell(t) - z \right) \chi_\ell i_\ell^* = \chi_\ell \left(H_\ell(t) - z \right) i_\ell^* + R_\ell i_\ell^* = \chi_\ell i_\ell^* \left( H(t) - z \right) + R_\ell i_\ell^*
    \end{equation}
    as operators from $\cH$ to $\cH_\ell$. Together with the fact that $i_\ell^* i_\ell = \1_\ell := \1_{\cH_\ell}$, we get the following identity, as operators from $\cH_\ell$ to $\cH_\ell$,
    \begin{align*}
       & ( H_\ell(t) -\lambda) \left\{ \chi_\ell i_\ell^*( H(t) - \lambda)^{-1} i_\ell + (\1_\ell - \chi_\ell) (H_\ell(0) - \lambda)^{-1} \right\} \\
       & \quad = \chi_\ell i_\ell^* i_\ell + R_\ell i_\ell^* (H(t)- \lambda)^{-1} i_\ell + (\1_\ell - \chi_\ell) - R_\ell (H_\ell(0) - \lambda)^{-1} \\
       & \quad = \1_\ell + R_\ell \left\{ i_\ell^* (H(t) - \lambda)^{-1} i_\ell -  (H_\ell(0) - \lambda)^{-1}  \right\} := \1_\ell + \mathcal{R}_\ell(t).
    \end{align*}
  We can bound $\mathcal{R}_\ell(t)$ by combining the bounds on the resolvents and on $R_\ell$ from Lemma~\ref{lem:estimates} to obtain
    \[
        \left\|\mathcal{R}_\ell(t) \right\|_{\op, \ell} \le \frac{2 C}{\varepsilon \ell},
    \]
    Note that the bound is independent of $t \in [0, 1]$. Then, for $\ell \ge C':= \max \{ 12, \frac{4 C}{\varepsilon} \}$, we have $\| \cR_\ell(t) \|_\op \le \frac{1}{2}$. We conclude that $H_\ell(t) -\lambda$ is invertible with
        \begin{align*}
        \|( H_\ell(t) -\lambda)^{-1}\|_{\op , \ell}
        &\le \left\| \chi_\ell i_\ell^* ( H(t) - \lambda)^{-1} i_\ell + (\1_\ell - \chi_\ell)(H_\ell(0) - \lambda)^{-1} \right\|_{\op, \ell} \, \|(\1_\ell + \mathcal{R}_\ell(t) )^{-1}\|_{\op, \ell} \\
       & \quad \le \frac{2}{\varepsilon}\, \frac{1}{1 - \norm{\cR_\ell(t)}_\op} \le \frac{4}{\varepsilon}.
    \end{align*}
    This is the desired result.
\end{proof}

To complete the proof, we have to show that the spectral projectors of $H_\ell(t)$ in the gaps of $\sigma_\bulk$ converge to spectral projectors of $H(t)$. We record a convenient property of projections.
    \begin{lemma} \label{lem:projections_base}
        Let $P$ be a projector of rank $m$, and $A$ be a bounded operator such that $\| P - A \|_\op < 1$. Then $\rank (A) \ge m$.
    \end{lemma}
\begin{proof}
    We have $\| P - P A P \|_\op < 1$, so, seen as maps from $E := \Ran(P)$ to itself, we have $\| \bbI_m - A |_E \| < 1$, which proves that $A |_E$ is full rank on $E$. So $\rank( A) \ge m$.
\end{proof}

In what follows, for $\lambda_1 < \lambda_2$, we denote by $P_{(\lambda_1, \lambda_2)}(A)$ the spectral projector of the self-adjoint operator $A$ on the open interval $(\lambda_1, \lambda_2)$. We recall that if $\lambda_1$ and $\lambda_2$ are not in the spectrum of $A$, then this spectral projector can be written as the Cauchy contour integral
\[
     P_{(\lambda_1, \lambda_2)}(A) = \dfrac{1}{2 \ri \pi} \oint_{\sC} \dfrac{\rd z}{z - A},
\]
where $\sC$ is any simple, positively oriented contour in $\C$ enclosing the interval $(\lambda_1, \lambda_2)$ and no other portions of the real line. We can take for instance $\sC$ the positively oriented circle with center $\frac12(\lambda_1 + \lambda_2)$ and radius $\frac12 (\lambda_2 - \lambda_1)$.

\begin{lemma} \label{lem:projections} Fix $t$ and let $\lambda_1 < \lambda_2$ which are not eigenvalues of $H(t)$. Then, for $\ell$ large enough, we have 
    $$\rank \left(P_{(\lambda_1, \lambda_2)}(H_\ell(t))\right) = \rank  \left(P_{(\lambda_1,\lambda_2)}(H(t))\right).$$
\end{lemma}
\begin{proof} 

As in the proof of Lemma~\ref{lem:resolvent-inclusion}, the idea is to localize the resolvent of $H_\ell(t)$ using the cut-off functions $\chi_\ell$. On the support of $\chi_\ell$, we compare it with the resolvent of $H(t)$, and outside this support, we compare it with $H_\ell(0)$ (independent of $t$). Since the interval $(\lambda_1,\lambda_2)$ is contained in a gap of $\sigma_\bulk$, the latter part does not contribute to the spectral projection.

\medskip

To start, let $\sC$ be the positively oriented circle as above, which encloses the interval $(\lambda_1, \lambda_2)$. Note that
\[
    \forall z \in \sC, \quad \dist (z, \sigma(H(t))) \ge \max \{ \dist( \lambda_1, \sigma(H(t))), \dist( \lambda_2, \sigma(H(t))) \} =: \varepsilon > 0.
\]
So, by Lemma~\ref{lem:resolvent-inclusion}, for $\ell \ge \frac{C}{\varepsilon}$, this contour is also included in the resolvent set of $H_\ell(t)$, and we have the bound
\[
    \forall z \in \sC, \quad \| (z - H_\ell(t))^{-1} \|_{\rm op} \le \frac{4}{\varepsilon}, \quad \text{and} \quad \| (z - H(t))^{-1} \|_{\rm op} \le \frac{1}{\varepsilon}.
\]

We use the localization function $\chi_\ell$ and $\phi_\ell:= \sqrt{1-\chi_\ell^2}$, such that $\chi_\ell^2 + \phi_\ell^2 =1$. The support of $\chi_\ell$ is in $[-\ell/3, \ell/3]$, the support of $\phi_\ell$ is in the complement of $(-\ell/4, \ell/4)$. The third point of Lemma~\ref{lem:estimates} (for $\chi_\ell$, but the proof is similar for $\phi_\ell$) shows that there is a constant $C$ so that
\begin{equation} \label{eq:bound_R_ell_phi}
    R_\ell :=  [H_\ell(t), \chi_\ell] \quad \text{and} \quad R_\ell^\phi :=  [H_\ell(t), \phi_\ell] \quad \text{satisfy} \quad
    \left\| R_\ell \right\|_{\op, \ell} \le \frac{C}{\ell}, \quad  \left\| R_\ell^\phi \right\|_{\op, \ell} \le \frac{C}{\ell}.
\end{equation}
Consider $z \in \sC$, and recall~\eqref{eq:resolvent-comutator-base} which states that
\[
    \chi_\ell i_\ell^* (H(t) -z) = ( H_\ell(t) -z) \chi_\ell i_\ell^* - R_\ell.
\]
Composing this identity with $(H_\ell(t) -z)^{-1}$ on the left and $( H(t) - z)^{-1}i _\ell \chi_\ell $ on the right gives
\begin{equation}\label{eq:resolvent-comutator}
 ( H_\ell(t) -z)^{-1} \chi_\ell^2 = \chi_\ell i_\ell^*( H(t) - z)^{-1} i_\ell \chi_\ell+ \cR^{(1)}_\ell(z),
\end{equation}
with $\|\cR^{(1)}_\ell (z)\|_{\op, \ell} \le \frac{8C}{\varepsilon^2 \ell}$.
On the other hand, we have $H_\ell(t) \phi_\ell = H_\ell(0) \phi_\ell$ as operators on $\cH_\ell$ (see Lemma~\ref{lem:estimates}(2)), hence, taking adjoints, $\phi_\ell H_\ell(t) = \phi_\ell H_\ell(0)$. Writing
\[
   \phi_\ell(H_\ell(0) - z) =  \phi_\ell (H_\ell(t) - z)  = (H_\ell(t) -z) \phi_\ell -  R_\ell^{\phi}
\]
and multiplying by $(H_\ell(t) -z)^{-1}$ on the left and $(H_\ell(0) - z)^{-1}\phi_\ell$ on the right gives
\begin{equation}\label{eq:resolvent-comutator-outside}
   (H_\ell(t) -z)^{-1} \phi_\ell^2 = \phi_\ell ( H_\ell(0) -z)^{-1} \phi_\ell + \cR^{(2)}_\ell(z),
\end{equation}
again with $\|\cR^{(2)}_\ell (z)\|_{\op} \le \frac{8C}{\varepsilon^2 \ell}$. Summing \eqref{eq:resolvent-comutator} and \eqref{eq:resolvent-comutator-outside}
gives
\begin{align*}
( H_\ell(t) -z)^{-1} 
&= ( H_\ell(t) -z)^{-1} (\chi_\ell^2 + \phi_\ell^2) \\
&= \chi_\ell i_\ell^*( H(t) - z)^{-1} i_\ell \chi_\ell   + \phi_\ell (H_\ell(0) -z)^{-1}\phi_\ell+ (\cR_\ell^{(1)} +\cR_\ell^{(2)})(z).
\end{align*}
Inserting this identity in the Cauchy formula gives
\begin{align*}
 P_{(\lambda_1, \lambda_2)}( H_\ell(t) )  = \chi_\ell i_\ell^* P_{(\lambda_1, \lambda_2)}( H(t) ) i_\ell \chi_\ell + \phi_\ell P_{(\lambda_1, \lambda_2)}( H_\ell (0) ) \phi_\ell + \frac{1}{2 \ri \pi  }\oint_\sC (\cR_\ell^{(1)} +\cR_\ell^{(2)})(z) \rd z.
\end{align*}

Since $(\lambda_1, \lambda_2) $ lies in a gap of $\sigma_\bulk$, the second term is actually zero. The last term can be bounded by $4 C |\sC| / (\pi \varepsilon^2 \ell)$. Thus, we can find $\ell_0$ sufficiently large such that for all $\ell \ge \ell_0$, we have
\[
\left\| P_{(\lambda_1, \lambda_2)}( H_\ell (t) ) - \chi_\ell i_\ell^* P_{(\lambda_1, \lambda_2)}( H(t) ) i_\ell \chi_\ell \right\|_\op < 1.
\]
From Lemma~\ref{lem:projections_base}, we conclude that
\[
\rank(P_{(\lambda_1, \lambda_2)}( H(t))) 
\ge \rank\left( \chi_\ell P_{(\lambda_1, \lambda_2)}( H(t)) \chi_\ell \right) 
\ge \rank\left(  P_{(\lambda_1, \lambda_2)}( H_\ell(t) )  \right).
\]

 For the opposite inequality, we have to define the cut-off function $\widetilde{\phi}_\ell$ such that
 $\widetilde{\phi}_\ell^2 + i_\ell \chi_\ell^2 i^*_\ell = \1_\cH$.
 Multiplying \eqref{eq:resolvent-comutator-base} on the left by $i_\ell \chi_\ell (H_\ell(t) - z)^{-1}$ and on the right by $(H(t)- z)^{-1}$
 gives the identity
 $$
 i_\ell \chi_\ell^2 i_\ell^* (H(t)- z)^{-1} = i_\ell \chi_\ell (H_\ell(t) - z)^{-1} \chi_\ell i_\ell^* + \widetilde{\cR}^{(1)}_\ell(z),
 $$
 with $\| \widetilde{\cR}^{(1)}_\ell (z)\|_{\op, \ell} \le \frac{8C}{\varepsilon^2 \ell}$. The analogue of \eqref{eq:resolvent-comutator-outside} reads 
 \begin{equation*}
\widetilde{\phi}_\ell^2( H(t) - z)^{-1}  = \widetilde{\phi}_\ell ( H(0) -z)^{-1} \widetilde{\phi}_\ell + \widetilde{\cR}^{(2)}_\ell(z).
\end{equation*}
with  $\| \widetilde{\cR}^{(2)}_\ell (z)\|_{\op, \ell} \le \frac{8C}{\varepsilon^2 \ell}$. We then apply the Cauchy formula to
 \begin{align*}
 ( H(t) -z)^{-1} 
 &= i_\ell \chi_\ell ( H_\ell(t) - z)^{-1} \chi_\ell i_\ell^* +  \widetilde{\phi}_\ell (H(0) -z)^{-1} \widetilde{\phi}_\ell + (\widetilde{\cR}_\ell^{(1)} +\widetilde{\cR}_\ell^{(2)})(z),
 \end{align*}
and note that the second term vanishes as before.
Thus, Lemma~\ref{lem:projections_base} gives for all sufficiently large values of $\ell$
\[
\rank(P_{(\lambda_1, \lambda_2)}( H_\ell(t) )) 
\ge \rank\left( i_\ell \chi_\ell P_{(\lambda_1, \lambda_2)}( H_\ell(t) ) \chi_\ell i_\ell^* \right) 
\ge \rank\left(  P_{(\lambda_1, \lambda_2)}( H(t) )  \right),
\]
and completes the proof. 
\end{proof}

We now have all the ingredients to complete the proof of Lemma~\ref{lem:Sf_dislocation} on the spectral flow of $t \mapsto H(t)$.

\begin{proof}[Proof of Lemma~\ref{lem:Sf_dislocation}]

Choose a partition $0 = t_0 \le t_1 \le \cdots \le t_M = 1$ and positive numbers $(a_i)_{1 \le i \le M}$ so that, for all $1 \le i \le M$, the projection $P_{(E - a_i, E + a_i)}\left( H (t) \right)$ is finite rank and $E\pm a_i$ belong to the resolvent set of $H(t)$ for all $t \in [t_{i-1}, t_i]$. We recall in Appendix~\ref{sec:appendix:SF} that the spectral flow of $t \mapsto {H}(t)$ is given by
\begin{align*}
     \Sf \left( {H}(t), E, [0, 1] \right)
     = \sum_{i=1}^{M-1} \rank \, P_{(E + a_i, E + a_{i+1})}( H (t_i) ),
\end{align*}
with the convention that $\rank \, P_{(\lambda_1, \lambda_2)}(A) = - \rank \, P_{(\lambda_2, \lambda_1)}(A)$ if $\lambda_2 < \lambda_1$. We used here that
\[
    \dim \, \Ker (H(1) - E)  = \dim \, \Ker (H(0) - E) ,
\]
which comes from the fact that $H(1) = \Sigma \, \widetilde{\oplus} \, H(0) $, so that there are no boundary terms

\medskip

The hypothesis on $E\pm a_i$ and the continuity of the eigenvalues as a function of $t$, shows that there is $\varepsilon_i>0$ such that 
$$
\forall t \in [t_i, t_{i+1}], \quad \dist(E \pm a_i, \sigma(H (t))) \ge \varepsilon_i.
$$
By Lemma~\ref{lem:resolvent-inclusion}, with $\varepsilon = \min_i \varepsilon_i$ there is $\ell_0$ such that, for all $\ell\ge \ell_0$, we have 
$$
\forall t \in [t_i, t_{i+1}], \quad \dist(E \pm a_i, \sigma(H_\ell(t))) \ge \frac{\varepsilon}{4}.
$$
Thus, the partition given by $\{t_i\}_{i=0, \cdots, M}$ and the corresponding values $a_i$ are suitable to compute the spectral flow of $H_\ell$ for all  $\ell >\ell_0$.
By applying Lemma~\ref{lem:projections}, we find a sufficiently large $\ell$, such that 
$$
\Ran \left(P_{(E+a_i, E+a_{i+1})} \left( H (t_i)\right) \right)
= \Ran \left(P_{(E+a_i, E+a_{i+1})} \left( H_\ell (t_i)\right) \right) \text{ for } i = 1, \cdots, M-1,
$$
and thus, by using Lemma~\ref{lem:Sf_finite_dislocated},
$$
 \Sf \left( {H} (t), E, [0, 1] \right) =  \Sf \left( {H}_\ell (t), E, [0, 1] \right) = -\cN(E). \qedhere
$$
\end{proof}

\section{Application for two--dimensional materials}
    \label{sec:2d}
    
    In this section, we apply the previous theory to the case of two--dimensional materials in the tight--binding approximation. Our goal is to describe the spectrum of such materials in the presence of a soft wall. In the one-dimensional case, we could without loss of generality fix the unit cell to be $[0,1)$ and the corresponding Fourier variable belongs to $[-\pi, \pi)$. To describe general two-dimensional models, we need the full terminology of solid state physics.

    \subsection{Reduction to the one-dimensional case}
    
    \subsubsection{Bravais lattices and bulk operators in tight-binding models}
    \label{sec:2d_bulk}
        Let us first fix some notation. We consider a $\Lat$--periodic crystal, where $\Lat$ is a Bravais lattice in $\R^2$, of the form
    \[
    \Lat = \ba_1 \Z \oplus \ba_2 \Z, \qquad \ba_1, \ba_2 \in \R^2, \quad \det(\ba_1, \ba_2) \neq 0.
    \]
    We denote by $\Gamma$ a fundamental cell of the lattice, {\em i.e.} a subset of $\R^2$ such that the disjoint union of its $\Lat$--translations cover $\R^2$: $\bigsqcup_{\bR \in \Lat} \Gamma + \bR = \R^2$. A typical choice is $\Gamma = \ba_1 [0, 1) \oplus \ba_2 [0, 1)$, but any choice works. We denote by $M \in \N$ the number of atoms in each unit cell $\Gamma$. The location of these atoms are $(\bx_1, \cdots, \bx_M)$ with $\bx_j \in \Gamma$. 
    
    \medskip
    
    The wave-function is an element $\Psi \in \ell^2(\Lat, \C^M)$ parametrized as
    \[
        \Psi = \left( \psi_\bR \right)_{\bR \in \Lat}, \quad \text{where} \quad \forall \bR \in \Lat, \quad \psi_{\bR} = \begin{pmatrix}
        \psi^{(1)}_{\bR} & \psi^{(2)}_{\bR}  & \cdots & \psi^{(M)}_{\bR} 
    \end{pmatrix}^T \quad \in \C^M.
    \]
    As noted in the introduction, one interpretation of this notation is that $\psi_{\bR}^{(m)}\in \C$ is the amplitude of the electron on the atom located at $\bR + \bx_m$. The discrete translation invariance of the perfect, infinite crystal translates in the convolution form of the two--dimensional tight--binding Hamiltonian $\ssH$ (we use straight letters for two--dimensional operators), of the form
    \begin{equation} \label{eq:def:A}
     \forall \Psi \in \ell^2(\Lat, \C^M),  \quad \forall \bR \in \Lat, \quad 
     (\ssH \Psi)_{\bR} = \sum_{\bR' \in \Lat} \ssh(\bR') \Psi_{\bR - \bR'},
    \end{equation}
    where $\ssh : \Lat \to \cM_M(\C)$ is a family of $M \times M$ matrices satisfying $h(- \bR) = h(\bR)^*$. We assume that $h \in \ell^1(\Lat, \C^{M \times M})$, which implies as before that $\ssH$ is a bounded operator on $\ell^2(\Lat, \C^M)$.
  
  \medskip
  
    Since $\sf H$ is a convolution operator, it is diagonal in Fourier basis. 
As usual in solid state physics, we introduce a pair of reciprocal lattice vectors $\ba_1^*$, $\ba_2^*$, that are defined  by the relations $\bra \ba_i^*, \ba_j \ket_{\R^2} = 2 \pi \delta_{ij}$. These vectors generate the reciprocal lattice $\RLat$, and the natural domain for the Fourier variable is the Brillouin zone $\Gamma^* := \R^2 / \RLat$ (seen as a torus). Reasoning as in Section \ref{ssec:periodic_bands}, we find that $\cF \ssH \cF^* = \int_{\Gamma^*}^\oplus {\sf H}_\bk \rd \bk$, where 
    \begin{equation} \label{eq:def:Ak_2d}
        \forall \bk \in \R^2, \qquad {\sf H}_\bk 
        = \sum_{\bR \in \Lat} \ssh(\bR) \re^{-\ri \bR \cdot \bk} \qquad \in \cS_{M}(\C).
    \end{equation}
    The map $\bk \mapsto \ssH_\bk$ is analytic (it is the sum of a uniformly absolutely convergent series of analytic functions) and $\RLat$--periodic, so it is enough to describe ${\sf H}_\bk$ on the Brillouin zone $\bk \in \Gamma^*$. We denote by $\lambda_{1, \bk} \le \cdots \le \lambda_{M, \bk}$ the eigenvalues of $\ssH_\bk$, ranked in increasing order. The $j$-th (two--dimensional) {\bf Bloch band} is the interval $\bigcup_{\bk \in \R^2} \{ \lambda_{j, \bk} \}$, and the spectrum of $\ssH$ is the union of these $M$ Bloch bands:
    \[
        \sigma_\bulk = \sigma \left( {\sf H} \right) = \bigcup_{\bk \in \R^2} \sigma (\ssH_\bk) = 
                \bigcup_{ j =1}^M \bigcup_{\bk \in \Gamma^*} \{ \lambda_{j, \bk} \}.
    \]
    
    In the next section, we will add a wall which is constant along the $\ba_2$--direction (see Section~\ref{sec:general_angles} below for the case of general angles). This implies that the model with the wall is still periodic in the $\ba_2$--direction. In particular, we may perform a partial Fourier transform in this direction. It is useful at this point to introduce the one--dimensional lattices
    \begin{align*}
        \Lat_1 := \ba_1 \Z, \quad \Lat_2 := \ba_2 \Z, & \qquad \text{so that} \qquad \Lat = \Lat_1 \oplus \Lat_2, \\
        \RLat_1 := \ba_1^* \Z, \quad \RLat_2 := \ba_2^* \Z, & \qquad \text{so that} \qquad \RLat = \RLat_1 \oplus \RLat_2.
    \end{align*}
    We also denote by $\Gamma^*_2 := (\ba_2^* \R) / \RLat_2$ the Brillouin zone of the $\Lat_2$ lattice, seen as a torus/line in the two--dimensional space $\R^2$, in the sense that we keep the orientation of this line along the $\ba_2$--direction. We write $\bk_2 \in \Gamma_2^*$ for the Fourier variable in this direction. In the literature, it is often identified with the one--dimensional real torus/line, with the identification $\bk_2 = k_2 \frac{\ba_2^*}{| \ba_2^* | }$ where $k_2 = | \bk_2 | \in \R / | \ba_2^* | \Z$. 
    
    With these conventions, the partial Fourier transform is $\cF_2 : \ell^2(\Lat, \C^M) \to L^2(\Lat_1 \times \Gamma^*_2, \C^M)$ defined by
    \begin{equation} \label{eq:partial_Fourier}
        \forall (\bR_1, \bk_2) \in \Lat_1 \times \Gamma^*_2, \qquad
    \left( \cF_2[\Psi] \right)(\bR_1, \bk_2) := \frac{1}{\sqrt{| \Gamma_2 |}} \sum_{\bR_2 \in \Lat_2} \Psi_{\bR_1 + \bR_2} \re^{- \ri \bk_2 \cdot \bR_2}.
    \end{equation}
    We obtain that $\cF_2 \ssH \cF_2^* = \int_{\Gamma^*_2}^\oplus H_{\bk_2} \rd \bk_2$, where, for all $\bk_2 \in \Gamma^*_2$, $H_{\bk_2}$ is an operator acting on the one--dimensional lattice $\ell^2(\Lat_1, \C^M)$, of the form
    \[
    \forall \Psi \in \ell^2 (\Lat_1, \C^M), \quad \forall \bR_1 \in \Lat_1, \qquad
    \left( H_{\bk_2} \Psi \right)_{\bR_1} = \sum_{\bR_1' \in \Lat_1} h_{\bk_2}(\bR_1') \Psi_{\bR_1 - \bR_1'},
    \]
    where the kernel $h_{\bk_2}$ is 
    \begin{equation} \label{eq:def:hk2}
        \forall \bR_1 \in \Lat_1, \qquad h_{\bk_2}(\bR_1) := \sum_{\bR_2 \in \Lat_2} \ssh(\bR_1 + \bR_2) \re^{-\ri \bk_2 \cdot \bR_2}.
    \end{equation}
Note that the series is convergent and the resulting one-dimensional kernel satisfies the hypothesis of the previous section, since $\ssh \in \ell^1( \Lat)$ implies $h_{\bk_2} \in \ell^1 (\Lat_1) $ by Fubini's Theorem.  
    
  \medskip  
          
    In what follows, we fix $\bk_2 \in \Gamma_2^*$, and study the one-dimensional periodic operator $H_{\bk_2}$, and its soft wall counterpart $H_{\bk_2}^\sharp(t)$. The spectrum of $H_{\bk_2}$ is purely essential, composed of $M$ (one-dimensional) Bloch bands. It is the union of the spectra of $\left(H_{\bk_2} \right)_{\bk_1}$, defined as (compare with~\eqref{eq:def:Hk})
        \[
      \forall \bk_1 \in \Gamma_1^*, \quad   \left(H_{\bk_2} \right)_{\bk_1} := \sum_{\bR_1 \in \Lat_1} h_{\bk_2}(\bR_1) \re^{- \ri \bk_1 \cdot \bR_1}.
    \]
     Together with~\eqref{eq:def:hk2}, we get that 
    \[
       \left(H_{\bk_2} \right)_{\bk_1} = \sum_{\bR_1 \in \Lat_1} \sum_{\bR_2 \in \Lat_2} \ssh(\bR_1 + \bR_2) \re^{- \ri \bk_2 \cdot \bR_2 - \ri \bk_1 \cdot \bR_1} = \sum_{\bR \in \Lat} \ssh(\bR) \re^{- \ri \bk \cdot \bR} = \ssH_{\bk_1 + \bk_2}.
    \]
    We deduce that
    \begin{equation} \label{eq:spectrum_Hk2}
    \sigma \left(   H_{\bk_2} \right) = \bigcup_{\bk_1 \in \Gamma^*_1} \sigma \left( h_{\bk_2}(\bk_1) \right) = \bigcup_{\bk_1 \in \Gamma^*_1} \sigma \left( \ssH_{\bk_1 + \bk_2} \right) 
    =  \bigcup_{\bk_1 \in \ba_1^* \R} \sigma \left( \ssH_{\bk_1 + \bk_2} \right).
    \end{equation}
    One should think of this formula as follows: the spectrum of $H_{\bk_2}$ is the \emph{projection} of the two--dimensional bands (the spectra of $\ssH_{\bk}$) on the $\Gamma^*_2 \approx \ba_2^* \R$ line, where the projection is along the $\ba_1^*$--direction. Note that $\ba_1^*$ is the direction orthogonal to the $\ba_2$ vector (which will be the direction of the wall). We extend this formula to the case of general angles below.

\subsubsection{Soft walls parallel to a lattice vector}
\label{sec:2d-softwall}

We now add a soft wall potential. We will choose our wall aligned with the $\ba_2$ vector (see Section~\ref{sec:general_angles} below for other \emph{conmensurate} directions). To do so, we introduce the normalized vector $\ba_2^\perp := \tfrac{\ba_1^*}{| \ba_1^* |}$. This vector satisfies $ \ba_2^\perp \cdot \ba_2  = 0$ and $\ba_2^\perp \cdot \ba_1  = \frac{2 \pi}{| \ba_1^* |}$. We consider a two--dimensional soft wall potential $w : \R^2 \to \cS_M(\C)$ of the form
\[
\forall \bx \in \R^2, \quad w(\bx) = w_*(\ba_2^\perp \cdot \bx),
\]
where $w_* : \R \mapsto \cS_M(\C)$ is a one--dimensional Lipschitz soft wall potential as in Section~\ref{sec:softWall}. Note that $w(\bx + \lambda \ba_2) = w(\bx)$, so this wall is constant along the $\ba_2$ direction. In addition, since $\ba_2^\perp$ is normalized, this potential has the same Lipschitz constant as $w_1$. To introduce a dislocation parameter $t \in \R$, we set (see also Remark~\ref{rem:Lipschitz})
\begin{equation} \label{eq:wt_2d}
    w_t(\bx) := w(\bx - t \ba_1) = w_* \left( \ba_2^\perp \cdot \left( \bx - t \ba_1 \right) \right) = w_* \left( \ba_2^\perp \cdot \bx - t \frac{2 \pi}{\| \ba_1^* \|} \right).
\end{equation}
The corresponding operator $W(t)$ acting on $\ell^2(\Lat, \C^M)$ is the block--diagonal operator
\[
\forall \Psi \in \ell^2(\Lat, \C^M), \quad \forall \bR \in \Lat, \quad (W(t) \Psi)_{\bR} = w_t(\bR) \Psi_{\bR}.
\]
For a physical crystal subject to a scalar potential $v_*$ which is constant in the $\ba_2$-direction, one would consider a diagonal potential of the form
\begin{equation} \label{eq:wt_2d_bis}
    w_t(\bR) = {\rm diag} \left( v_* (\ba_2^\perp \cdot \left( \bR - t \ba_1 + \bx_1  \right) ), \cdots, v_* (\ba_2^\perp \cdot \left( \bR - t \ba_1 + \bx_M  \right) \right),
\end{equation}
where we recall that $\bx_1, \cdots, \bx_M$ are the positions of the atoms with respect to the center of the unit cell. Note that if $\bR = n \ba_1 + m \ba_2$, then $\ba_2^\perp \cdot \bR = n \ba_2^\perp \cdot \ba_1$.

\medskip

Finally, the soft wall model is the operator
\[
    \ssH^\sharp(t) = \ssH + W(t), \quad \text{with domain} \quad \cD(W) \subset \ell^2(\Z^2, \C^M).
\]
This operator is periodic in the $\ba_2$--direction. After the partial Fourier transform in this direction, we obtain that $\cF_2 \ssH^\sharp_t \cF_2 = \int_{\Gamma^*_2}^\oplus H_{\bk_2}^\sharp(t) \rd \bk_2$ where $H_{\bk_2}^\sharp(t)$ acts on $\ell^2(\Z, \C^M)$ with
\[
    H^\sharp_{\bk_2}(t) = H_{\bk_2} + W(t).
\]
We recover the soft wall described in Section~\ref{sec:softWall}: the normalization for $t$ in~\eqref{eq:wt_2d} has been chosen so that when $t$ moves from $0$ to $1$, the wall shifts by the vector $\ba_1$. We can now apply the previous results to this case.


\subsubsection{Edge spectrum in the two--dimensional setting}

For an energy $E \in \R \setminus \sigma_\bulk := \sigma(\ssH)$, we denote by $\cN(E)$ the number of (two--dimensional) Bloch bands of the two--dimensional bulk operator $\ssH$ below $E$. It is the only integer for which
\[
\forall \bk \in \R^2, \qquad \lambda_{\cN(E), \bk} < E < \lambda_{\cN(E) + 1, \bk},
\]
where we recall that $\lambda_{1, \bk} \le \cdots \le \lambda_{M, \bk}$ are the eigenvalues of the $M \times M$ matrix $\ssH(\bk)$ defined in~\eqref{eq:def:Ak_2d}. Similarly, for $\bk_2 \in \Gamma^*_2$, and for an energy $E \in \R \setminus \sigma(H_{\bk_2})$, we denote by $\cN_{\bk_2}(E)$ the number of (one--dimensional) Bloch bands of $H_{\bk_2}$ below $E$. It is the integer for which
\[
\forall \bk_1 \in \ba_1^* \R, \qquad \lambda_{\cN(E), \bk_1 + \bk_2} < E < \lambda_{\cN(E) + 1, \bk_1 + \bk_2}.
\]
If $E \in \R \setminus \sigma_\bulk$ is in the essential gap of the full two--dimensional operator $\ssH_\bk$, then, we have $\cN_{\bk_2}(E) = \cN(E)$ for all $\bk_2 \in \Gamma^*_2$. However, it may happen that an energy $E$ belongs to the essential gap of $H_{\bk_2}$ for {\em some} values of $\bk_2$, and is not in an essential gap of the full operator $\ssH$. In this case, $\cN_{\bk_2}(E)$ may be well-defined, but not $\cN(E)$. This happens for instance in the Wallace model for graphene (see Section~\ref{sec:exemple_wallace} below).

\medskip

We can now state our result in the two--dimensional case. It is a straightforward application of Proposition~\ref{prop:basics_intro} and Theorem~\ref{th:main_general_1d}.
\begin{theorem} \label{th:2d}
    Assume that $w_1 : \R \mapsto \cS_M$ is a Lipschitz soft wall. Then, with the previous notation,
    \begin{enumerate}
        \item  For all $\bk_2 \in \Gamma^*_2$, and all $t \in \R$, the operator $H^\sharp_{\bk_2}(t)$ is a well-defined self-adjoint operator whose domain is independent of $\bk_2$ and $t$ and given by
        \[
        \cD := \left\{ \Psi \in \ell^2(\Lat_1, \C^M), \quad W_{t = 0} \Psi \in \ell^2(\Lat_1, \C^M) \right\}.
        \] 
        \item The map $\Gamma^*_2 \times \R \ni (\bk_2, t)  \mapsto H^\sharp_{\bk_2}(t)$  is norm-resolvent continuous and translation equivariant in~$t$. 
        \item For all fixed $\bk_2 \in \Gamma^*_2$, the spectrum $t \mapsto \sigma(H^\sharp_{\bk_2}(t))$ is $1$-periodic in $t$. The essential spectrum $\sigma_{\rm ess}(H^\sharp_{\bk_2}(t))$ is independent of $t$, and equals $\sigma (H_{\bk_2})$ described in~\eqref{eq:spectrum_Hk2}.
        \item For all fixed $\bk_2 \in \Gamma^*_2$, and for all $E \in \R \setminus \sigma(H_{\bk_2})$, we have
        \[
        \Sf \left( H_{\bk_2}^\sharp(t), E, [0, 1] \right) = - \cN_{\bk_2}(E).
        \]
    \end{enumerate}
\end{theorem}

Concerning the spectrum at fixed $t = t_0$, we can apply Theorem~\ref{th:main_fix_t0}. We first make the observation that if $w(\cdot)$ is $\nu$--Lipschitz, then, for all $\bx \in \R^2$, the map $t \mapsto w_t(\bx)$ is $\nu \frac{| \Gamma |}{\| \ba_2 \|}$--Lipschitz. Indeed, we have, using the classical identity $\| \ba_1^* \| = 2 \pi \| \ba_2 \| | \Gamma |^{-1}$, that 
\[
   \left|  \widetilde{w}_{t}(\bx) - \widetilde{w}_{t'}(\bx)  \right| = 
   \left| w \left( \ba_2^\perp \cdot \bx - t \frac{2 \pi}{\| \ba_1^* \|} \right) - w \left( \ba_2^\perp \cdot \bx - t' \frac{2 \pi}{\| \ba_1^* \|} \right) \right| 
   \le \nu \frac{2 \pi}{\| \ba_1^* \|} | t - t' |  = \nu \frac{| \Gamma |}{\| \ba_2 \|} | t - t'|.
\]
In particular, the maps $t \mapsto W_t$ and $t \mapsto H_{\bk_2}^\sharp(t)$ are also $\nu \frac{| \Gamma |}{\| \ba_2 \|}$--Lipschitz. We obtain the following corrolary of  Theorem~\ref{th:main_fix_t0}.

\begin{theorem} \label{th:2d_fix_t0}
    Let $w : \R \mapsto \cS_N$ be a soft wall which is $\nu$--Lipschitz, for some $\nu > 0$. Then, for all $\bk_2 \in \Gamma^*_2$, all energy  $E$ belonging to an essential gap of width $g(E) \ge \nu | \Gamma | \| \ba_2 \|^{-1}$ of $H_{\bk_2}$, and all $t_0 \in \R$, the operator $ H^\sharp_{\bk_2} (t_0)$ has at least $\cN(E)$ eigenvalues in each interval of the form $(\lambda, \lambda+\nu | \Gamma | \| \ba_2 \|^{-1} ]$ in this gap.\\
    In particular, there are at least $\left\lfloor \dfrac{g(E) \| \ba_2 \|}{\nu | \Gamma | } \right\rfloor \cN_{\bk_2}(E)$ eigenvalues in this gap.
\end{theorem}

We skip its proof, as it is similar to the one of Theorem~\ref{th:main_fix_t0} (see also Remark~\ref{rem:Lipschitz}). This time, the density of eigenvalues is lower bounded by $\frac{\nu | \Gamma |}{\cN_{\bk_2}(E) \| \ba_2 \|}$.

\subsection{Walls in other directions}
\label{sec:general_angles}
We now give one possible way to extend our result to walls which are rotated by any commensurable angle. In the case of incommensurate angles, one can follow~\cite{Gon-21} to prove that {\em all essential bulk gaps are filled with edge spectrum}, for all values of $t$.

\medskip
We will show that a wall along any commensurable angle can be cast in the previous framework by increasing the size of the unit cell.
For two co-prime integers $n,m$ with $n \neq 0$ and $m \neq 0$, we define the new vectors
\begin{equation} \label{eq:def:tilde_a12}
    \widetilde \ba_1 := n \ba_1, \quad \text{and} \quad \widetilde \ba_2 := n \ba_1 + m \ba_2.
\end{equation}
In this section, the wall will be aligned with the vector $\widetilde{\ba}_2$, and move along the $\widetilde{\ba}_1$--direction (which is also the $\ba_1$--direction). Varying $n$ and $m$ allows to describe all commensurable directions for the wall, except for the $\ba_1$ and $\ba_2$--directions (which were studied in the previous section).

These two vectors generate a new lattice $\widetilde \Lat$, which is a sublattice of the original $\Lat$. We denote by $\widetilde{\Gamma}$ a corresponding unit cell (supercell). The number of atoms per unit cell is now
$$
\widetilde{M} = |n\, m| M .
$$
We consider the tight-binding Hamiltonian $\widetilde{\ssH}$ as the representation of the original $\ssH$ with $\widetilde \Gamma$ as unit cell. Note that the new kernel $\widetilde{\ssh}$ is now a function from $\widetilde{\Lat}$ to $\C^{\widetilde{M} \times \widetilde{M}}$.

\subsubsection{Bulk spectrum}

The spectrum of the bulk Hamiltonian is identical in both frameworks, but written in the Fourier decomposition for the $\widetilde \Lat$ lattice, it consists of $|n \, m|$-times more bands on a $|n\, m|$-times smaller Brillouin zone: a single band for the original Hamiltonian $\ssH$ corresponds to $|n\, m| $ bands of the \emph{supercell Hamiltonian} $ \widetilde \ssH$.
To make things concrete, we write $L= |n m|$ and label the $L$ points of $\Lat \cap \widetilde \Gamma$ as $\by_1, \cdots, \by_L$. We then define $\cS: \ell^2(\Lat, \C^M) \mapsto \ell^2(\widetilde{\Lat}, (\C^M)^L)$ by
\[
\bigl(\cS \Psi\bigr)(\bR) = \left(\Psi(\bR+ \by_1), \cdots, \Psi(\bR + \by_L) \right)^{\top}
\] 
With this transformation, 
$\widetilde \ssH := \cS \ssH \cS^*$ acts as
\[
\bigl( \widetilde \ssH \Phi \bigr)(\bR) = \sum_{\bR' \in \widetilde{\Lat}} \widetilde \ssh(\bR') \Phi_{\bR - \bR'} \quad
\text{ where } \quad \widetilde \ssh(\bR')_{ij} = \ssh(\bR' - \by_i + \by_j).
\]
\medskip
 Then, a computation shows that the dual lattice of $\widetilde{\Lat}$ is $\widetilde{\Lat}^* := \widetilde{\ba}_1^* \Z \oplus \widetilde{\ba}_2^*\Z$, with
\begin{equation} \label{eq:tildea*_with_a*}
    \widetilde \ba_1^* = \frac{\ba^*_1}{n} - \frac{\ba^*_2}{m} \qquad \widetilde \ba_2^*= \frac{\ba^*_2}{m}.
\end{equation}
This time, it is the dual lattice $\Lat^*$ which is a sublattice of $\widetilde{\Lat}^*$. The corresponding Brillouin zone $\widetilde{\Gamma^*} := \R^2 / \widetilde{\Lat}^*$ has an area $| \widetilde{\Gamma^*}  | = \frac{1}{| n \, m |} | \Gamma^* |$. The operator $\widetilde{\ssH}$ is again a convolution, and we denote by $\widetilde{\ssH}({\bk})$ the corresponding Fourier fibers of the operators. The map $\bk \mapsto \widetilde{\ssH}(\bk)$ is smooth and $\widetilde{\Lat}^*$--periodic. We describe this relation more precisely in this short Lemma.

\begin{lemma} \label{lem:Bloch_tilde}
      We have
      \[
      \sigma \left( \widetilde{\ssH}(\bk) \right) = \bigcup_{\widetilde{\bK} \in \widetilde{\Lat}^* \cap \Gamma^*} \sigma \left( \ssH(\bk + \widetilde{\bK}) \right)
      = \bigcup_{\widetilde{\bK} \in \widetilde{\Lat}^*} \sigma \left( \ssH(\bk + \widetilde{\bK}) \right).
      \]
\end{lemma}
\begin{proof}
    This result is classical but we provide a simple proof.
   Recall the definition
    \[
\widetilde \ssH_\bk = \sum_{\bR \in \widetilde \Lat} \widetilde\ssh(\bR) \re^{-\ri \bR \cdot \bk} \qquad \in \cS_{L M}(\C).    
    \]
We fix $\bk \in \R^2$ (or in the Brillouin zone $\R^2 /\widetilde \Lat^*$) and compute its spectrum. To this end, for some $u \in \C^M$ and $\bK \in \widetilde \Lat^*$, we define
\[
u_{\bK} := \left(\re^{-\ri \by_1 \cdot (\bk+\bK) } u, \cdots, \re^{-\ri \by_L \cdot (\bk+\bK) } u \right)^{\top},
\]   
and compute, taking advantage of the periodicity of $\widetilde \ssH_\bk$,
\begin{align*}
\left(\widetilde \ssH_\bk u_{\bK} \right)_i 
= \left(\widetilde \ssH_{\bk + \bK} u_{\bK} \right)_i 
&= \sum_{\bR \in \widetilde \Lat}\re^{-\ri \bR \cdot (\bk+ \bK)}  \left( \widetilde\ssh(\bR) u_{\bK} \right)_i \\
&=\sum_{\bR \in \widetilde \Lat}\re^{-\ri \bR \cdot (\bk+ \bK)} \sum_{j=1}^L \ssh(\bR -\by_i + \by_j)\re^{-\ri \by_j \cdot (\bk+\bK) } u .
\end{align*}
Relabeling $\bR':= \bR -\by_i + \by_j$, we note that this variable runs over all sites of the original lattice $\Lat$, so we obtain finally
\begin{align*}
\left(\widetilde \ssH_\bk u_{\bK} \right)_i &=\re^{-\ri \by_i \cdot (\bk+ \bK)}\sum_{\bR' \in  \Lat}\re^{-\ri \bR'  \cdot (\bk+ \bK)}  \ssh(\bR') u = \re^{-\ri \by_i \cdot (\bk+ \bK)} \ssH_{\bk+\bK} u.
\end{align*}   
If $u$ is an eigenvector of $\ssH_{\bk+\bK}$, then $u_{\bK}$ is an eigenvector of $\widetilde \ssH_\bk$ with the same eigenvalue, for all $\bK \in \widetilde{\Lat}^*$. In addition, the eigenvectors $u_\bK$ are linearly independent for $\bK \in \widetilde \Lat^* / \Lat^*$, so we have obtained the $L M$ eigenvalues of the supercell Hamiltonian $\widetilde \ssH(\bk)$, and the result follows.
\end{proof}

When we perform a partial Fourier transform of $\widetilde{\ssH}$ in the $\widetilde{\ba}_2$--direction, and obtain the new family $\widetilde{H}_{\widetilde{\bk}_2}$ with $\widetilde{\bk}_2 \in \widetilde{\Gamma}^*_2$. 

\begin{lemma}
  We have
  \[
   \sigma \left( \widetilde{H}_{\widetilde{\bk}_2} \right) = \bigcup_{\widetilde{\bk}_1 \in \widetilde{\ba}_1^* \R} \sigma \left( \ssH(\widetilde{\bk}_1 + \widetilde{\bk}_2) \right).
  \]
\end{lemma}
In other words, the spectrum of $\widetilde{H}_{\widetilde{\bk}_2}$ is again the projection of the spectra of initial operator $\ssH_\bk$ along the $\ba_1^*$--direction (which is the direction orthogonal to the wall). It is therefore quite easy to compute the (essential) spectrum of the one--dimensional bulk operator $\widetilde{H}_{\widetilde{\bk}_2}$ for any angle: it is a projection of the two--dimensional Bloch bands of the bulk operator $\ssH_\bk$ (which is independent of the angle) along a particular direction. Figure~\ref{fig:graphene halfspace} illustrates this phenomenon for the Wallace model.

\begin{proof}
    As in~\eqref{eq:spectrum_Hk2}, we find that
\[
\sigma \left( \widetilde{H}_{\widetilde{\bk}_2} \right) = \bigcup_{\widetilde{\bk}_1 \in \widetilde{\Gamma}_1^*} \sigma \left( \widetilde{\ssH}(\widetilde{\bk}_1 + \widetilde{\bk}_2) \right).
\]
Together with Lemma~\ref{lem:Bloch_tilde}, we deduce that
\[
  \sigma \left( \widetilde{H}_{\widetilde{\bk}_2} \right) = \bigcup_{\widetilde{\bK} \in \widetilde{\Lat}^*} \bigcup_{\widetilde{\bk}_1 \in \widetilde{\Gamma}_1^*} \sigma \left( \ssH (\widetilde{\bk}_1 + \widetilde{\bk}_2 + \widetilde{\bK}) \right)
    = \bigcup_{\widetilde{\bk}_1 \in \widetilde{\ba}_1^* \R} \bigcup_{\widetilde{\bK}_2 \in \widetilde{\ba}_2^* \Z} \sigma \left( \ssH (\widetilde{\bk}_1 + \widetilde{\bk}_2 + \widetilde{\bK}_2^*) \right).
\]
It remains to prove that the union in $\widetilde{\bK}_2^*$ can be dropped. Recall that we assume that $n$ and $m$ are coprime. Let $(p,q) \in \Z^2$ so that $pn + qm = 1$. We have, using~\eqref{eq:tildea*_with_a*},
\[
    \ssH \left( \bk - p\, n \widetilde{\ba}_1^* \right) = \ssH \left( \bk - p \ba_1^* + \frac{pn}{m} \ba_2^* \right) = \ssH \left( \bk - p \ba_1^* - q \ba_2^* + \frac{1}{m} \ba_2^* \right) = \ssH \left( \bk + \widetilde{\ba}_2^* \right),
\]
where we used the $\Lat$--periodicity of $\bk \mapsto \ssH(\bk)$ in the last equality. This proves that any integer shift in the $\widetilde{\ba}_2^*$--direction can be implemented as another shift in the $\widetilde{\ba}_1^*$--direction. The result follows.
\end{proof}

\subsubsection{Spectral flow for the edge Hamiltonian with general angles}

We denote by $\widetilde \ssH^\sharp(t)$ the tight-binding model with a soft wall defined as in the previous section, but now with $\widetilde \ba_2$. Applying Theorem~\ref{th:2d_fix_t0} to this case results in  
\begin{theorem} \label{th:lot_of_edge_states_if_incommensurate}
Let $w : \R \mapsto \cS_N$ be a soft wall which is $\nu$--Lipschitz, for some $\nu > 0$. If $E$ belongs to an essential gap of width $g(E) \ge  \nu$, then for all $t_0 \in \R$ the operator $ \widetilde{H}^\sharp(t_0)_{\widetilde{\bk}_2}$ has at least $\cN(E)$ eigenvalues in each interval of the form $(\lambda, \lambda+\nu |\Gamma|/\|\widetilde \ba_2\|]$ in this gap.
\end{theorem}  

Note that the density of edge states is lower bounded by $\cN(E ) \| \widetilde \ba_2 \|/ ( \nu | \Gamma |)$. When the cut becomes close to incommensurate, in the sense that $\| \widetilde \ba_2 \|$ becomes large, then there are numerous edge states for each value of $\widetilde{\bk}_2$. 

\begin{proof}
According to Theorem~\ref{th:2d}, 
\[
\Sf \left({\widetilde H^\sharp_{\widetilde \bk_2}(t), E, [0,1]}\right) = -\cN_{\widetilde\ssH} (E).
\]
Each Bloch band in the supercell framework corresponds to $L$ Bloch bands of the original Hamiltonian, so we have $\cN_{\widetilde\ssH} (E) = L \cN_{\ssH} (E)$. On the other hand, we claim that the spectrum of $\widetilde H^\sharp_{\widetilde \bk_2}(t)$ is periodic in $t$ with period $1/L$. This can be seen as follows. 

Recall that $n,m$ are coprimes, and let $p,q \in \Z$ be such that $pn + qm = 1$ as before. Then, we have $\frac{1}{ nm } = \frac{p}{n} + \frac{q}{m}$. This gives, 
\[
    \widetilde w_{\frac{1}{nm}}(\bx) =  w_* \left(  \widetilde{\ba}_2^\perp  \cdot \left( \bx - \frac{1}{nm} \widetilde{\ba}_1 \right) \right) = \widetilde w\left(  \widetilde{\ba}_2^\perp  \cdot \left( \bx - \frac{1}{nm} \widetilde{\ba}_1 + \frac{1}{m} \widetilde{\ba}_2 \right)  \right)
    = \widetilde w\left(  \widetilde{\ba}_2^\perp  \cdot \left( \bx + \ba_2 \right)  \right),
\]
so $\widetilde w_{\frac{1}{nm}}(\bx) = \widetilde w_{0}(\bx + \ba_2)$. After a shift of the wall by $t = \frac{1}{nm}$, we recover the initial model, translated by $-\ba_2$. This proves our claim.
We conclude that
\[
\Sf \left({\widetilde H^\sharp_{\widetilde \bk_2}(t), E, \left[ 0,\frac{1}{L} \right]} \right)=\frac{1}{L}\Sf \left( \widetilde H^\sharp_{\widetilde \bk_2}(t), E, [0,1] \right) = - \cN_{\ssH} (E).
\]
Recalling that the eigenvalue branches for $t \mapsto \widetilde \ssH^\sharp_{\widetilde \bk_2}(t)$ are Lipschitz with constant $\nu |\widetilde \Gamma| / \|\widetilde \ba_2\| = \nu L \Gamma/ \|\widetilde \ba_2\| $ and mimicking the proof of Theorem~\ref{th:main_fix_t0} gives the result.   
\end{proof}


\subsection{Example : the Wallace model for graphene}
\label{sec:exemple_wallace}

We now apply the previous results in the case of graphene, within the Wallace approximation. We consider three different angles, and compute the edge spectral flow in these three cases.

\subsubsection{The Wallace model}
\begin{figure}[]
    \centering
      \hspace{\fill}
    \includegraphics[height=0.3 \textwidth]{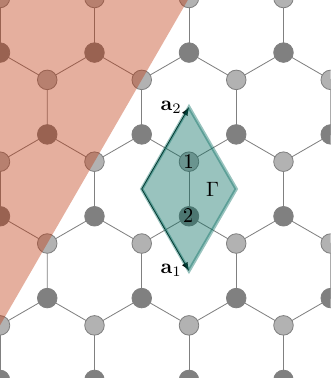}
    \hspace{\fill}
    \includegraphics[height=0.3 \textwidth]{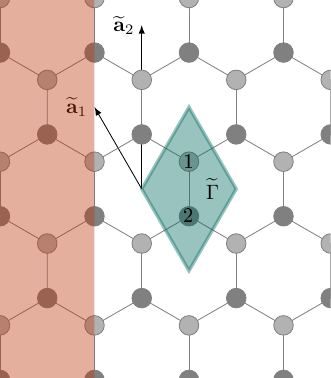}
    \hspace{\fill}
    \includegraphics[height=0.3 \textwidth]{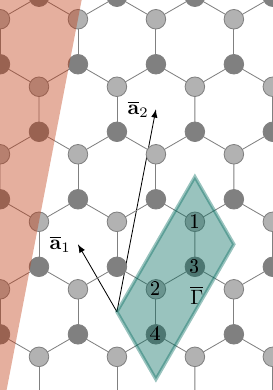}
      \hspace{\fill}
    \caption{Our convention for the zigzag direction (left), armchair direction with $(n,m) = (-1, 1)$ (middle), and another angle with $(n,m) = (2, -1)$ (right). The shaded cell represents $\Gamma$ (resp. $\widetilde{\Gamma}$, $\overline{\Gamma}$).}
    \label{fig:convention_angle}
\end{figure}
Let us start by describing the bulk spectrum of graphene. Graphene has a honeycomb lattice structure, which is the triangular Bravais lattice, with two atoms per unit cell. As basis vectors, we may choose
\[
    \Lat = \ba_1 \Z \oplus \ba_2 \Z
    \quad \text{with} \quad 
    \ba_1 = \frac{a_0}2 \begin{pmatrix}    1 \\ -\sqrt{3} \end{pmatrix}
    \quad \text{and} \quad
    \ba_2 = \frac{a_0}2 \begin{pmatrix}  1 \\ \sqrt{3} \end{pmatrix},
    \qquad \text{(zig--zag convention)},
\] 
where $a_0 = \sqrt{3}d_0$ is the graphene lattice constant and $d_0$ is the distance between neighboring carbon atoms. For unit cell we take $\Gamma = [0,1) \ba_1 + [0, 1) \ba_2$ as shown in Figure~\ref{fig:convention_angle} and there are $M = 2$ carbon atoms, located at 
$\bx_1 = \frac13( \ba_1 + 2\ba_2)$ and $\bx_2 = \frac13 (2 \ba_1 + \ba_2)$. The Wallace Hamiltonian only contains hopping between nearest neighbours, with
\[
\begin{cases}    
(\ssH\Psi)_{\bR}^{(1)} & = \sst_0 \left( \Psi_{\bR}^{(2)} + \Psi_{\bR - \ba_1}^{(2)} + \Psi_{\bR + \ba_2}^{(2)} \right), \\
    (\ssH \Psi)_{\bR}^{(2)} & = \sst_0 \left( \Psi_{\bR}^{(1)} + \Psi_{\bR + \ba_1}^{(1)} + \Psi_{\bR - \ba_2}^{(1)} \right).
 \end{cases}
\]
Here, $\sst_0$ is the carbon-carbon hopping amplitude and since it just sets the energy scale, we will take $\sst_0= 1$ throughout this Section. There are only five non-zero matrices $\ssh(\bR)$, given by
\[
\ssh(\bnull) = \begin{pmatrix}
    0 & 1 \\ 1 & 0
\end{pmatrix}, \qquad
\ssh(\ba_1) =  \ssh(-\ba_2) = \begin{pmatrix}
    0 & 1 \\ 0 & 0
\end{pmatrix},
\qquad
\ssh(-\ba_1) = \ssh(\ba_2) =  \begin{pmatrix}
    0 & 0 \\ 1 & 0
\end{pmatrix},
\]
We find that the Bloch fibers are given by 
\[
    \ssH_\bk = \begin{pmatrix}
    0 & 1 + \re^{ -\ri \bk\cdot \ba_1} + \re^{  \ri \bk \cdot \ba_2} \\ c.c. & 0
\end{pmatrix}.
\]
The map $\bk \mapsto \ssH_\bk$ is $\Lat^*$--periodic, with $\Lat^* = \ba_1^* \Z \oplus \ba_1^* \Z$, with
\[
    \ba_1^* = \frac{2 \pi}{\sqrt{3} a_0} \begin{pmatrix} \sqrt{3} \\ -1   \end{pmatrix}, \qquad
    \ba_2^* = \frac{2 \pi}{\sqrt{3} a_0} \begin{pmatrix} \sqrt{3} \\ 1   \end{pmatrix}.
\]
The eigenvalues of $\ssH_\bk$ are $\pm | 1 + \re^{ -\ri \bk\cdot \ba_1} + \re^{  \ri \bk \cdot \ba_2}|$. Thus, the two bands are mirror images that meet at the so called Dirac points $\bk \in \bK + \Lat^*$ or $\bk \in \bK' + \Lat^*$,  $\bK := \frac13(\ba_1^* + \ba_2^*)$ and $\bK' := \frac23(\ba_1^* + \ba_2^*)$. These are the only points where the gap closes.

\subsubsection{Zigzag orientation}
With this choice for the lattice vector $\ba_2$, and assuming as before that the wall is constant along the $\ba_2$--direction, we can see that the associated halfspace model would corresponds to the zig-zag or bearded cut (see the left panel of Fig.~\ref{fig:convention_angle}). 

\medskip

After a partial Fourier transform in the $\bk_2$ direction, we find the Jacobi operator with blocks
\[
b= h_{\bk_2} (\bnull) = \begin{pmatrix}
0 & 1+ \re^{-\ri \bk_2 \cdot \ba_2 } \\ c.c. & 0
\end{pmatrix}, \quad 
a^* = h_{\bk_2} (\ba_1) = \begin{pmatrix}
0 & 0\\ 1 & 0
\end{pmatrix}.
\]
We recover the SSH matrices studied in Section~\ref{sec:SSH}, with $J_1 := 1+ \re^{-\ri \bk_2 \cdot \ba_2 }$ and $J_2 = 1$. Note that $J_1$ is now complex valued, but its phase can be removed by a unitary transformation that does not affect the wall potential.

For shortness, define $r(\bk_2) = | 1 + \re^{- \ri \bk_2 \cdot \ba_2} | = | 2 \sin \left( \frac{\bk_2 \cdot \ba_2}{2} \right) |$.
The essential spectrum is just
\[
    \sigma(H_{\bk_2}) = \sigma_\ess(H_{\bk_2}) = [- |1 + r(\bk_2)|, - | 1 - r(\bk_2) |] \cup [ |1 - r(\bk_2), | 1 + r(\bk_2) |],
\]
The gap of the one-dimensional operator $H_{\bk_2}$ closes only for $\bk_2 = \frac13 \ba_2^* + \Lat_2^*$ and $\bk_2 = \frac23 \ba_2^* + \Lat_2^*$, corresponding  to the projections of the $\bK$ and $\bK'$ points on the line $\ba_2^* \R$, along the $\ba_1^*$--direction. This can be seen in the first panel of Figure~\ref{fig:graphene halfspace}. Apart from these two points, we have $\cN_{\bk_2}(E = 0) = 1$. Numerical simulations of the edge modes are presented in Figures \ref{fig:graphene_edge_spectrum_different_k2} and \ref{fig:graphene_edge_spectrum_different_t}.

\medskip

Note that $r(\bk_2) < 1$ if and only if $\bk_2 \in (-1/3, 1/3) \ba_2^*$ modulo $\Lat_2^*$. In some sense, the Dirac cones $\pm \frac{1}{3}\ba_2^*$ separate the two phases $| J_1 | < | J_2 |$ and $| J_2| < | J_1 |$ of the SSH model. According to Remark~\ref{rem:hardcut_SSH}, this explains in particular why a hard cut would create a line of edge modes in the gap $(-1/3, 1/3) \ba_2^*$ or in the gap $(1/3, 2/3) \ba_2^*$ modulo $\Lat_2^*$, depending on the location of a hard wall with respect to the lattice, which corresponds either to a zigzag edge or to a bearded zigzag edge.
Also note that the spectrum is symmetric around zero, and that for $\bk_2 = \ba_2/2$ (at the edge of the first Brillouin zone) the chain decouples and the essential spectrum is just $\{-1,1\}$. 

\medskip

\subsubsection{The armchair orientation}
The armchair direction can be obtained by taking $n = -1$ and $m = 1$, that is (we use the tilde notation for the armchair convention)
$\widetilde{\ba}_1 = -\ba_1$, $\widetilde{\ba}_2 = -\ba_1 + \ba_2$, $\widetilde{\ba}_1^* = -\ba_1^* - \ba_2^*$ and $\widetilde{\ba}_2^* = \ba_2^*$, namely
\[
    \widetilde{\ba}_1 = \frac{a_0}2 \begin{pmatrix}    -1 \\ \sqrt{3} \end{pmatrix}, \quad 
    \widetilde{\ba}_2= \frac{a_0}2 \begin{pmatrix}    0 \\ 2 \sqrt{3} \end{pmatrix}, \qquad
    \widetilde{\ba}_1^*  = \frac{2 \pi}{\sqrt{3} a_0} \begin{pmatrix} -2 \sqrt{3} \\ 0   \end{pmatrix}, \quad
    \widetilde{\ba}_2^* = \frac{2 \pi}{\sqrt{3} a_0} \begin{pmatrix} \sqrt{3} \\ 1   \end{pmatrix}.
\]
See Fig.~\ref{fig:convention_angle} for the lattice configuration (middle) and Fig.~\ref{fig:graphene halfspace} for the reciprocal space. The matrix elements are given by
\[
    \widetilde{\ssh}(\bnull) = \begin{pmatrix}
    0 & 1 \\ 1 & 0
    \end{pmatrix}, \quad
    \widetilde{\ssh}(\widetilde{\ba}_1) = \widetilde{\ssh}(\widetilde{\ba}_2 - \widetilde{\ba}_1) =  \begin{pmatrix}
    0 & 0\\ 1 & 0
    \end{pmatrix}, \quad
    \widetilde{\ssh}(-\widetilde{\ba}_1) = \widetilde{\ssh}(\widetilde{\ba}_1 - \widetilde{\ba}_2) = \begin{pmatrix}
    0 & 1 \\ 0 & 0
    \end{pmatrix}.
\]

This time, the spectrum of $\widetilde{H}_{\widetilde{\bk}_2}$ is the projection of the Bloch bands of the Wallace model along the horizontal direction, which sends the two inequivalent Dirac points $\bK$ and $\bK'$ into the single point $\bk_2 = 0$. For all other values of $\widetilde{\bk}_2 \in \R \widetilde{\ba}_2^*$, there is a band gap, and $\cN_{H_{\widetilde{\bk}_2}}(0) = 1$. So a spectral flow of $-1$ will appear when a wall is added and moved along $\widetilde{\ba}_1$.

\begin{figure}[ht]
    \centering
    \includegraphics[page=1, width= 0.3 \textwidth]{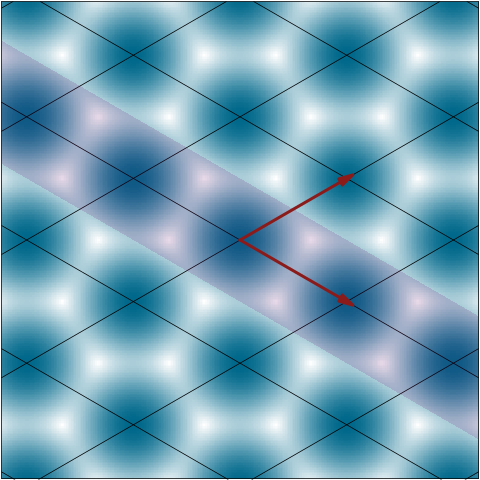}
    \includegraphics[page=2, width= 0.3 \textwidth]{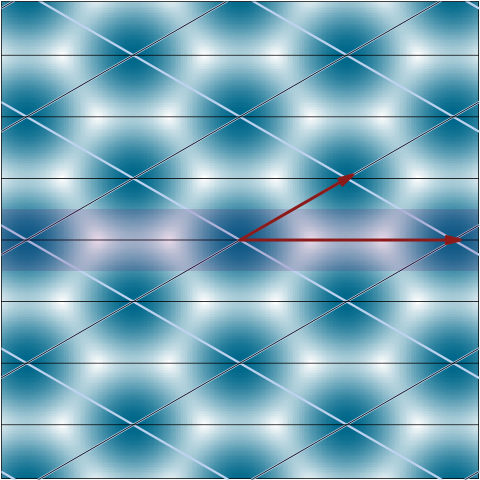}
    \includegraphics[page=3, width= 0.3 \textwidth]{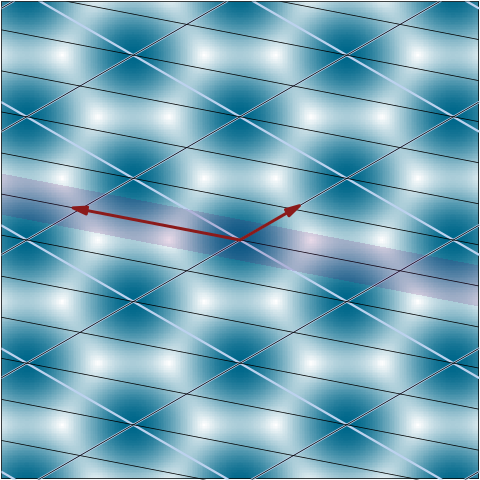}
    \caption{Brillouin zones and reciprocal vectors for the Wallace model corresponding to a zigzag cut, armchair cut ($n=1, m=-1$) and a cut with $n=2$, $m=-1$.  In the second two pictures, the gray lines indicate the original Brillouin zone, which is covered by $|n m|$ copies of the new Brillouin zone. The blue-white gradient shows the upper energy band for the Wallace model, the white dots are the Dirac points. The band structure of $H_{\bk_2}$ (resp. $\widetilde{H}_{\widetilde{\bk}_2}$, $\overline{H}_{\overline{\bk}_2}$) is obtained by projecting parallel to the shaded band onto the $ \ba_2^*$-vector (resp. $\widetilde{\ba}_2$, $\overline{\ba}_2$) that points to the upper right corner.}
    \label{fig:graphene halfspace}
\end{figure}

\subsubsection{Another rational cut}
To illustrate the discussion in Section~\ref{sec:general_angles}, we consider a different rational direction with $n = -1$ and $m = 2$, that is (we use the bar notation for this angle) $\overline{\ba}_1 = -\ba_1$, $\overline{\ba}_2 = -\ba_1 + 2\ba_2$, $\overline{\ba}_1^* = -\ba_1^* - \frac12 \ba_2^*$ and $\overline{\ba}_2^* = \frac12 \ba_2^*$, namely
\[
    \overline{\ba}_1 = \frac{a_0}2 \begin{pmatrix}    -1 \\ \sqrt{3} \end{pmatrix}, \quad 
    \overline{\ba}_2= \frac{a_0}2 \begin{pmatrix}    1 \\ 3 \sqrt{3} \end{pmatrix}, \qquad
    \overline{\ba}_1^*  = \frac{\pi}{\sqrt{3} a_0} \begin{pmatrix} -3\sqrt{3} \\ 1   \end{pmatrix}, \quad
    \overline{\ba}_2^* = \frac{\pi}{\sqrt{3} a_0} \begin{pmatrix} \sqrt{3} \\ 1   \end{pmatrix}.
\]
The lattice structure is shown in the rightmost panel of Figure~\ref{fig:convention_angle} and the reciprocal lattice in Figure~\ref{fig:graphene halfspace}.
This time, the lattice $\overline{\Lat} := \overline{\ba}_1 \Z \oplus \overline{\ba}_2 \Z$ strictly contains the initial graphene lattice $\Lat$. The corresponding new unit cell $\widetilde{\Gamma}$ consists of two copies of the original unit cell and contains $4$ atoms. 
In order to construct the $4\times 4$ matrices that will give $\ssh$, it is convenient to label first the sites on one sublattice and then those on the other sublattice, so that $\ssh$ is block-off-diagonal, see again Fig.~\ref{fig:convention_angle}. 

This gives $\overline{\ssh}({\bR}) = \begin{pmatrix}
0 & f({\bR}) \\ g({\bR}) & 0
\end{pmatrix}$ with $g({\bR}) = f(-{\bR})^*$. The nonzero entries are given by
\[
f (\bnull) = \begin{pmatrix}
1 & 0 \\ 1 & 1
\end{pmatrix}, 
\qquad
f (-\overline{\ba}_1)= \begin{pmatrix}
1 & 0 \\ 0 & 1
\end{pmatrix}, 
\qquad
f(\overline{\ba}_1 - \overline{\ba}_2) = \begin{pmatrix}
0 & 1 \\ 0 & 0
\end{pmatrix},
\]
and $f(\overline{\bR}) = 0$ otherwise. The operator $\overline{H}_{{\bk}_2}$ is defined for ${\bk}_2 \in \R \overline{\ba}_2^*$, and is $\overline{\ba}_2^*$--periodic. Note that $| \overline{\ba}_2^* | = \frac12 | \ba_2^*|$: the Brillouin zone is twice smaller than in the previous cases. Simple geometric considerations shows that $\bK$ and $\bK'$ Dirac point are projected on $\pm \frac13 \overline{\ba}_2^*$ modulo $\overline{\ba}_2^* \Z$. Note that the two projections differ. In the terminology of~\cite{AkhBee-08, FefFliWei-22}, this angle is of {\em zig--zag} type. It is of {\em armchair} type if the projections of $\bK$ and $\bK'$ coincide. In our work however, there are no fundamental differences between these two cases. Except from these special values, there is a gap around zero, but now $\cN_{H_{\widetilde{\bk}_2}}(0) = 2$, so a spectral flow of $-2$ appears in the gap.

\subsubsection{Numerical illustration of the edge spectrum for different cuts}

For the three cuts described before, we numerically compute the corresponding edge (and bulk) spectra\footnote{\label{footnote:data}The code can be found here: \url{https://gitlab.com/davidgontier/softwall_jacobimatrix}}. We use the same strategy as in Section~\ref{sec:SSH} to detect edge modes.

\medskip

For the soft wall, we took $w_t(\bx)$ of the form~\eqref{eq:wt_2d}-\eqref{eq:wt_2d_bis}, with the scalar potential
\[
    v_1(x) := \begin{cases}
        0 & \quad \text{for} \quad x \ge 0 \\
        - \nu x  & \quad \text{for} \quad x \le 0.
    \end{cases}.
\]
The corresponding edge models $\ssH^\sharp(t)$ (zigzag orientation), $\widetilde{\ssH}^\sharp(t)$  (armchair orientation) and $\overline{\ssH}^\sharp(t)$ (our commensurate angle with $n = 2$ and $m = -1$, called {\em angle} in what follows) are translation equivariant. Recall that our results shows that, outside the projection of the Dirac cones, we must have
\begin{equation} \label{eq:expected_Sf_numerics}
\Sf \left( H^\sharp_{\bk_2}(t), [0, 1], E \right) = \Sf \left( \widetilde{H}^\sharp_{\overline{\bk}_2}(t), [0, 1], E \right) = 1 \quad \text{while} \quad
\Sf \left( \overline{H}_{\overline{\bk}_2}^\sharp(t), [0, 1], E \right) = 2.
\end{equation}

\underline{The edge modes as functions of $k_2$.}
In Figure~\ref{fig:graphene_edge_spectrum_different_t}, we plot the edge spectra of $\bk_2 \mapsto H_{\bk_2}(t)$, $\widetilde{\bk}_2 \mapsto \widetilde{H}_{\widetilde{\bk}_2}(t)$ and $\overline{\bk}_2 \mapsto \overline{H}_{\overline{\bk}_2}(t)$ for different values of $t \in [0, 1]$, and with the Lipschitz constant $\nu = 1$. The grey parts represent the bulk essential spectrum, and the red curves are the edge eigenvalues. 

\medskip

In these plots, the $x$--axis corresponds to $| \bk_2 |$, respectively $| \widetilde{\bk}_2 |$, $\overline{\bk}_2$. By periodicity, we represent the plots on a portion $[-\frac12 | \ba_2^* |, \frac12 | \ba_2^* |]$, etc. Note that $| \overline{\ba}_2^* | = \frac12 | \widetilde{\ba}_2^* | = \frac12 | \ba_2^* |$, so the $x$--axis for the last plot is twice smaller. 

\medskip

According to these simulations, it seems that the edge mode curves $\bk_2 \mapsto \lambda(\bk_2)$ do not necessarily starts and ends at the Dirac cones, as is probably the case for hard cut truncations~\cite{FefFliWei-22}. One reason is that the addition of a soft wall breaks the chiral symmetry: the spectrum of $H^\sharp(t)$ is no longer symmetric with respect to the origin, so the eigenvalue $E = 0$ no longer plays a special role for edge modes.

\medskip

Another remark is that there is no qualitative difference between the different cuts: we see edge curves in each gap.

\begin{figure}
    \centering
    \includegraphics[width=0.9\textwidth]{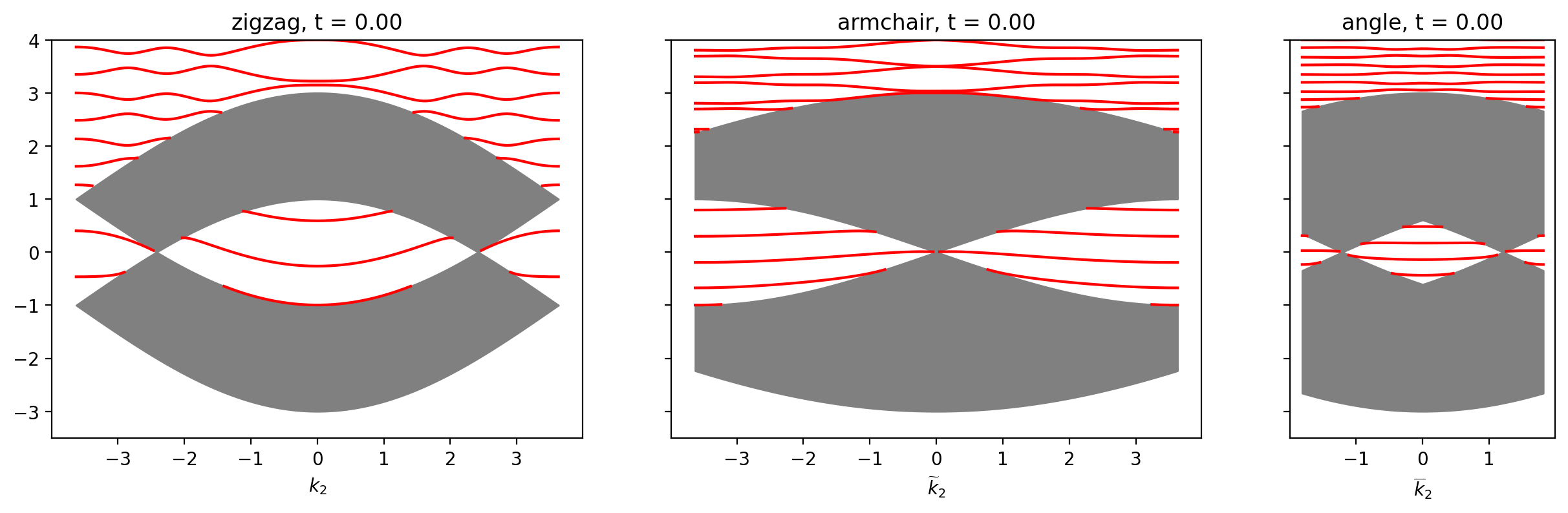}
    \includegraphics[width=0.9\textwidth]{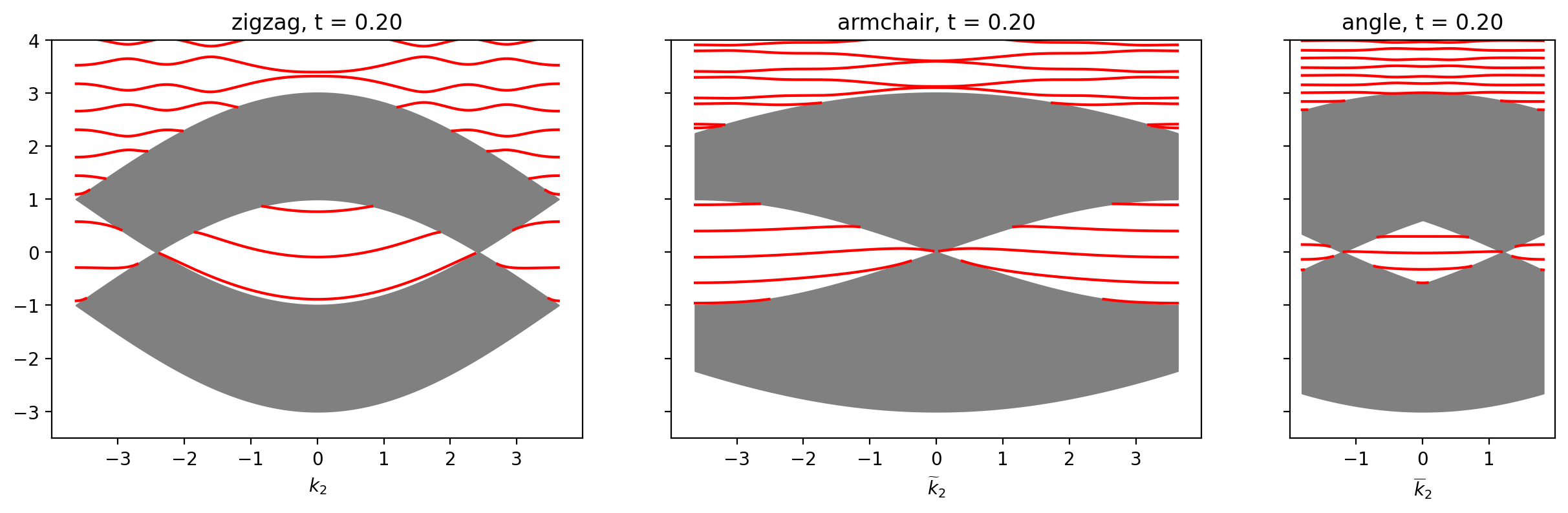}
    \includegraphics[width=0.9\textwidth]{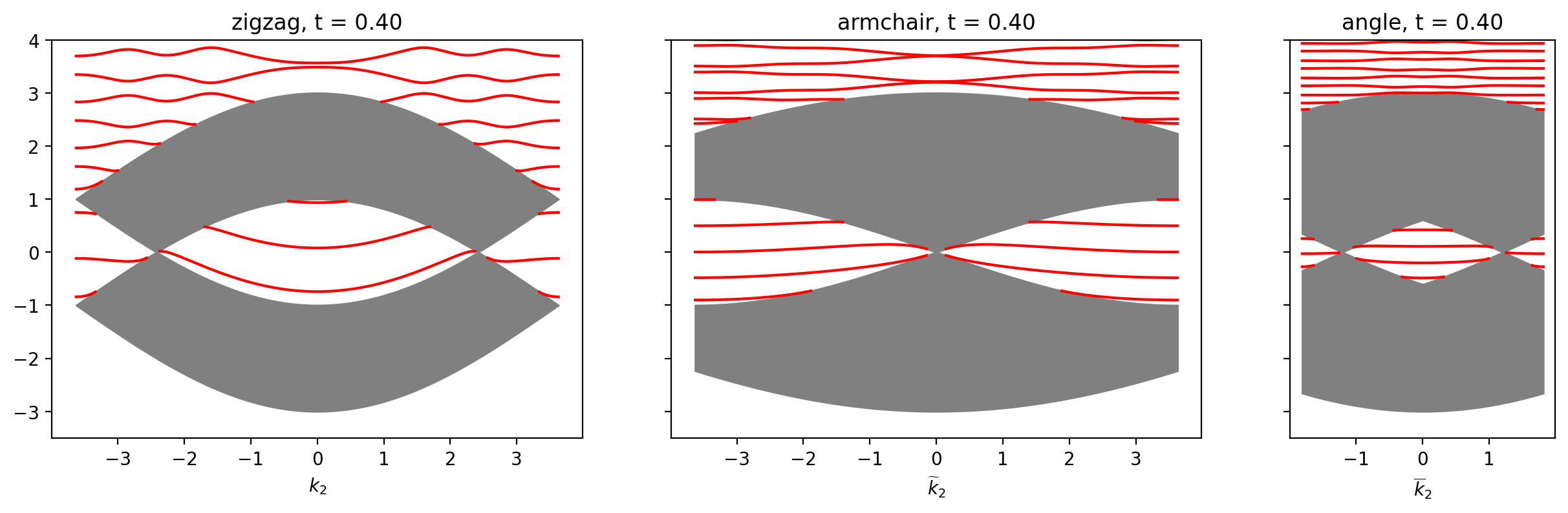}
    \includegraphics[width=0.9\textwidth]{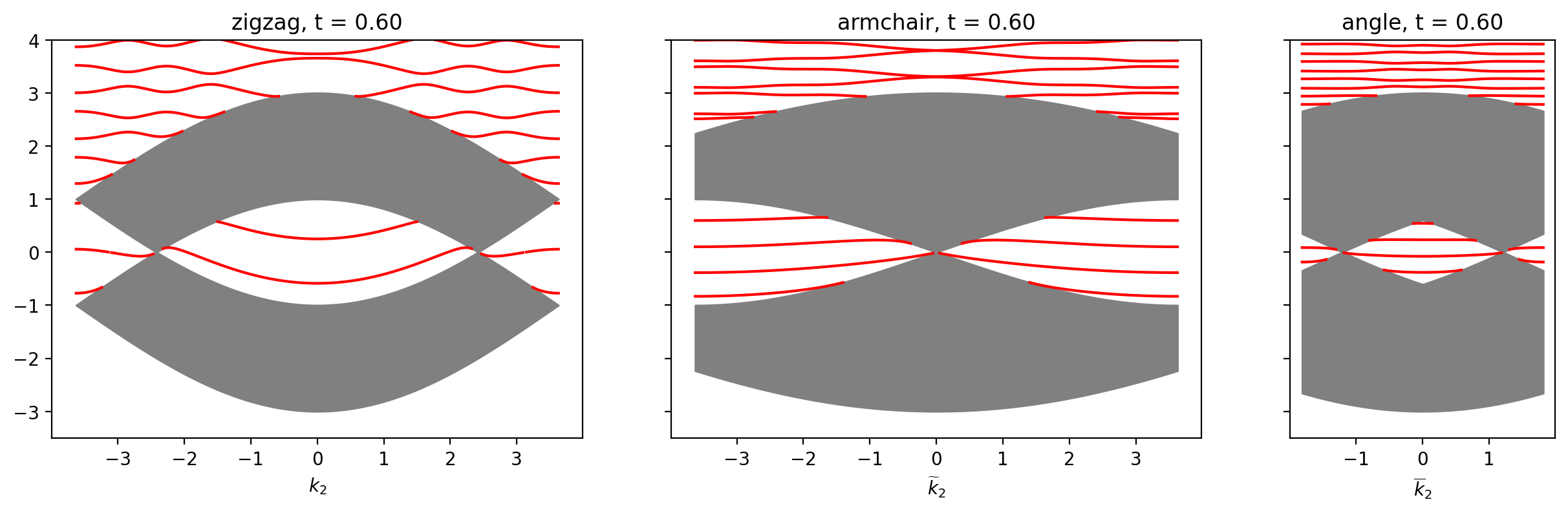}
    \includegraphics[width=0.9\textwidth]{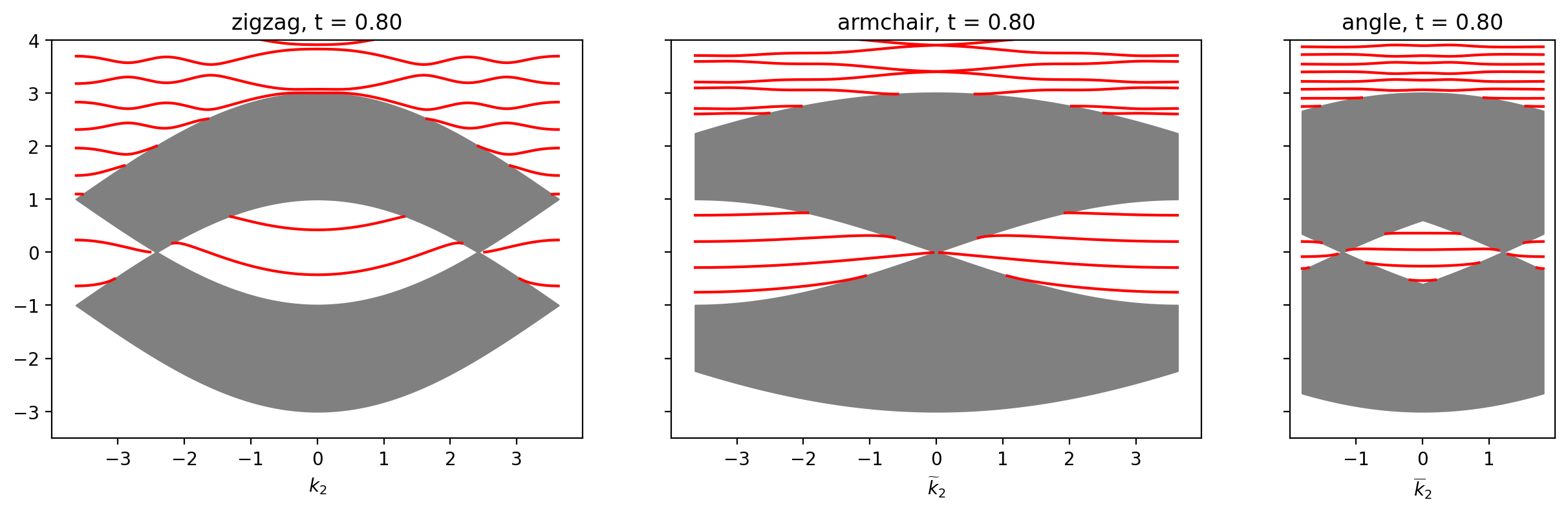}
    
    \caption{Spectrum of the edge Hamiltonians $\bk_2 \mapsto H_{\bk_2}^\sharp(t)$, $\widetilde{\bk}_2 \mapsto \widetilde{H}_{\widetilde{\bk}_2}^\sharp(t)$ and $\overline{\bk}_2 \mapsto \overline{H}^\sharp_{\overline{\bk}_2}(t)$ for different values of $t \in \{ 0, 0.2, 0.4, 0.6, 0.8 \}$, and with $\nu = 1$.}
    \label{fig:graphene_edge_spectrum_different_k2}
\end{figure}


\underline{The edge modes as functions of $t$.}
In Figure~\ref{fig:graphene_edge_spectrum_different_t}, we fix $\bk_2 = \frac16 \ba_2^*$, $\widetilde{\bk}_2 = \frac16 \widetilde{\ba}_2^*$ and $\overline{\bk}_2 = \frac16 \overline{\ba}_2^*$, and plot the edge spectra of $t \mapsto H_{\bk_2}(t)$, $t \mapsto \widetilde{H}_{\widetilde{\bk}_2}(t)$ and $t \mapsto \overline{H}_{\overline{\bk}_2}(t)$. We took different values of the Lipschitz constants. In these figures, we clearly observe the expected spectral flows in~\eqref{eq:expected_Sf_numerics}, and that the curves become steeper and steeper as $\nu$ increases. 

\medskip

Similar spectral flows have been observed in a dislocated model for wave propagation in hexagonal structures, see~\cite{DelFli-23}.

\begin{figure}
    \centering
    \includegraphics[width=0.9\textwidth]{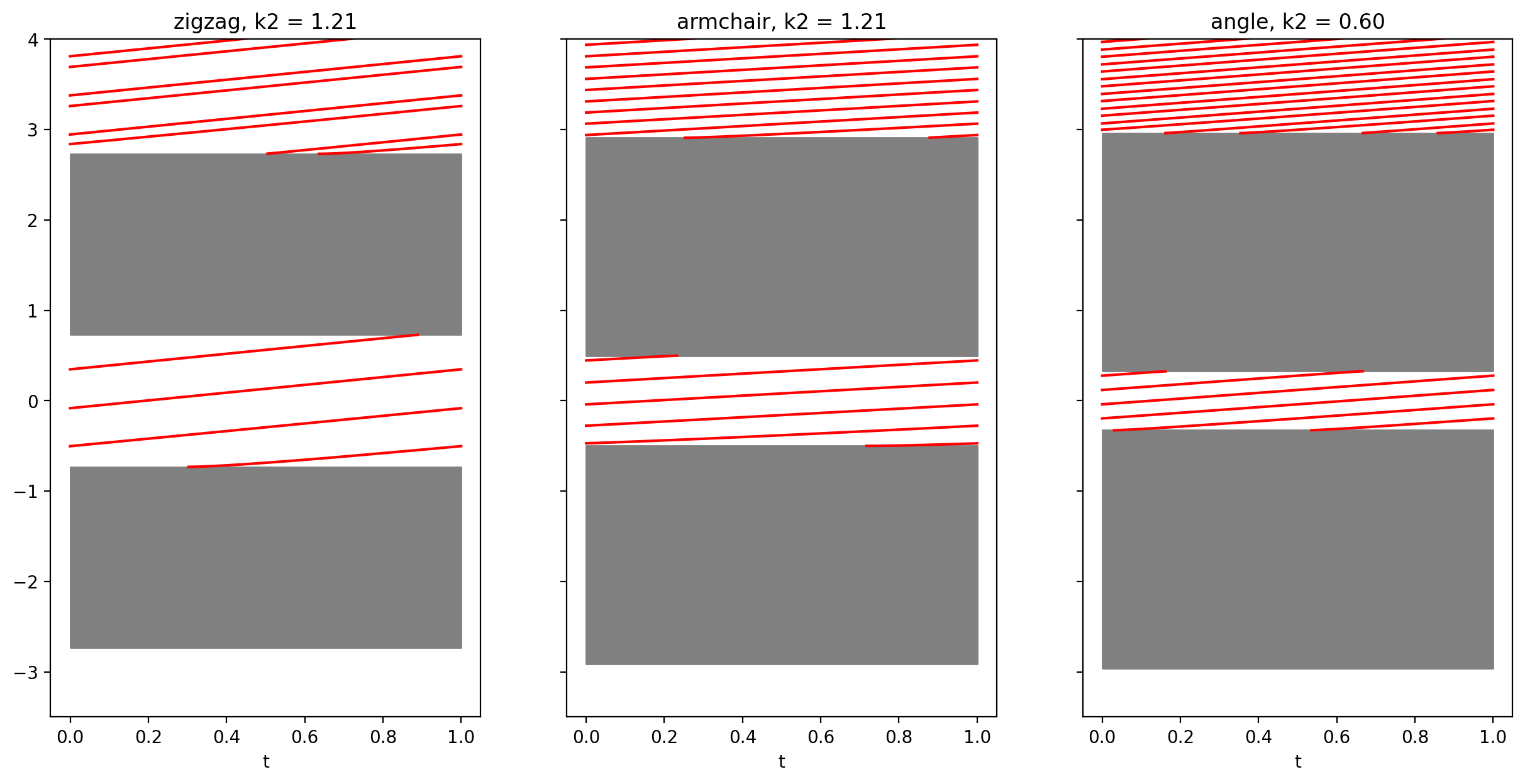}
    \includegraphics[width=0.9\textwidth]{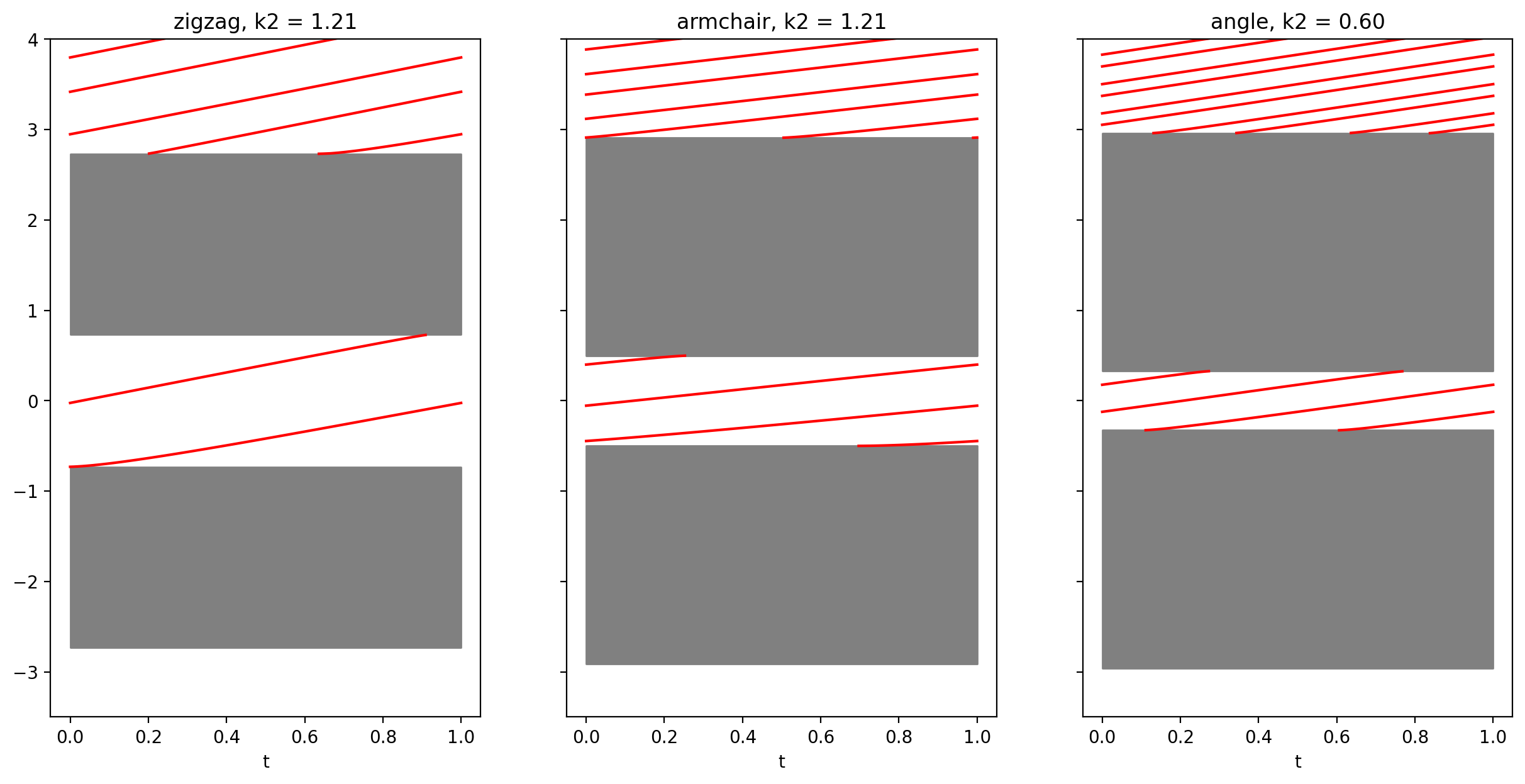}
    \includegraphics[width=0.9\textwidth]{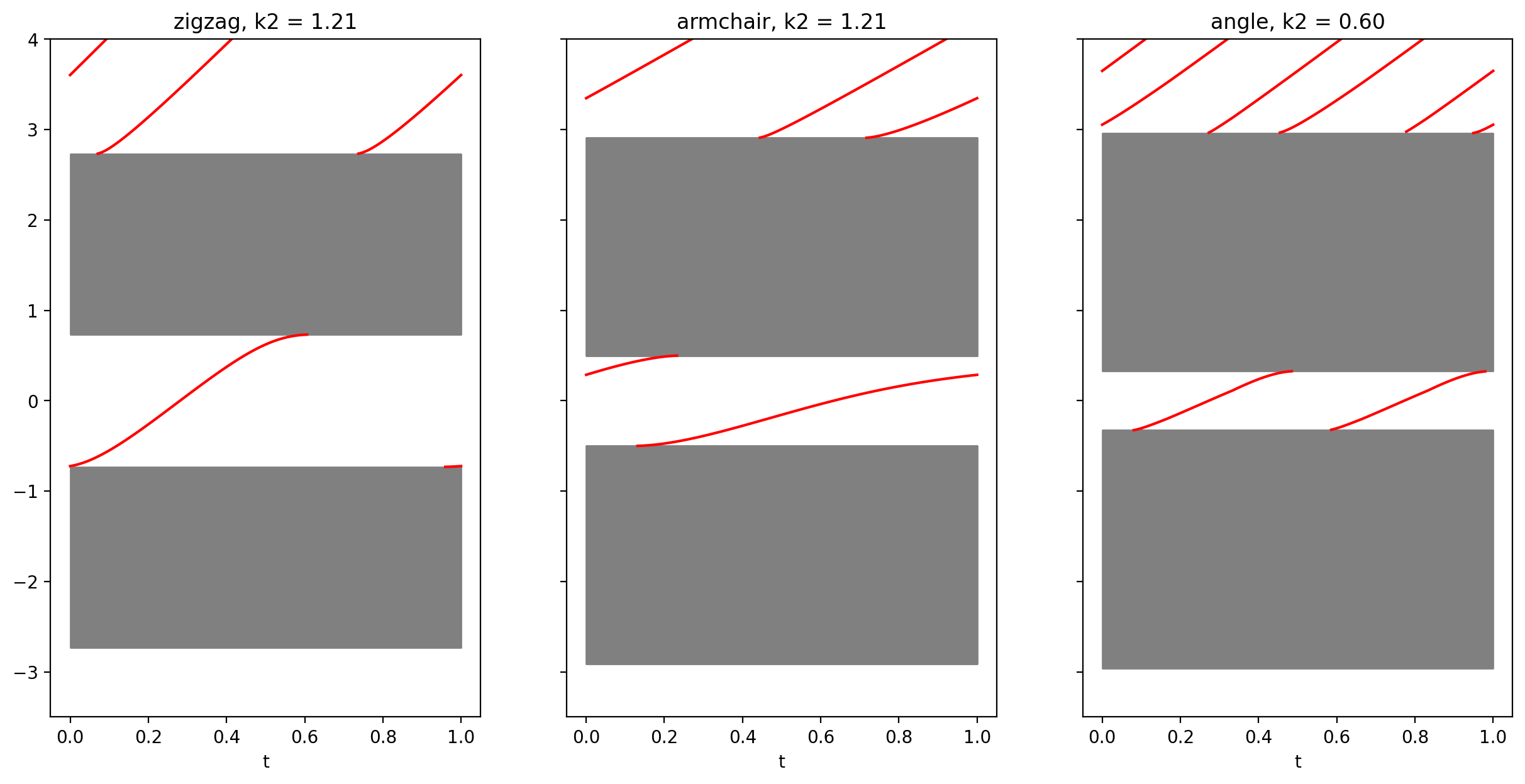}
    
    \caption{Spectrum of the edge Hamiltonians $t \mapsto H_{\bk_2}(t)$, $t \mapsto \widetilde{H}_{\widetilde{\bk}_2}(t)$ and $t \mapsto \overline{H}_{\overline{\bk}_2}(t)$ for different values of $\nu \in \{ 0.5, 1, 5\}$.}
    \label{fig:graphene_edge_spectrum_different_t}
\end{figure}


\appendix

\section{The Spectral Flow}
\label{sec:appendix:SF}

The notion of spectral flow has been introduced in~\cite{AtiPatSin-76, AtiSin-69}, and studied extensively in the case of continuous path of symmetric Fredholm operators (see also~\cite{Phi-96} and the recent book~\cite{DolSchWat-23} for a complete picture). Recall that for a bounded self-adjoint operator $A$, the operator $A - E$ is Fredholm iff $E$ is not in the essential spectrum of $A$. In the unbounded case, we can work instead with the family $B_t := f(A_t - E)$, where $f$ is a bounded increasing function with $f(\lambda) = \lambda$ around $\lambda = 0$. The goal of this Appendix is to give a self--contained version of the Spectral flow, using only tools from spectral theory, and which is enough for our purpose.

\subsection{Norm resolvent topology}

In what follows, we fix $\cH$ a separable Hilbert space. We define the norm-resolvent metric on the set of (possibly unbounded) self-adjoint operators $\cA(\cH)$ acting on $\cH$ by
\[
    d(A, B) := \left\| \cR(\ri, A) - \cR(\ri, B) \right\|_{\rm op},
\]
where we set $\cR(z, A) := (z - A)^{-1}$.
This endows $\cA(\cH)$ with the norm-resolvent topology. Recall that the point $\ri \in \C$ does not play any special role, and that, for all $z \in \C \setminus \R$, the corresponding metric with $z$ instead of $\ri$ defines the same topology. More specifically, we have the following.
\begin{lemma}
    Let $A$ be a self-adjoint operator, and let $z \in \C \setminus \sigma(A)$. Then there is $\varepsilon > 0$ and $C \ge 0$ so that, for all self-adjoint operators $B$ with $d(A, B) < \varepsilon$, we have $z \in \C \setminus \sigma(B)$, and $\| (z - B)^{-1} - (z - A)^{-1} \|_{\rm op} \le C d(A,B)$.
\end{lemma}

We refer to~\cite[Remark IV.3.13]{Kat-95} for the proof, see in particular~\cite[Eqn. IV (3.10)]{Kat-95}.

\begin{lemma} \label{lem:example_nrc}
    If $A$ is self-adjoint and $H$ is bounded self-adjoint, then  the map $\lambda \mapsto A + \lambda H$ is continuous for the norm-resolvent metric. In particular $\R \ni \lambda \mapsto A - \lambda$ is continuous.
\end{lemma}
\begin{proof}
    Let us prove continuity at $\lambda_0 \in \R$. We have
    \[
        \frac{1}{\ri - (A + \lambda H)} - \frac{1}{\ri - (A + \lambda_0 H)} = (\lambda - \lambda_0) \frac{1}{\ri - (A + \lambda H)} H \frac{1}{\ri - (A + \lambda_0 H)}.
    \]
    Since $A$ and $A + \lambda H$ are self-adjoint, we have $\| (\ri - A)^{-1} \|_{\rm op} \le 1$ (and similarly for $A + \lambda H$), so  $d(A+ \lambda_0 H, A + \lambda H) \le | \lambda - \lambda_0 | \cdot \| H \|_{\rm op}$, which implies continuity.
\end{proof}

\subsection{Definition of the Spectral Flow}
We consider a family $[0,1] \ni t \mapsto A_t$ of self-adjoint operators on $\cH$, continuous for this topology. We denote by $\sigma( \{ A_t \})$ and $\sigma_{\rm ess}( \{ A_t \})$ the spectrum and essential spectrum of the family $\{ A_t \}$ as the closure of the union (over $t$) of the spectra and essential spectra. Namely,
\[
    \sigma \left( \{ A_t \} \right) := \overline{ \bigcup_{t \in [0, 1]} \sigma(A_t) } \quad \text{and} \quad
    \sigma_{\rm ess} \left( \{ A_t \} \right) := \overline{ \bigcup_{t \in [0, 1]} \sigma_{\rm ess}(A_t) }.
\]
Both $\sigma( \{ A_t \})$ and $\sigma_{\rm ess}( \{ A_t \})$ are closed subsets of $\R$. Our goal is to define the spectral flow $\Sf(A_t, E, t \in [0, 1])$ for any $E \in \R \setminus \sigma_{\rm ess} \left( \{ A_t \} \right)$. It will correspond to the net number of branches of eigenvalues of $A_t$ crossing the energy $E$ downwards. 

\medskip

Recall that for any self-adjoint operator $A$, if $(E-g, E+g) \cap \sigma_{\ess}(A) = \emptyset$, then there is $0 \le a \le g$ so that both numbers $E - a$ and $E + a$ are not in the spectrum of $A$ (if $E \notin \sigma(A)$, one can simply take $a = 0$). The spectral projector $P_{(E - a, E + a)}(A)$ is then of finite rank, and given by the Cauchy integral
\[
    P_{(E - a, E + a)}(A) =  P_{[E - a, E + a]}(A) = \frac{1}{2 \ri \pi} \oint_{\sC} \frac{\rd z}{z - A},
\]
where $\sC$ is the positively oriented complex circle of center $E$ and radius $a$. We recall the following classical result. We skip the proof.
\begin{lemma} \label{lem:continuity_spectral_projection}
    With the same notation as before, there is $\varepsilon > 0$ so that the map $B \mapsto P_{(E - a, E + a) }(B)$ is continuous on 
   $\{ B \in \cA(\cH), \ d(A, B) < \varepsilon\}$. In particular, we have $\rank \,  P_{(E - a, E + a) }(B) = \rank \, P_{(E - a, E + a) }(A)$ for all such $B$.
\end{lemma}

We can now define  the spectral flow of $\{ A_t \}$ at the energy $E$. For all $t \in [0, 1]$, we let $a(t)\ge 0$ be so that $\{ E \pm a(t) \}$ are not in the spectrum of $A_t$. By Lemma~\ref{lem:continuity_spectral_projection}, there is $\varepsilon(t) > 0$ so that the result of the Lemma holds. Set
\[
    U(t) := \left\{ t' \in [0, 1], \quad d(A_t, A_{t'}) < \varepsilon(t) \right\}.
\]
By continuity of the map $t \mapsto A_t$, the set $U(t)$ is an open set of $[0, 1]$, and the union of all $U(t)$ covers the compact set $[0, 1]$. By classical arguments, we can find a subdivision $0 = t_0 < t_1 < \cdots < t_M = 1$ and real non-negative numbers $a_1, a_2, \cdots, a_M \ge 0$, which we call the {\em widths}, so that 
\begin{equation} \label{eq:condition_SF}
    \forall 1 \le i \le M, \quad [t_{i-1}, t_{i}] \ni t \mapsto P_{(E - a_i, E + a_i)}(A_t) 
    \quad \text{is continuous on $[t_{i-1}, t_{i}]$}.
\end{equation}
In other words, the spectrum of $A_t$ does not cross the levels $\{ E \pm a_i \}$ for $t \in [t_{i-1}, t_i]$. Notice in particular that $E + a_{i}$ and $E + a_{i+1}$ do not belong to the spectrum of $A_i$. The {\bf spectral flow} of $\{ A_t \}$  is defined by
\[
    \Sf(A_t, E, t \in [0, 1]) := \sum_{i=1}^{M}  \rank \,   P_{[E, E + a_i)} (A_{t_{i-1}})  -  \rank \,  P_{[E, E + a_i)} (A_{t_{i}})  .
\]
Note that $E$ may be an eigenvalue of $A_{t_i}$ here, so that $P_{[E, E + a_i]}$ is the spectral projection of $A_{t_i}$, {\em including} the eigenspace ${\rm Ker}(A_i - E)$. It is unclear with this definition that the spectral flow has nice topological properties. However, after relabelling the sums, we get the equivalent formula
\begin{equation} \label{eq:def:SF}
     \boxed{ \Sf(A_t, E, t \in [0, 1]) = \dim \Ker (A_{0} - E) - \dim \Ker(A_1 - E) + \sum_{i=1}^{M-1} \rank \, P_{(E+a_i, E + a_{i+1})} (A_{t_i}) ,}
\end{equation}
with the convention that $\rank \, {P}_{(a, b)}(A) = - \rank P_{(b,a)}(A)$ if $b < a$. Note that due to the boundary terms, the spectral flows depends on the value of $E$ in the gap. However, if we restrict the spectral flow to operators $A_t$ satisfying $\dim \Ker (A_{0} - E) = \dim \Ker(A_1 - E)$, then we will prove below that the spectral flow is independent of $E$ in the gap. This happens in particular if $A_1$ is unitarily equivalent to $A_0$.

\medskip

It is a classical result (see~\cite{Phi-96}) that the spectral flow is a well-defined quantity, which is independent of the choice of the $(t_i)$ and the $(a_i)$. The only condition is that the family $(t_i, a_i)$ satisfies~\eqref{eq:condition_SF}. The following example illustrates that the spectral flow counts the number of branches of eigenvalue going downwards.

\begin{example}
Assume that the spectrum of $A_t$ exhibits a single decreasing branch of eigenvalue $\lambda(t)$ crossing downwards the energy $E$ at $t_* \in (0, 1)$. For the subdivision, take $M = 3$, $t_0 = 1$, $t_1 = t_* - \varepsilon$, $t_2 = t_* + \varepsilon$ and $t_3 = 1$, and set $a_1 = 0$, $a_2 = a > 0$, $a_3 = 0$, where $a > \max \{ \lambda(t_1) - \lambda(t_*), \lambda(t_* - \lambda(t_2) )\}$. Then we get 
\[
    \rank \, P_{(E + a_1, E + a_2)} (A_{t_1}) = 1, \quad 
    \rank \, P_{(E + a_2, E + a_3)} (A_{t_2}) = 0,
\]
so the spectral flow equals $1$ in this case. See also Fig.~\ref{fig:spectralFlow_example}. 
\end{example}

\begin{figure}[ht]
    \centering
   \includegraphics[width = 0.7\textwidth]{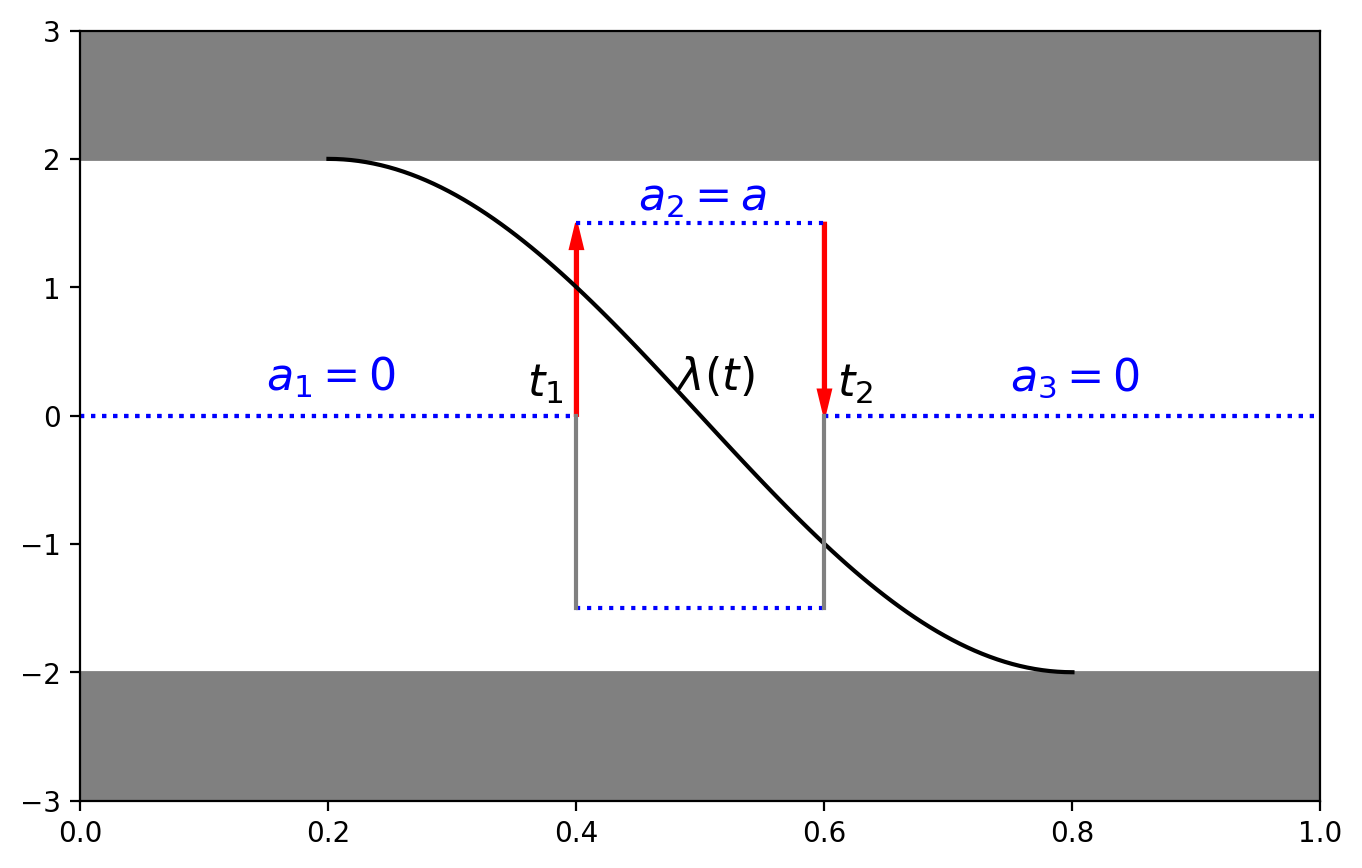}
    
    \caption{The subdivision $(t_i)$ and corresponding widths $(a_i)$. In this example, the spectral flow is $1$ at $E = 0$. The spectral flow is a signed sum  of branches of eigenvalues crossing the red horizontal segments (sign $+$ if the arrow is upwards, sign $-$ otherwise).}
    \label{fig:spectralFlow_example}
\end{figure}

\subsection{Stability of the spectral flow}

The Spectral flow is stable in the following sense.
\begin{lemma} \label{lem:stability_SF_ball}
Let $t \mapsto A_t$ be a norm-resolvent continuous family of self-adjoint operators, and $E \in \R\setminus \sigma_{\rm ess}(\{A_t\})$. There is $\varepsilon > 0$ so that, for any norm-resolvent continuous family $(B_t)$ of self-adjoint operators satisfying $d(A_t, B_t) < \varepsilon$ for all $t \in [0, 1]$ and $B_0 = A_0$,  $B_1 = A_1$, it holds $\Sf(A_t, E, [0, 1]) = \Sf(B_t, E, [0, 1])$.
\end{lemma}

We emphasize that we require the paths $(A_t)$ and $(B_t)$ to have the same endpoints here. This is due to the boundary terms in~\eqref{eq:def:SF}. The proof follows the exact same lines as in~\cite[Proposition 3]{Phi-96}, so we do not repeat it here.

\begin{lemma} \label{lem:stability_SF}
Let $(s,t) \mapsto (A_t^{(s)})$ be a norm-resolvent continuous family of self-adjoint operators, such that the endpoints $A_{0}^{(s)} = A_0$ and $A_{1}^{(s)} = A_1$ are independent of $s$. Assume in addition that $\Sf(A_t^{(s)}, E, t \in [0, 1])$ is well-defined for all $s$. Then this spectral flow is independent of $s$.
\end{lemma}

This result follows by taking a suitable subdivision in the homotopy parameter $s$ and applying the previous Lemma. The bottomline is that the only obstruction for homotopic invariance is that the spectral flow at $E$ becomes ill-defined, which happens when the essential gap closes. 

\medskip

Recall Weyl's theorem, which states that if $A$ is self-adjoint and $K$ is compact symmetric, then $A + s K$ is self-adjoint, and $\sigma_\ess(A + s K) = \sigma_\ess(A)$ for all $s$: the essential gaps are stable under compact perturbations, and therefore never close. Applying the previous lemma to $A_t^{(s)} := A_t + s K$ gives the following result.

\begin{lemma} \label{lem:stability_SF_compact}
Let $(K_t)$ be a continuous family of symmetric compact operators with $K_0 = K_1 = 0$. Then $\Sf(A_t, E, [0, 1]) = \Sf(A_t + K_t, E, [0, 1])$.
\end{lemma}
There is no assumption on the smallness of the family $(K_t)$ in the compact case. Finally, we record the following useful result.

\begin{lemma}
If $f : \R \to \R$ is a strictly increasing function, then
\[
    \Sf( f(A_t), f(E), [0, 1]) = \Sf( A_t, E, [0, 1]).
\]
\end{lemma}
Here, $f(A_t)$ is defined via the spectral calculus. We skip the proof, as it is straightforward. In the case where $(A_t)$ are uniformly bounded from below by some $C$, one can take $f_1(x) := (C - x)^{-1}$ and $f_2(x) = x - (C - E)^{-1}$, and deduce that
\[
    \Sf( A_t, E, [0, 1]) = \Sf \left(  (C - A_t)^{-1} , (C - E)^{-1}, [0, 1] \right)
    =  \Sf \left(  (C - A_t)^{-1} - (C - E)^{-1}, 0, [0, 1] \right).
\]
The advantage of the right-hand side is that the family $t \mapsto (C - A_t)^{-1} - (C - E)^{-1}$ is a norm continuous family of {\bf bounded} symmetric operators, for which $0$ is not in the essential spectrum. So this family is Fredholm self-adjoint, and we are back to the {\em usual} theory.

\subsection{The translation equivariant case}

In the previous section, the spectral flow depends on the energy $E$ in the gap. This is due to the presence of boundary terms in~\eqref{eq:def:SF}. 

\medskip

For $I \subset \R$ an interval of $\R$, we denote by $\cC_I$ the set of norm-resolvent continuous paths of self-adjoint operator $(A_t)$ satisfying that
\[
    \forall E \in I, \quad \dim \Ker (A_0 - E)  = \dim \Ker (A_1 - E).
\]
Typical examples are:
\begin{itemize}
    \item If $\R \ni t \mapsto (A_t)$ is $1$--periodic in $t$, that is $A_{t+1} = A_t$, then $(A_t)$ is in $\cC_I$ with $I = \R$.
    \item If $A_{1}$ is unitary equivalent to $A_0$, say $A_1 = U A_0 U^*$, then $(A_t)$ is in $\cC_I$ with $I = \R$.
    \item If $A_1 = \Sigma \oplus A_0$ (as in Section~\ref{sec:dislocated}), then $(A_t)$ is in $\cC_I$ with $I = (-\infty, \Sigma)$ or $I = (\Sigma, + \infty)$.
\end{itemize}

In this case, the boundary terms in~\eqref{eq:def:SF} cancels, and we obtain the simpler formula
\begin{equation*}
     \boxed{ \Sf(A_t, E, t \in [0, 1]) =  \sum_{i=1}^{M-1} \rank \, P_{(E+a_i, E + a_{i+1})} (A_{t_i}) } \quad \text{for} \quad A \in \cC_I, \quad \text{and} \quad E \in I.
\end{equation*}

Now, the results of the previous Lemma holds, without the endpoints constraint. Mimicking the proof of Lemma~\ref{lem:stability_SF_ball} gives the following.
\begin{lemma}
Let $A_t$ be a continuous operator family in $\cC_I$, and let $E \in I \setminus \sigma_{\ess}(\{ A_t \})$. There is $\varepsilon > 0$ so that, for any norm-resolvent continuous family $(B_t)$ of self-adjoint operators satisfying $d(A_t, B_t) < \varepsilon$ for all $t \in [0, 1]$ and such that $A_t + B_t \in \cC_I$, we have $\Sf(A_t, E, [0, 1]) = \Sf(B_t, E, [0, 1])$.
\end{lemma}

From this stability Lemma, we easily deduce the counterpart of Lemma~\ref{lem:stability_SF}.

\begin{lemma} \label{lem:stability_SF_periodic}
    Let $(s,t) \mapsto (A_t^{(s)})$ be a norm-resolvent continuous family of self-adjoint operators, such that, for all $s$, the family $t \mapsto A_t^{(s)}$ is in $\cC_I$. Let $E \in I$, and assume that $\Sf(A_t^{(s)}, E, t \in [0, 1])$ is well-defined for all $s$. Then this spectral flow is independent of $s$.
\end{lemma}

In particular, using Lemma~\ref{lem:example_nrc} stating that $\lambda \mapsto A - \lambda$ is continuous, we get the following.

\begin{lemma} \label{lem:SF_indepedent_of_E}
Let $(A_t) \in \cC_I$, and let $g \subset \R \setminus \sigma_\ess( \{ A_t \})$ be an essential gap of the family $(A_t)$, so that $g \subset I$. Then the spectral flow $\Sf(A_t, E, [0, A])$ is independent of $E$ in the gap $g$.
\end{lemma}

Finally, we state the counterpart of Lemma~\ref{lem:stability_SF_compact}.
\begin{lemma} \label{lem:stability_SF_compact_periodic}
Let $(K_t)$ be a continuous family of symmetric compact operators such that, for all $s \in [0, 1]$, we have $(A_t + s K_t) \in \cC_I$. Then, for $E \in I$, we have $\Sf(A_t, E, [0, 1]) = \Sf(A_t + K_t, E, [0, 1])$.
\end{lemma}

\section*{Acknowledgements}
The research leading to these results has received funding from CNRS, AAP IRL 2022, and from ANID, Chile, through the CMM Basal grant FB210005 and through Fondecyt Project 11220194. We thank the anonymous referees for the thorough reading of the manuscript and their suggestions that greatly improved the paper.

\section*{Statements}
All authors declare that they have no conflicts of interest. The Python code used to generate the figures is available at \url{https://gitlab.com/davidgontier/softwall_jacobimatrix}.

\bibliographystyle{siam}
\bibliography{biblio}

\end{document}